\documentclass[onecolumn,draftclsnofoot,12pt]{IEEEtran}
\usepackage[pdftex]{graphicx}
\usepackage{amsfonts}
\usepackage{amsmath}
\usepackage{amssymb,amsthm}
\usepackage{bbm}
\usepackage{float}
\usepackage{stfloats}
\usepackage{caption}
\usepackage{graphicx}
\usepackage{cite}
\usepackage{tikz}
\usepackage{xcolor}
\usepackage{verbatim}
\graphicspath{{},{../}}
\usepackage[colorlinks=true,linkcolor=black,anchorcolor=black,citecolor=black,filecolor=black,menucolor=black,runcolor=black,urlcolor=black]{hyperref}
\usepackage{multirow}
\usepackage{framed} 
\usepackage{booktabs}
\usepackage{enumitem}
\usepackage{wasysym} 
\usepackage{bbm}
\allowdisplaybreaks
\definecolor{armygreen}{rgb}{0.29, 0.33, 0.13}

\newcommand{\expU}[1]{e^{#1}}

\newcommand{\SINR}{\mathsf{SINR}}

\newcommand{\SThres}{\tau}
\newcommand{\Ball}{\mathcal{B}}
\newcommand{\expects}[2]{\mathbb{E}_{#1}\left[#2\right] }

\newcommand{\fracS}[2]{#1/#2}
\newcommand\expect[1]{\mathbb{E}\left[#1\right]}
\newcommand\prob[1]{\mathbb{P}\left[#1\right]}

\newcommand{\prim}{\mathrm{p}}
\newcommand{\seco}{\mathrm{s}}
\newcommand{\bloc}{\mathrm{b}}
\newcommand{\los}{\mathrm{L}}
\newcommand{\nlos}{\mathrm{N}}
\newcommand{\omni}{\mathrm{omni}}
\newcommand{\noise}{\sigma^2}
\newcommand{\gst}[1]{g_\mathrm{st}\left(#1\right)}
\newcommand{\gsr}[1]{g_\mathrm{sr}\left(#1\right)}
\newcommand{\gpt}[1]{g_\mathrm{pt}\left(#1\right)}
\newcommand{\gpr}[1]{g_\mathrm{pr}\left(#1\right)}
\newcommand{\dd}{\mathrm{d}}
\newcommand{\X}{\mathbf{X}}
\newcommand{\x}{\mathbf{x}}
\newcommand{\Y}{\mathbf{Y}}

\renewcommand{\ast}{a_\mathrm{st}}
\newcommand{\asr}{a_\mathrm{sr}}
\newcommand{\apt}{a_\mathrm{pt}}
\newcommand{\apr}{a_\mathrm{pr}}
\newcommand{\bst}{b_\mathrm{st}}
\newcommand{\bsr}{b_\mathrm{sr}}

\newcommand{\bpr}{b_\mathrm{pr}}
\newcommand{\qst}{q_\mathrm{st}}

\newcommand{\qpr}{q_\mathrm{pr}}

\newcommand{\tbr}[1]{}
\newcommand{\tbrb}[1]{}
\newcommand{\indside}[1]{\mathbbm{1}\left(#1\right)}
\newcommand{\subsectionA}[1]{\textbf{\textit{#1}}:}

\newtheorem{theorem}{Theorem}
\newtheorem{corr}{Corollary}
\newtheorem{lemma}{Lemma}
\newtheorem{remark}{Remark}
\newtheorem{example}{Example}

\IEEEoverridecommandlockouts
\begin{document}
\pagenumbering{gobble}  
\title{Spatially-aware Secondary License Sharing in mmWave Networks}
\author{Shuchi Tripathi, Abhishek K. Gupta \vspace{-.2in}
\thanks{The authors are with the Department of Electrical Engineering, Indian Institute of Technology Kanpur, India. Email: \{shuchi, gkrabhi\}@iitk.ac.in. A part of the paper was presented in \cite{tripathi2025can}.
}}
\maketitle

\begin{abstract}
In this work, we consider a multi-operator mmWave network implementing secondary license sharing (SLS) where a primary license holder leases secondary licenses to secondary users, allowing them to access its licensed spectrum under some pre-defined transmission constraints. The highly directional nature of mmWaves, along with their sensitivity to blockages, naturally confines the interference to/from devices to narrow angular sectors within a certain range around themselves. This motivates us to consider a spatially-aware SLS that determines a secondary link's activity based on its distance/orientation relative to the primary link, as well as blockages around it. By leveraging the tools of stochastic geometry, we develop an analytical framework to design and study such spatially-aware SLS in mmWave networks. Our analysis quantifies the transmission opportunities available to secondary users and the resulting coverage probabilities for both primary and secondary links. We characterize the effect of directionality and blockage conditions, along with transmission restrictions and secondary users' density, on the performance of both operators. Via numerical investigation, we derive various insights. We show that blockage conditions can change the shape of coverage plots and thus affect key conclusions. Further, blockage and directionality can increase the transmission opportunities for secondary users, improving the feasibility and gains of SLS.
\end{abstract}
\section{Introduction}\label{Section:Blockage:SinglePrimary:Introduction}
Owing to high directionality and susceptibility to blockages, mmWave networks exhibit a reduced level of interference \cite{akdeniz2014millimeter, petrov2017interference,tripathi2021millimeter}. As a result, exclusive licensing where operators own exclusive spectrum licenses may lead to severe spectrum underutilization \cite{saleem2014primary, tehrani2016licensed}. Uncoordinated license sharing among operators can increase their coverage rates while reducing license costs without additional overhead, resulting in cheaper services for end users \cite{tehrani2016licensed, jurdi2018modeling, gupta2020spectrum}. This also minimizes entry barriers for new or small-scale operators, which aligns with recent efforts of spectrum regulatory authorities ({\em e.g.} the Federal Communications Commission (FCC) \cite{fccnoi14,fccnoi15}) towards making mmWave spectrum sharing more flexible. Particularly, this enables secondary license sharing (SLS), where primary operators lease secondary licenses for their unused or underutilized licensed spectrum to secondary users under certain restrictions that ensure minimal impact on their quality of service (QoS) \cite{gupta2016gains, gupta2016restricted, bhattacharjee2022cognitive}. It is crucial to understand how to design such restrictions to fully realize the potential of SLS, considering the propagation characteristics of mmWave bands.

\subsectionA{Related Work}\label{Section:Blockage:SinglePrimary:Related Work}
Utilizing stochastic geometry (SG) based models \cite{andrews2016modeling}, {\em e.g.}, the Boolean blockage models and Poisson point process (PPP)  deployments \cite{bai2014analysis,bai2014coverage}, prior works have analyzed the performance of mmWave networks, highlighting the role of directionality and blockages in reducing interference. Hence, spectrum sharing is seen as a natural feature of mmWave networks. In our prior work, we have shown the feasibility of a license-sharing framework in mmWave bands where operators, who own individual licenses to their spectrum, allow each other to share their spectrum without any coordination \cite{gupta2016feasibility}. While such an uncoordinated sharing mechanism can provide some gains in the overall sum rate, it can also degrade the rate of edge users due to a lack of any protection. To bridge this gap, secondary license sharing (SLS) models were proposed, where a licensed operator (termed primary operator) can lease secondary licenses to other operators, allowing them to access its spectrum under specific (pre-defined) constraints. SLS aims to provide reasonable channel access opportunities to secondary devices without compromising the QoS of primary users, thereby increasing effective spectrum utilization. It also provides sufficient flexibility to the spectrum license markets, which have been marked as essential by various regulatory authorities across the globe \cite{fccnoi14,fccnoi15}. To implement these restrictions, secondary operators may require various mechanisms, including cognitive channel sensing (CCS) and access to third-party databases. Hence, it is desirable to implement light and static constraints to avoid heavy overheads that can diminish the gain of such sharing. For example, secondary licenses can introduce an interference-threshold transmit constraint, requiring all secondary transmitters to ensure that their interference at the primary receiver does not exceed the pre-defined threshold.

Authors in \cite{nguyen12} have analyzed the coverage performance of a traditional cognitive wireless network scenario consisting of primary and secondary devices with omnidirectional communication in the absence of blockages. The work derived closed-form expressions for the primary and secondary coverage probabilities under an interference-threshold transmit constraint. However, the work did not include the effects of directional antennas and blockages, which are inherent to mmWave communication. In \cite{gupta2016gains,gupta2016restricted}, we presented an SLS model for mmWave networks under an interference-threshold transmit constraint. The work considered omnidirectional transmission to compute the interference level, resulting in a distance-based restriction which forms a radial protection zone around each primary user. Similarly, the work in \cite{wang2018underlay} studied the performance of cognitive mmWave networks with directional antennas, did not account for the impact of the relative orientation of primary and secondary links on secondary transmission opportunities while computing the secondary coverage. The work failed to capture the correlation between the transmission activity and the interference of a secondary link, caused by its relative orientation, and instead averaged over realized antenna gain, resulting in an effective distance-based restriction. The highly directional nature of mmWaves naturally confines the interference to/from devices to narrow angular sectors \cite{bhattarai2016overview}. In our previous work \cite{tripathi2021coverage, tripathi2024coverage}, we showed that SLS under the interference threshold restriction results in an asymmetric protection zone around the primary receiver. This highlights the importance of incorporating the orientation of secondary users in determining their transmit activities. However, the work did not consider the presence of blockages, an essential feature of mmWave communication systems.

The presence of blockages results in a probabilistic dual slop path-loss \cite{gupta2015sinr} that can significantly change the general system behaviour while introducing new trends {\em e.g.,} loss of SINR invariance, and non-monotonicity of coverage with respect to network density \cite{gupta2020fundamentals}. Due to differences in path-loss, the interference threshold's effect on a secondary transmitter with a line-of-sight (LOS) link to the primary receiver is significantly different from that with a non-LOS (NLOS) primary receiver, causing fundamentally different trade-offs between transmission activity and interference for the two types of secondary transmitters. This adds complexity to the design of the optimal transmit restriction in such systems. Moreover, the presence of blockages can further confine interference from/to a device around itself, thereby creating additional transmission opportunities for secondary users. Therefore, to achieve the full potential of SLS, it is essential to incorporate awareness of their neighbourhoods in determining the transmission opportunities of secondary users. This work aims to bridge this gap by introducing spatially-aware SLS, where a secondary user’s access depends on its interference at the primary receiver, which includes the impact of antenna directionality, the relative distance/orientation of the primary link, as well as the blockage states of corresponding links. Therefore, studying the performance of such spatially-aware SLS in the presence of blockages and directional communication is the main focus of this work.

\subsectionA{Contributions} \label{Section:Blockage:SinglePrimary:Contributions}
In the paper, we develop an SG-based analytical framework for a mmWave network of a primary operator implementing SLS, which allows it to lease its licensed spectrum to a secondary operator via secondary licenses that restrict the secondary operator to follow a pre-specified interference-threshold criterion. In particular, a secondary transmitter is allowed to transmit only if its interference at any of the primary receivers is below a certain threshold. Note that the secondary interference depends on the directivity gain and blockage state of the link. Hence, we aim to design a spatially-aware SLS that determines the secondary link's activity based on its antenna directionality, relative distance/orientation of the primary link, as well as the blockage states of their corresponding links to maximize the secondary transmission activity as much as possible without affecting primary QoS. The work aims to understand the benefits that spatially-aware SLS can provide and understand its impact on the coverage in the presence of directionality and blockages. The work has the following key contributions.  
\begin{itemize}[leftmargin=6pt]
\item We develop an analytical framework to study the impact of blockages and directionality on the performance of spatially-aware SLS in mmWave networks. We derive the medium access probability (MAP) of each secondary transmitter, followed by the activity factor (AF) of the secondary network, to quantify the transmission opportunities for secondary users. We have shown that blockages improve the secondary transmission opportunities. 	
\item We further derive the coverage probability of primary and secondary links, highlighting the roles of various system parameters. Further, we observe that the presence of moderate blockages can improve the coverage of both primary and secondary links, especially when directionality is present. We also provide simplified coverage expressions for a few special cases.
\item Using derived results, we present various design insights. We show that the presence of blockages can change the fundamental shape of coverage vs threshold curves, along with key conclusions, and hence alters how various system parameters affect the performance. In particular, when both LOS and NLOS links are present, the observed coverage curves are significantly different from the case when only one type of link is present. This may be due to the fact that the interference threshold's effect on a secondary transmitter with a LOS link to the primary receiver is significantly different than that with a NLOS primary receiver, causing fundamentally different trade-offs between transmission activity and interference for the two types of secondary transmitters. We observe that whether blockages improve or degrade the secondary performance depends on the level of interference restriction put on SLS and the availability of the LOS serving link. We show that the interference threshold is crucial for network performance, and hence it must be selected carefully while accounting for directionality and blockage parameters. 
\end{itemize}

\noindent\textbf{Notation}: \label{Section:Blockage:SinglePrimary:Notation}
For a location $\x$,  $x=||\x||$ and  $\angle \x$ denotes the angle of vector $\x$ with respect to the X-axis. $\Ball(\x,r)$ denotes a ball with radius $r$ and center $\x$. $x\angle \theta$ denotes the location with radial distance $x$ and angle $\theta$. 
$ \mathrm{p}, \,\mathrm{s}, \,\mathrm{t}$ and $ \mathrm{r} $ stand for primary, secondary, transmitter, and receiver, respectively. $\Gamma(a, b) = \int_{0}^{b} t^{a - 1} e^{-t}\, \mathrm{d}t $ denotes the lower incomplete Gamma functions.

\section{System Model} \label{Section:Blockage:SinglePrimary:SystemModel}
In this paper, we study a multi-operator mmWave network implementing SLS where the primary operator (the one owning a license for a particular band) leases this licensed band to secondary operators via secondary licenses. 

\subsectionA{Network model} \label{Section:Blockage:SinglePrimary:SystemModel:NetworkModel}
We consider a single primary link with $\X_{\prim\mathrm{0}}$ and $\Y_{\prim\mathrm{0}}$ denoting the locations of the primary transmitter and receiver, respectively. Without losing generality, we assume that the primary link, denoted by $\X_{\prim\mathrm{0}}-\Y_{\prim\mathrm{0}}$, with link-length $r_\prim$ is aligned with the X-axis such that the primary receiver location is fixed at the origin \textit{i.e.} $\Y_{\prim\mathrm{0}} = \mathbf{o}$ (see Fig. \ref{fig:Blockage:SinglePrimary:System-Model}(a)). The network consists of many secondary links with their transmitters distributed as a 2D-homogeneous PPP $\Phi_{\mathrm{s}}$ with density $\lambda_{\mathrm{s}}$ \cite{andrews2016modeling, andrews2023introduction}. Here, $\Phi_{\mathrm{s}} = \{ \X_{\mathrm{s}i} \}$, where $\X_{\mathrm{s}i} = x_{\seco i} \angle{\theta_{\seco i}}$ is the location of $i^\mathrm{th}$ secondary transmitter with its corresponding receiver is located at $\Y_{\seco i}$, having a link-length of $r_{\mathrm{s}i}$. Here, $\theta_{\mathrm{s}i}$ and $\omega_{\mathrm{s}i}$ represent the angular direction of the $i^\mathrm{th}$ secondary transmitter and receiver in the anti-clockwise direction from the X-axis, respectively (see Fig. \ref{fig:Blockage:SinglePrimary:System-Model}(a)). Depending on the specific application, $\omega_{\mathrm{s}i}$ and $r_{\mathrm{s}i}$ can either be fixed or treated as random variables (RV) with some distribution \textit{e.g.} uniform. Thus, the marked PPP $\left\{\X_{\mathrm{s}i}, \, \left(\omega_{\mathrm{s}i}, \, r_{\mathrm{s}i} \right) \right\}$ provides the complete information about the locations of secondary transmitter and receiver pairs. We further consider that each secondary receiver is oriented differently with the fixed distance from its assigned transmitters. Alternatively, $\omega_{\mathrm{s}i}$ is a RV uniformly distributed between 0 and $2\pi$, whereas the values of $r_{\mathrm{s}i}$ are fixed {\em i.e} $r_{\mathrm{s}i} = r_\seco$. However, the analysis that is being provided can be trivially extended to the general distribution of $r_{\mathrm{s}i}$ and $\omega_{\mathrm{s}i}$. Let $p_\prim$ and $p_\seco$ be the primary and secondary transmit power.
\begin{figure}[t!]
\centering
{\includegraphics[trim = 5 5 5 5, clip, scale=0.15]{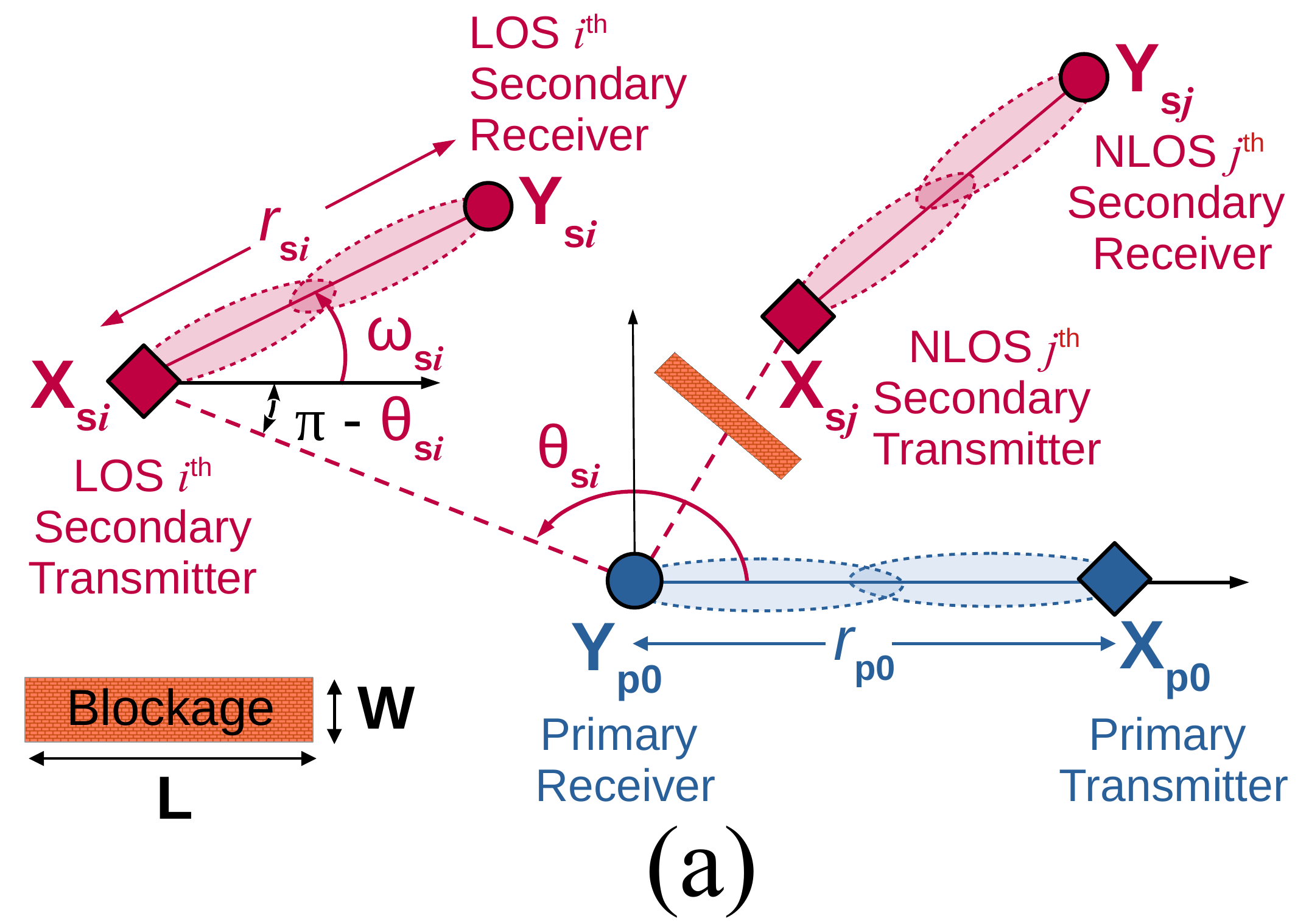} \label{fig:Blockage:SinglePrimary:SimpleCase} }
\hspace*{-0.1in}
{\includegraphics[trim = 5 5 5 5, clip, scale=0.15]{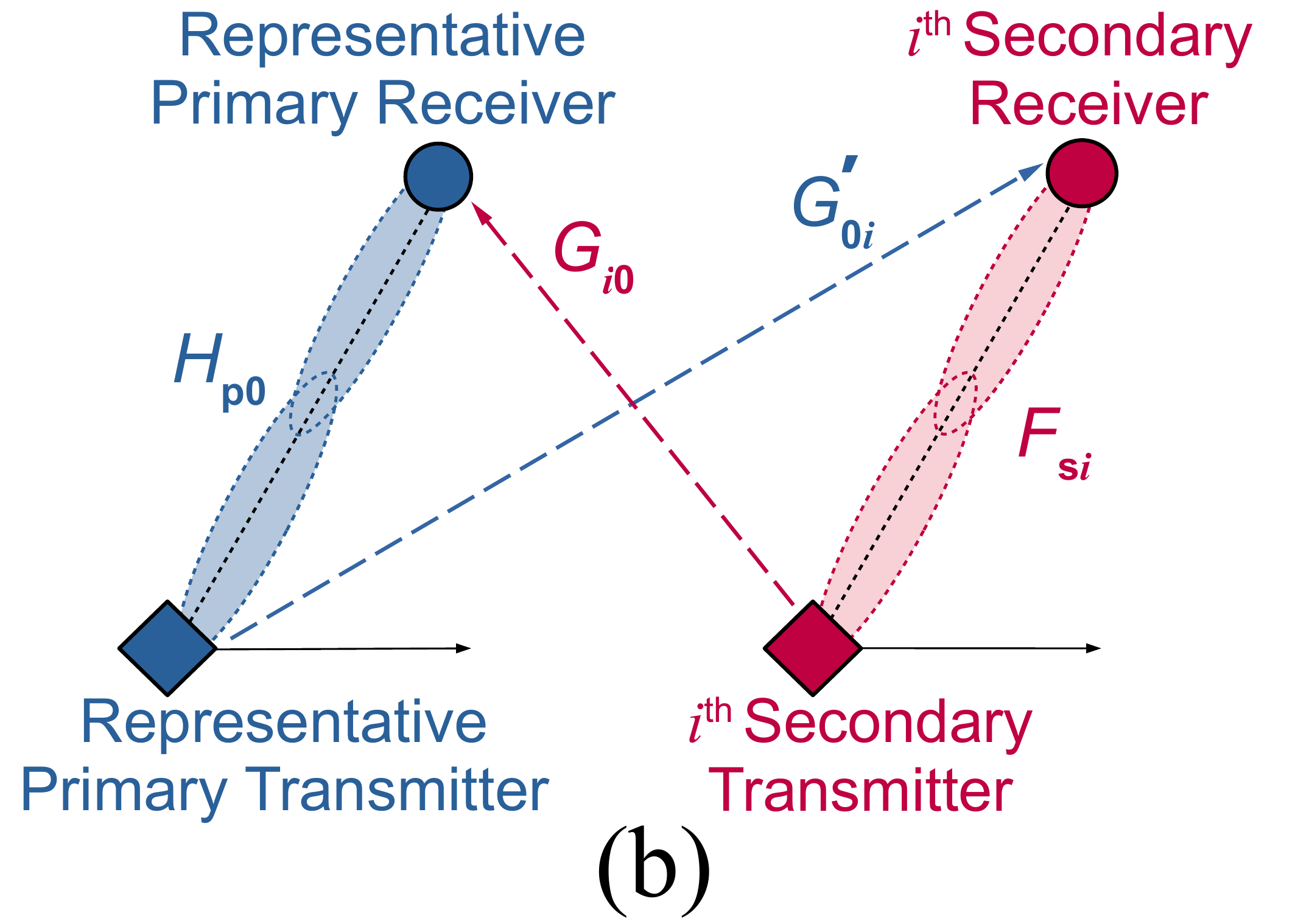}\label{fig:Blockage:SinglePrimary:Fading}}
\caption{An illustration of the system model for analysis of (a) the primary coverage network and (b) the channel fading model.}
\label{fig:Blockage:SinglePrimary:System-Model}
\vspace{-.1in}
\end{figure}

\subsectionA{Fading} \label{Section:Blockage:SinglePrimary:SystemModel:Fading}
The co-link fading coefficients for the primary link and the $i^{\mathrm{th}}$ secondary link are denoted by $H_{\prim\mathrm{0}}$ and $F_{\mathrm{s}i}$, respectively while the cross-link fading coefficients for the $i^{\mathrm{th}}$-secondary-transmitter-to-primary-receiver and the primary-transmitter-to-$i^{\mathrm{th}}$-secondary-receiver are represented by $G_{i\mathrm{0}}$ and $G^{'}_{\mathrm{0}i}$, respectively (see Fig. \ref{fig:Blockage:SinglePrimary:System-Model}(b)). $F^{'}_{i\mathrm{0}}$ denotes the cross-link fading coefficient for the $i^{\mathrm{th}}$-secondary-transmitter-to-$\mathrm{0}^{\mathrm{th}}$-secondary-receiver. We consider Rayleigh fading for all links, primarily to maintain analytical tractability, rather than obscure insights by increasing complexity. Past studies have shown that fading assumptions do not affect insights. Nonetheless, the results obtained in this study may be readily modified to take into account other fading distributions \cite{gupta2020does, ghatak2021deploy}.

\subsectionA{Directional communication} \label{Section:Blockage:SinglePrimary:SystemModel:DirectionalCommunication}
The antenna beam pattern of the device $v$ of type $u$ is denoted by $g_{uv}(\theta)$ with each device equipped with $M_{uv}$ antenna elements. Here, $u \in \{ \mathrm{p}, \mathrm{s} \} $, $ v \in \{ \mathrm{t}, \mathrm{r} \} $, and $\theta \in [-\pi, \pi)$ represents the angle with respect to the beam orientation. Thus, all devices have directional antennas. Further, we will occasionally consider the following two special beam patterns to simplify the results.
\subsubsection{Sectorized beam pattern} 
\label{Section:Blockage:SinglePrimary:SystemModel:DirectionalCommunication:SectorBeam}
The beam pattern of a directional antenna system can be well approximated by a sectorized beam pattern \cite{andrews2016modeling, andrews2023introduction}, defined as 
\begin{align}
g_{uv}(\theta) = \begin{cases}
a_{uv} \qquad \text{if}\ \lvert \theta \rvert \leq \phi_{uv}/2\\
b_{uv} \qquad \text{if}\ \lvert \theta \rvert > \phi_{uv}/2,
\end{cases} 
\label{eq:GainApproximation}
\end{align} 
where $\phi_{uv} $ is the beamwidth, $a_{uv}$ is the main lobe gain and $b_{uv}$ is the side lobe gain. Note that $\phi_{uv}$ is a reciprocal function of $M_{uv}$. For an uniform linear array (ULA), $a_{uv} = M_{uv}$ and $\phi_{uv} = \kappa/M_{uv}$ for some constant $\kappa$. The probability of having $a_{uv}$ gain at a random direction is $q_{uv} = {\phi_{uv}}/{2\pi} = \kappa/(2\pi M_{uv})$ with $a_{uv}q_{uv} + b_{uv}(1 - q_{uv}) = 1$ which ensure that the radiated power is normalized. Let $\kappa'= \kappa/2\pi$ then $b_{uv} = (1 - \kappa')/(1 - \kappa'/M_{uv})$.
\subsubsection{Ideal beam pattern} 
\label{Section:Blockage:SinglePrimary:SystemModel:DirectionalCommunication:IdealBeam}
A special case of the sectorized beam pattern with $ b_{uv} = 0 $. We will use this approximation occasionally to draw direct insights about the advantages of spatially-aware SLS over omni-SLS (defined in the special cases of Section \ref{Section:Blockage:SinglePrimary:SystemModel}). 

\subsectionA{Blockage model} \label{Section:Blockage:SinglePrimary:SystemModel:BlockageModel}
Similar to \cite{gupta2016feasibility,gupta2016gains}, we assume a blocking model where the central points of all rectangular blockages having an average length $L$ and an average width $W$ are distributed according to a homogeneous PPP $\Phi_\bloc$ with density $\lambda_\bloc$. Thus, the number of blockages crossing the $\mathrm{path}:\X_{ui}\to\Y_{\tilde{u}j}$ is a Poisson RV  with mean $\mu z_{ij} + p$ where $z_{ij} = \|\X_{ui} - \Y_{\tilde{u}j}\|$, $\mu = 2 \lambda_\bloc (\mathbb{E}[W] + \mathbb{E}[L]) / \pi$ and $p = \lambda_\bloc \mathbb{E}[W] \mathbb{E}[L]$ \cite{bai2014analysis,bai2014coverage}. Note that $\mu$ and $p$ are blockage parameters  \cite{bai2014analysis,bai2014coverage}. For simplicity, we assume that any device of type $u,\tilde{u} \in \{\mathrm{p}, \, \mathrm{s}\}$, the $\mathrm{path}:\X_{ui}\to\Y_{\tilde{u}j}$ is blocked independently. The LOS probability of the $\mathrm{path}:\X_{ui}\to\Y_{\tilde{u}j}$ is given as $p_{\los} (z) = e^{-(\mu z + p)}$, where $z$ is link distance. Let us define the blockage state $T_{ui \to \tilde{u}j}=\indside{\mathrm{path}:\X_{ui}\to\Y_{\tilde{u}j} \text{ is }  \mathrm{LOS}}$ where $\indside{\cdot}$ is an indicator function. Clearly, $T_{ui \to \tilde{u}j}$ is a Bernoulli RV which takes two values $1$ and $0$ with probability $p_{\los} (z)$ and $1 - p_{\los} (z)$, respectively. For a path of length $z$ and blockage state $T_{ui \to \tilde{u}j} $, its path-loss is given as
\begin{align}
\ell_{T_{ui \to \tilde{u}j} } (z) = \begin{cases}
\ell_1 (z) = C_\los z^{-\alpha_\los}  &\text{when } T_{ui \to \tilde{u}j}  = 1 \\
\ell_0 (z) = C_\nlos z^{-\alpha_\nlos} &\text{when } T_{ui \to \tilde{u}j}  = 0
\end{cases},
\label{eq:Blockage:SinglePrimary:PathLoss}
\end{align}
where $\alpha_T$ and $C_T$ with $T = \{\los,\nlos\}$ denote the path-loss exponent and near-field constant, respectively.

\subsectionA{Spatially-aware SLS} \label{Section:Blockage:SinglePrimary:SystemModel:CognitiveCommunication}
The SLS imposes a spatially-aware restriction allowing a secondary device to transmit only if its interference at the primary receiver, given its directional gain and blockage state, is less than $\rho$. $\rho$ can be seen as the primary-transmit-protection-threshold and can usually be pre-specified in the secondary licenses. The received power at the primary receiver due to the $i^{\mathrm{th}}$ secondary transmitter is
\begin{align}
P_{T_{\seco i \to \prim0}} &= p_\seco \, \kappa_{i\mathrm{0}} \, G_{i\mathrm{0}} \, \ell_{T_{\seco i \to \prim\mathrm{0}}} (x_{\seco i}),
\label{eq:Blockage:SinglePrimary:RxPower}
\end{align}
where $\kappa_{i\mathrm{0}} = g_{\mathrm{pr}} ( \theta_{\mathrm{s}i}) g_{\mathrm{st}} ( \theta_{\mathrm{s}i} - \pi - \omega_{\mathrm{s}i})$. Here, $g_{\mathrm{pr}}$ and $ g_{\mathrm{st}}$ are the antenna patterns of the primary receiver and secondary transmitter, respectively. Let $P^\los_{\seco i \to \prim0}$ and $P^\nlos_{\seco i \to \prim0}$ denote the received power of $P_{T_{\seco i \to \prim0}}$ when $T_{\seco i \to \prim0}$ is 1 and 0, respectively. 
Now, the $i^{\mathrm{th}}$ secondary transmitter at location $\X_{\mathrm{s}i}$ will transmit only if $P_{T_{\seco i \to \prim0}}$ is less than $\rho$. Let us denote this event by $E_{i}$ \textit{i.e.} $E_{i} = \left \{ P_{T_{\seco i \to \prim0}} < \rho \right \} = \left \{p_\seco \kappa_{i\mathrm{0}} G_{i\mathrm{0}} \ell_{T_{\seco i \to \prim\mathrm{0}}} (x_{\seco i})  < \rho \right \}$. If $U_{i}$ be an indicator for the occurrence of event $E_{i}$, then
\begin{align}
&U_{i} = \mathbbm{1} \left (p_\seco \, \kappa_{i\mathrm{0}} \, G_{i\mathrm{0}} \, \ell_{T_{\seco i \to \prim\mathrm{0}}} (x_{\seco i})  < \rho \right ). 
\label{eq:Blockage:SinglePrimary:Indicator}
\end{align}

The $\mathrm{path}:\X_{\seco i}\to\Y_{\prim0}$ can either be LOS or NLOS, therefore the indicator $U_{i}$ in \eqref{eq:Blockage:SinglePrimary:Indicator} can also be written as 
\begin{align*}
U_{i} = \begin{cases}
U^\los_{i} = \mathbbm{1} \big(P^\los_{\seco i \to \prim\mathrm{0}} < \rho \big) &\text{when } T_{\seco i \to \prim0} = 1 \\
U^\nlos_{i} = \mathbbm{1} \big(P^\nlos_{\seco i \to \prim\mathrm{0}} < \rho \big) &\text{when } T_{\seco i \to \prim0} = 0
\end{cases}. 
\end{align*}

\subsectionA{Special cases} 
\label{Section:Blockage:SinglePrimary:SystemModel:BlockageModel-SpecialCases}
We will occasionally consider the following special cases to simplify results.
\begin{enumerate}[leftmargin=10pt]
\item{\em The zero blockage loss (ZBL) scenario} 
\label{Section:Blockage:SinglePrimary:SystemModel:BlockageModel:ZBL}
refers to an absence of any blockage resulting a LOS-only scenario with $p_{\los} (z) = 1$ and $P_{T_{\mathrm{s}i \to \mathrm{p0}}} \! = \! P^\los_{\mathrm{s}i \to \mathrm{p0}}$ for all links .
\item{\em The high blockage loss (HBL) scenario} 
\label{Section:Blockage:SinglePrimary:SystemModel:BlockageModel:HBL}
refers to the case where all links are blocked resulting a NLOS-only scenario with $p_{\los} (z) = 0$ and $P_{T_{\mathrm{s}i \to \mathrm{p0}}} \! = \! P^\nlos_{\mathrm{s}i \to \mathrm{p0}}$ for all links.
\item{\em The NLOS in outage (NOL) scenario} 
\label{Section:Blockage:SinglePrimary:SystemModel:BlockageModel:NOL}
refers to case where NLOS links are in complete outage {\em i.e.,} $C_\nlos = 0$. Hence, $U_{i} = U^\los_{i}$ when $T_{\seco i \to \prim0} = 1$ and $U_{i} = 1$ when $T_{\seco i \to \prim0} = 0$. 
\item {\em Omni-SLS} 
refers to the use of omni-directional antennas with $g_{uv}(\theta) = 1,\,\forall \,\theta$.  Here, $P_{T_{\seco i \to \prim\mathrm{0}}}$ and $U_{i}$ are given as \eqref{eq:Blockage:SinglePrimary:RxPower} and \eqref{eq:Blockage:SinglePrimary:Indicator}, respectively with $\kappa_{i\mathrm{0}} = 1$. 
\end{enumerate}

\section{Analysis} \label{Section:Blockage:SinglePrimary:Analysis}
In this section, we analyze the effect of directionality and blockages on the transmission opportunities for secondary users and the coverage of primary/secondary links to understand benefits that spatially-aware SLS can provide.

\subsection{Medium access probability (MAP)} \label{Section:Blockage:SinglePrimary:Analysis:MAP}
For the $i^\mathrm{th}$ secondary transmitter at location $\X_{\mathrm{s}i}$, its MAP is defined as the probability that it can access the channel {\em i.e.} $p_{\mathrm{m}i} = \mathbb{P} (U_i = 1)$. Using \eqref{eq:Blockage:SinglePrimary:Indicator}, we can write $p_{\mathrm{m}i} = \mathbb{P} (G_{i\mathrm{0}} < (\rho/p_\seco)/$ $(\kappa_{i\mathrm{0}} \ell_{T_{\seco i \to \prim\mathrm{0}}} (x_{\seco i})))$. Using $G_{i\mathrm{0}} \sim \exp(1)$ and substituting $\ell_{T_{\seco i \to \prim\mathrm{0}}}$ from \eqref{eq:Blockage:SinglePrimary:PathLoss} will give its MAP as 
\begin{align*}
p_{\mathrm{m}i} = \begin{cases}
1 - p_\los (x_{\seco i}) e^{ - \left({\kappa_\los}/{\kappa_{i\mathrm{0}}}\right) x_{\seco i}^{\alpha_\los}} \!\! &\text{when } T_{\seco i \to \prim0} = 1 \\
1 - p_\nlos (x_{\seco i}) e^{ - \left({\kappa_\los}/{\kappa_{i\mathrm{0}}}\right) x_{\seco i}^{\alpha_\los}} \!\! &\text{when } T_{\seco i \to \prim0} = 0
\end{cases}. 
\end{align*}

Averaging over the blockages gives the MAP as stated in the following lemma.

\begin{lemma} \label{Lemma:Blockage:SinglePrimary:MAP}
The MAP of the $i^{\mathrm{th}}$ secondary transmitter at location $\X_{\mathrm{s}i} = x_{\mathrm{s}i} \angle{ \theta_{\mathrm{s}i} }$ is given as
\begin{align}
p_{\mathrm{m}i} = 1 - \sum\nolimits_{T = \{\los,\nlos\}} p_T (x_{\seco i}) \exp \left( - \left({\kappa_T}/{\kappa_{i\mathrm{0}}}\right) x_{\seco i}^{\alpha_T} \right), 
\label{eq:Blockage:SinglePrimary:MAP}
\end{align}
where $p_\los (x_{\seco i}) = \exp(-\mu x_{\seco i} - p)$, $p_\nlos (x_{\seco i}) = 1 - p_\los (x_{\seco i})$, $\kappa_T = (\rho/p_\seco)/(C_T)$ and $\kappa_{i\mathrm{0}} = g_{\mathrm{pr}} ( \theta_{\mathrm{s}i}) g_{\mathrm{st}} ( \theta_{\mathrm{s}i} - \pi - \omega_{\mathrm{s}i})$. 
\end{lemma}

\begin{remark} \label{Remark:Blockage:SinglePrimary:MAP:ZBL:HBL:NOL}
For the ZBL, HBL and NOL scenarios, MAP of the $i^{\mathrm{th}}$ secondary transmitter in \eqref{eq:Blockage:SinglePrimary:MAP} can be simplified as $p_{\mathrm{m}i,\mathrm{ZBL}} =$ $1 - \exp( - (\kappa_\los/\kappa_{i\mathrm{0}}) x_{\seco i}^{\alpha_\los})$, $p_{\mathrm{m}i,\mathrm{HBL}} =$ $1 - \exp( - ({\kappa_\nlos}/{\kappa_{i\mathrm{0}}}) x_{\seco i}^{\alpha_\nlos})$ and $p_{\mathrm{m}i,\mathrm{NOL}} = 1 - p_\los (x_{\seco i}) \exp(- ({\kappa_\los}/{\kappa_{i\mathrm{0}}}) x_{\seco i}^{\alpha_\los}) $, respectively.
\end{remark}

\begin{remark} \label{Remark:Blockage:SinglePrimary:MAP-Omni}
Note that MAP of the $i^{\mathrm{th}}$ secondary transmitter under omni SLS is given by \eqref{eq:Blockage:SinglePrimary:MAP} with $\kappa_{i\mathrm{0}} = 1$.
\end{remark} 

\begin{corr} \label{Corollary:Blockage:SinglePrimary:MAP-ideal}
The MAP in \eqref{eq:Blockage:SinglePrimary:MAP} under ideal beam pattern with zero side-lobe gain assumption is simplified as $p_{\mathrm{m}i,\mathrm{ideal}}$
\begin{align*}
= \begin{cases}
1 \! - 
\!\!
\sum\nolimits_{T} 
p_T (x_{\seco i}) \, e^{- \frac{\kappa_T x_{\seco i}^{\alpha_T}}{\apr\ast}}  &\text{if } \,
\begin{aligned} 
&\lvert \theta_{\mathrm{s}i} \rvert \leq \phi_{\mathrm{pr}}/2, \\ 
&\lvert \theta_{\mathrm{s}i} - \pi - \omega_{\mathrm{s}i} \rvert \leq \phi_{\mathrm{st}}/2 
\end{aligned} \\
1 \ &\text{otherwise}
\end{cases}.
\end{align*}
\end{corr}
\begin{IEEEproof}
For ideal beam pattern, $b_{\mathrm{pr}} = b_{\mathrm{st}} = 0 $. Therefore, $\kappa_{i\mathrm{0}} = a_{\mathrm{pr}} a_{\mathrm{st}}$ when $\lvert \theta_{\mathrm{s}i} \rvert \leq \phi_{\mathrm{pr}}/2, \lvert \theta_{\mathrm{s}i} - \pi - \omega_{\mathrm{s}i} \rvert \leq \phi_{\mathrm{st}}/2$ and $0$ otherwise. Substituting $\kappa_{i\mathrm{0}}$ in \eqref{eq:Blockage:SinglePrimary:MAP} gives the desired result.
\end{IEEEproof}

\subsection{Activity factor (AF)} \label{Section:Blockage:SinglePrimary:Analysis:AF}
To quantify the transmission opportunity of the secondary network within a certain distance $R$ around the primary, we define a metric AF as the ratio of the active to the total number of secondary transmitters inside a region $\mathcal{B} (0, R)$ {\em i.e.} 
\begin{align}
\eta_\seco = \frac{1}{\pi \lambda_\seco R^2} \, \mathbb{E} \Big[ \sum\nolimits_{\X_{\mathrm{s}i} \in \Phi_\seco \cap \mathcal{B} (0, R)} U_i \Big],
\label{eq:Blockage:SinglePrimary:AF:Definition}
\end{align} 
where $\mathcal{B} (0, R)$ denotes the region of interest. AF variation with $R$ reflects the effect of primary operator on the secondary activity as the distance from the primary receiver increases. It is given in the following theorem (see Appendix \ref{thrm:proof:Blockage:AF} for proof). 

\begin{theorem}\label{Theorem:Blockage:SinglePrimary:AF}
The AF of the secondary network inside the region of interest $\mathcal{B}(0,R)$ is given as 
\begin{align}
&\!
\eta_\seco = 1 - \!\frac{1}{\pi R^2} 
\! 
\int_0^{2\pi}
\!\!
\mathbb{E}_{\omega_\seco} \bigg[ \int_0^{R} 
\!\! 
\sum_T
p_T (x_\seco) \, e^{ - \left({\kappa_T}/{\kappa_\mathrm{0}}\right) x_\seco^{\alpha_T}} x_\seco \dd x_\seco \bigg] \dd \theta_\seco \nonumber \\
&= 1 - \!\frac{1}{\pi R^2} 
\!
\int_0^{2\pi} 
\!\! 
\mathbb{E}_{\omega_\seco} \left[ \, e^{-p} \sum_{n = 0}^\infty 
\!
\frac{(-\mu)^n}{n!} 
\sum_T
\alpha_T^{-1} \, \psi_T \bigg( \! n + 2, \frac{\kappa_T}{\kappa_\mathrm{0}} \! \bigg) \! 
- \frac{1}{\alpha_\nlos} \psi_\nlos \bigg(\! 2, \frac{\kappa_\nlos}{\kappa_\mathrm{0}} \! \bigg) \right] \dd \theta_\seco,
\label{eq:Blockage:SinglePrimary:AF}
\end{align}
where $T = \{\los,\nlos\}$, $p_\los (x_{\seco}) = \exp(-\mu x_{\seco} - p)$, $p_\nlos (x_{\seco}) = 1 - p_\los (x_{\seco})$, $\kappa_\mathrm{0} = g_{\mathrm{pr}} (\theta_\seco) g_{\mathrm{st}} (\theta_\seco - \pi - \omega_\seco)$ and 
\begin{align}
\psi_T (m, u) = 
\begin{cases}
\psi_\los (m, u) =  u^{-\frac{m}{\alpha_\los}} \, \Gamma (m/\alpha_\los, uR^{\alpha_\los}) \!\! &\mathrm{when} \ T = \los \\
\psi_\nlos (m, u) = - u^{-\frac{m}{\alpha_\nlos}} \, \Gamma (m/\alpha_\nlos, uR^{\alpha_\nlos}) \!\! &\mathrm{when} \ T = \nlos \\
\end{cases}.
\label{eq:Blockage:SinglePrimary:AF:PsiFunction}
\end{align} 
\end{theorem}

\begin{remark} \label{Remark:AF-Omni}
Note that AF of the secondary network under omni SLS $\eta_{\seco, \omni}$ is given by \eqref{eq:Blockage:SinglePrimary:AF} with $\kappa_\mathrm{0} = 1$.
\end{remark}

We can simplify \eqref{eq:Blockage:SinglePrimary:AF} for sectorized/ideal beam pattern to get the following results (see Appendix \ref{thrm:proof:Blockage:AF-sector} for proof).

\begin{corr}\label{Corollary:Blockage:SinglePrimary:AF-sector}
Under sectorized beam pattern, AF is given as 
\begin{align}
&\eta_{\seco, \mathrm{sec}} = 1  - \frac{2}{R^2} \int_0^{R} 
\!
\sum\limits_{T} 
p_T (x_\seco) \sum\limits_{i = 1:4}
\!\!
\mathcal{Q}_i \, e^{- ({\kappa_T}/{ \mathcal{G}_i}) x_\seco^{\alpha_T}} x_\seco \dd x_\seco \nonumber \\
&= 1 \! - \! \frac{2}{R^2} \sum_{i = 1:4} \! \mathcal{Q}_i \left[ \sum_{n = 0}^\infty \frac{(-\mu)^n}{n!} \! \sum_{T} \frac{e^{-p}}{\alpha_T} \ \psi_T \bigg(n + 2, \frac{\kappa_T}{\mathcal{G}_i} \bigg) 
- \frac{1}{\alpha_\nlos} \ \psi_\nlos \Big(2, \frac{\kappa_\nlos}{\mathcal{G}_i}\Big) \right],
\label{eq:Blockage:SinglePrimary:AF:Sector}
\end{align} 
where the values of $\mathcal{Q}_i$ and $\mathcal{G}_i$ are given as\\ \\
{\hspace*{5cm}
\resizebox{0.490\textwidth}{!}{
\begin{tabular}{|c|c|c|c|c|}
\hline
$i$ & $1$ & $2$ & $3$ & $4$ \\ \hline
$\mathcal{Q}_i$ & $\qpr\qst$ & $\qpr(1 - \qst)$ & $(1-\qpr)\qst$ & $(1-\qpr)(1-\qst)$ \\ \hline
$\mathcal{G}_i$ & $\apr\ast$ & $\apr\bst$ & $ \apr\qst$ & $\bpr\bst$ \\ \hline
\end{tabular}}}
\end{corr}

\begin{remark}\label{Remark:Blockage:SinglePrimary:AF-ideal}
Under ideal beam pattern assumption ($\bpr = \bst$ $= 0$), the AF in \eqref{eq:Blockage:SinglePrimary:AF:Sector} can be simplified as 
\begin{align}
&\eta_{\seco, \mathrm{ideal}} = 
1 \! - \! \frac{2 \mathcal{Q}_1}{R^2} 
\bigg[ \sum_{n = 0}^\infty \! \frac{(-\mu)^n}{n!} \! \sum_{T} \frac{e^{-p}}{\alpha_T} \, \psi_T \bigg( \! n + 2, \frac{\kappa_T}{\mathcal{G}_1} \! \bigg) \! - \! \frac{1}{\alpha_\nlos} \, \psi_\nlos \bigg(\! 2, \frac{\kappa_\nlos}{\mathcal{G}_1} \! \bigg) \bigg]\!. \!\!
\label{eq:Blockage:SinglePrimary:AF:Ideal}
\end{align}
\end{remark}

We can further simplify $\eta_\seco$, $\eta_{\seco, \mathrm{sec}}$, $\eta_{\seco, \mathrm{ideal}}$, and $\eta_{\seco, \mathrm{omni}}$ for ZBL, HBL, and NOL scenarios as given in Table \ref{Table::Blockage:SinglePrimary:AF:Omni:SectorBeam:IdealBeam:ZBL:HBL:NOL}. 
\begin{table*}[ht!]
\centering
\caption{AF $\eta_\seco$, $\eta_{\seco, \mathrm{sec}}$, $\eta_{\seco, \mathrm{ideal}}$ and $\eta_{\seco, \mathrm{omni}}$ under ZBL, HBL and NOL scenarios.}
\label{Table::Blockage:SinglePrimary:AF:Omni:SectorBeam:IdealBeam:ZBL:HBL:NOL}
\resizebox{0.95\textwidth}{!}{
\begin{tabular}{|l|c|c|c|c|}
\hline
Blockage scenarios/beam  pattern 	& $\eta_\seco$ & $\eta_{\seco, \mathrm{sec}}$ &  $\eta_{\seco, \mathrm{ideal}}$ &  $\eta_{\seco, \mathrm{omni}}$ \\ \hline
ZBL: $1 - ({1}/{R^2}) \, ({2}/{\alpha_\los}) \times$	& $\frac{1}{2\pi} \int_0^{2\pi} \expects{\omega_\seco}{\psi_\los \left(2, {\kappa_\los}/{\kappa_0}\right)} \dd \theta_\seco$	& $\sum\nolimits_{i = 1:4} \! \mathcal{Q}_i \, \psi_\los \left(2, \, {\kappa_\los}/{\mathcal{G}_i}\right)$ 	& $\mathcal{Q}_1 \, \psi_\los \left(2, \, {\kappa_\los}/{\mathcal{G}_1}\right)$ & $\psi_\los \left(2, \, {\kappa_\los} \right)$ \\ \hline
HBL: $1 + ({1}/{R^2}) \, ({2}/{\alpha_\nlos}) \times$	& $\frac{1}{2\pi} \int_0^{2\pi} \expects{\omega_\seco}{\psi_\nlos \left(2, {\kappa_\nlos}/{\kappa_0}\right)} \dd \theta_\seco$ & $\sum\nolimits_{i = 1:4} \! \mathcal{Q}_i \, \psi_\nlos \left(2, \, {\kappa_\nlos}/{\mathcal{G}_i}\right)$ 	& $\mathcal{Q}_1 \, \psi_\nlos \left(2, \, {\kappa_\nlos}/{\mathcal{G}_1}\right)$ & $ \psi_\nlos \left(2, \, {\kappa_\nlos}\right)$ \\ \hline
NOL: $1 - \frac{1}{R^2} \frac{2}{\alpha_\los} \, e^{-p} \, \sum\nolimits_{n = 0}^\infty \frac{(-\mu)^n}{n!} \times$	& $\frac{1}{2\pi} \int_0^{2\pi} \expects{\omega_\seco}{\psi_\los \left(n+ 2, {\kappa_\los}/{\kappa_0}\right)} \dd \theta_\seco$ & $ \sum\nolimits_{i = 1:4} \, \mathcal{Q}_i \, \psi_\los \left( n + 2, \, {\kappa_\los}/{\mathcal{G}_i} \right)$ & $ \mathcal{Q}_1 \, \psi_\los \left( n + 2, \, {\kappa_\los}/{\mathcal{G}_1} \right)$ & $\psi_\los \left( n + 2, \, \kappa_\los \right)$ \\ \hline
\end{tabular}}\vspace{-.1in}
\end{table*}
In  ZBL or HBL, there is only one type of link (either all LOS or all NLOS), resulting in AF equal to the one obtained without taking blockages into account. Hence, utilizing results in \cite{tripathi2024coverage}, we can show that antenna directionality improves AF due to the dominating effect of a decrease in the antenna beamwidth over an increase in the main-lobe gain. Similarly, for the NOL scenario under ideal ULA beam pattern, from Table \ref{Table::Blockage:SinglePrimary:AF:Omni:SectorBeam:IdealBeam:ZBL:HBL:NOL}, we have 
\begin{align*}
&\frac{1-{\eta}_{\seco, \mathrm{ideal, NOL}}}{1-{\eta}_{\seco, \mathrm{omni, NOL}}} =(\kappa/(2\pi))^{\frac{4}{\alpha_\los}} \, \mathcal{Q}_1^{1 - \frac{2}{\alpha_\los}} \, \mathcal{S}
\\
&\text{with }
\mathcal{S}= \frac{\int_0^{\fracS{\kappa_\los \, R^{\alpha_\los}}{\mathcal{G}_1}} t^{\frac{2}{\alpha_\los} - 1} e^{-t -\mu \mathcal{G}_1^{\fracS{1}{\alpha_\los}} t^{\fracS{1}{\alpha_\los}}} \dd t}{ \int_0^{\kappa_\los \, R^{\alpha_\los}} t^{\fracS{2}{\alpha_\los} - 1} e^{-t -\mu t^{\fracS{1}{\alpha_\los}}} \dd t} \nonumber.
\end{align*} 
Note that in the expression of $\mathcal{S}$, the upper limit $\fracS{\kappa_\los \, R^{\alpha_\los}}{\mathcal{G}_1}<{\kappa_\los \, R^{\alpha_\los}}{}$ and the numerator integrand is less than the denominator integrand. Hence, $\mathcal{S}<1$, which shows that antenna directionality improves the AF of the secondary networks. 

Further from Theorem \ref{Theorem:Blockage:SinglePrimary:AF}, we can see that blockages ({\em i.e.} higher $\mu$) improve AF by increasing the number of active secondary transmitters. However, it also decreases their individual interference at the primary receiver, which results in a trade-off as explained in the following example.

\begin{example}\label{Example:Blockage:SinglePrimary:MAP}
If we ignore fading and directionality, the SLS restriction  $C_Tu^{-\alpha_T} \le \rho/p_\seco$ leads to all active secondary transmitters to lie outside the ball of radius $ u = \left({C_T p_\seco}/{\rho}\right)^{1/\alpha_T}$. In LOS only case, the mean interference power due to all active secondary links is $\mathbb{E} [I_\los] = \mathbb{E} [ \, \sum\nolimits_{ \X_{\mathrm{s}i} \in \Phi_\seco} p_\seco$ $C_\los x_{\seco i}^{- \alpha_\los}] = ({2 \pi \lambda_\seco p_\seco C_\los} \, u^{2 - \alpha_\los } )/({\alpha_\los - 2})$. Similarly, for the NLOS only case, the mean interference is $\mathbb{E} [ I_\nlos ] = (2 \pi \lambda_\seco p_\seco C_\nlos$ $\tilde{u}^{2 - \alpha_\nlos })/({\alpha_\nlos - 2})$ with $\tilde{u} = \left({C_\nlos p_\seco}/{\rho}\right)^{1/\alpha_\nlos}$. We can show that $\tilde{u} < u$, confirming that blockages increase the number of active secondary transmitters. To approximately compare the mean interference for the two cases, let us define $\zeta=C_\los/C_\nlos>1$ and $\alpha_\los = \alpha_\nlos = \alpha$. Then, we have $\tilde{u}/{u} = \zeta^{- 1/\alpha}$ and $\expect{I_\los}/\expect{I_\nlos} = \zeta^{^{2 - 2/\alpha}} > 1$ for $\alpha > 2$. This means that even after relaxed restrictions, the NLOS-only scenario gives a lower interference. 
\end{example}

The presence of other factors such as noise, fading, and directionality can add complexity to this comparison, especially when $0 < L_\mu < \infty$. For example, the presence of noise can reduce the effect of the interference on the primary performance. Similarly, fading can reduce the difference between $I_\los$ and $I_\nlos$, and the corresponding MAPs. 
 
\subsection{Coverage of primary link} \label{Section:Blockage:SinglePrimary:Analysis:PrimaryCov}
The instantaneous SINR at the primary receiver is 
\begin{align}
\SINR_{\prim\mathrm{0}} &= \frac{ H_{\mathrm{p0}} p_\prim \gpt{0}\gpr{0} \ell_{T_{\prim\mathrm{0}}} (r_\prim) }{ \sigma^2 \ + \ I_{ \mathrm{s} }  \left( \Phi_{\mathrm{s}}  \right) },
\label{eq:Blockage:SinglePrimary:SINR:PrimaryLink}
\end{align}
where $\noise$ is the noise power and $\ell_{T_{\prim\mathrm{0}}} (r_\prim) = \ell_{T_{\prim\mathrm{0} \to \prim\mathrm{0}}} (\|\X_{\prim\mathrm{0}} -$ $\Y_{\prim\mathrm{0}}\|)$ is the path loss (see \eqref{eq:Blockage:SinglePrimary:PathLoss}) of the primary link with link-length $r_\prim$. Here, $I_{ \mathrm{s} }  \left( \Phi_{\mathrm{s}} \right)$ is the total interference due to all active secondary transmitter at the primary receiver, given by
\begin{align}
I_{ \mathrm{s} }  \left( \Phi_{\mathrm{s}} \right) &= \sum\nolimits_{\X_{\mathrm{s}i} \in \Phi_{\seco}} p_\seco \kappa_{i\mathrm{0}} U_{i} G_{i0} \ell_{T_{\seco i \to \prim\mathrm{0}}} (x_{\seco i}).
\label{eq:Blockage:SinglePrimary:SINR:PrimaryLink:SecondaryInterference}
\end{align}

The performance of the primary link can be measured by the complementary cumulative density function (CCDF) of $\SINR_{\prim\mathrm{0}}$, also known as coverage probability, as given in the following theorem (See Appendix \ref{thrm:proof:Blockage:PrimaryCov} for proof).

\begin{theorem}\label{Theorem:Blockage:SinglePrimary:PrimaryCov}	
The coverage probability of the primary link under spatially aware SLS with transmit-restriction threshold $\rho$ is 
\begin{align}
&p_\mathrm{cp}(\SThres, \mu, \rho) =\prob{\SINR_{\prim\mathrm{0}} > \tau \lvert \, \rho, \mu}= \!\!\!\! 
\sum\limits_{T \in \{\los,\nlos\}} 
\!\!\!\!
p_T (r_\prim)
\exp 
\!
\left(
\!
- s_T \sigma^2 
\!
- \lambda_\seco 
\!\!
\int_0^{2\pi} 
\!\!\!\!
\int_0^\infty 
\!\!
\left( 
\!
1 
\! -  \!\!\!\!\!\!
\sum_{T_1 \in \{\los,\nlos\}} 
\!\!\!\!\!
p_{T_1} (x_\seco) \ \times
\right.\right.
\nonumber \\
&\
\left.\left.
\mathbb{E}_{\omega_\seco}  
\!\!
\left[
\!
\frac{1 
\! - \!
e^{- \frac{(\rho/p_\prim/C_T) r_\prim^{\alpha_T}}{\gpt{0}\gpr{0}} 
\!
\left(
\!\!
\tau +
\!
\frac{p_\prim}{p_\seco} 
\!
\frac{C_T}{C_{T_1}}
\!
\frac{x_\seco^{\alpha_{T_1}}}{r_\prim^{\alpha_T}} 
\!
\frac{\gpt{0}\gpr{0}}{ g_{\mathrm{pr}} \left(\! \theta_\seco \!\right) g_{\mathrm{st}} \left(\! \theta_\seco \! - \! \pi \! - \! \omega_\seco \!\right)} 
\!
\right)}}{1 + \tau \frac{p_\seco }{p_\prim } \frac{C_{T_1}}{C_T} \left(\frac{r_\prim^{\alpha_T}}{x_\seco^{\alpha_{T_1}}}\right)  \frac{g_{\mathrm{pr}} \left( \theta_\seco \right) g_{\mathrm{st}} \left( \theta_\seco - \pi - \omega_\seco \right)}{\gpt{0}\gpr{0}}} \ +
e^{- \frac{ (\rho/p_\seco/C_{T_1}) x_\seco^{\alpha_{T_1}}}{g_{\mathrm{pr}} \left( \theta_\seco \right) g_{\mathrm{st}} \left( \theta_\seco - \pi - \omega_\seco \right)}}  
\right] 
\right) 
x_\seco \mathrm{d}x_\seco \mathrm{d}\theta_\seco 
\right), 
\label{eq:Blockage:SinglePrimary:PrimaryCov} \\
&= \sum_{T\in \{\los,\nlos\}} p_T (r_\prim) \exp 
\!
\left( \! - s_T \sigma^2 \! - \! \lambda_\seco 
\!
\left[ \mathcal{N} \! \left(\frac{2}{\alpha_\nlos}, s_T \rho  \right)
\!
\mathsf{n}_\mathrm{3} 
\!
\left(
\frac{2}{\alpha_\nlos}
\right)
\! 
\! 
+
e^{-p} \sum\nolimits_{n = 0}^\infty
\!
\frac{(- \mu)^{n}}{n!}
\!
\left[
\mathcal{N} \! \left( \frac{(n \! + \! 2)}{\alpha_\los}, s_T \rho \right) 
\times
\right.\right.\right.
\nonumber\\
&\qquad\qquad\qquad\qquad
\left.\left.\left.
\mathsf{n}_\mathrm{3} 
\!
\left(
\frac{(n \! + \! 2)}{\alpha_\los}
\right)
\! - \mathcal{N}  \left(
\frac{(n \! + \! 2)}{\alpha_\nlos}, s_T \rho \right) \mathsf{n}_\mathrm{3} 
\left(
\frac{(n \! + \! 2)}{\alpha_\nlos}
\right)
\right]  
\right] \right) \!, \!\!\!
\label{eq:Blockage:SinglePrimary:PrimaryCov:Simplified}
\end{align}
where $s_T = \tau r_\prim^{\alpha_T} /(C_T p_\prim g_{\mathrm{pt}} (0) g_{\mathrm{pr}} (0))$,  
\begin{align*}
\mathcal{N} ( {m}/{\alpha_{T_1}}, s) = \frac{1}{\alpha_{T_1}} \kappa_{T_1}^{\!\! - {m}/{\alpha_{T_1}}}
\!
\left[ 
s^{\frac{m}{\alpha_{T_1}}}
\!
\mathsf{n}_\mathrm{1} 
\!
\left(
\!
\frac{m}{\alpha_{T_1}}
\!
\right) 
\! - \!
\Gamma 
\!
\left( \frac{m}{\alpha_{T_1}} \right) 
\!\! + \!
\mathsf{n}_\mathrm{2} 
\!
\left(
\!
\frac{m}{\alpha_{T_1}}, {s} 
\right) 
\!
\right].
\end{align*}
Here, $\mathsf{n}_\mathrm{1} (k) = \pi \csc (\pi k)$, $\mathsf{n}_\mathrm{2} (k, \nu) = \int_{\nu}^\infty e^{- u} u^{-1} (u - \nu)^{k} \dd u$,  $\mathsf{n}_\mathrm{3} (k) =  \mathbb{E}_{\omega_\seco}[ \int_0^{2\pi} ( g_{\mathrm{pr}} ( \theta_\seco )g_{\mathrm{st}} ( \theta_\seco - \pi - \omega_\seco ) )^{k} \mathrm{d}\theta_\seco ]$. 
\end{theorem}

We can observe that the effect of antenna directionality is separable via $s_T$ and $\mathsf{n}_\mathrm{3}$ terms in the expression of $p_{\mathrm{cp}}$. Therefore, to better understand the impact of antenna directionality, we are now going to further simplify $p_{\mathrm{cp}}$ under the sectorized/ideal beam pattern approximation to get the following results (see Appendix \ref{thrm:proof:Blockage:PrimaryCov-sector} for proof).

\begin{corr} \label{Corollary:Blockage:SinglePrimary:PrimaryCov-Sector}
Under sectorized beam approximation, the primary coverage $p_\mathrm{cp, sec}(\tau, \mu, \rho)$
 is given by \eqref{eq:Blockage:SinglePrimary:PrimaryCov:Simplified} with $s_T =$ $\tau r_\prim^{\alpha_T}/(C_T p_\prim \apt \apr)$ and $\mathsf{n}_\mathrm{3} (k) = 2\pi \sum\nolimits_{i = 1:4} \mathcal{Q}_i $ $ \mathcal{G}_i^{k}$.
\end{corr}

\begin{corr} \label{Corollary:Blockage:SinglePrimary:PrimaryCov-Ideal}
Under ideal beam pattern assumption, 
the primary coverage $p_{\mathrm{cp}, \mathrm{ideal}}(\tau, \mu, \rho)$ is given by \eqref{eq:Blockage:SinglePrimary:PrimaryCov:Simplified} with $s_T = \tau r_\prim^{\alpha_T}/(C_T p_\prim \apt \apr)$ and $\mathsf{n}_\mathrm{3} \left(k\right) = 2\pi \mathcal{Q}_1 \mathcal{G}_1^{k}$.
\end{corr}

\begin{remark} \label{Remark:Blockage:SinglePrimary:PrimaryCov-Omni}
The primary coverage $p_\mathrm{cp, omni}(\SThres, \mu, \rho)$ for a network under omni-SLS is given by \eqref{eq:Blockage:SinglePrimary:PrimaryCov:Simplified} with $\mathrm{n_3} (\cdot) = 2\pi$ and $s_T = {\tau r_\prim^{\alpha_T}}/{C_T p_\prim}$.
\end{remark}

We can further simplify $p_{\mathrm{cp}}$, $p_{\mathrm{cp, sec}}$, $p_{\mathrm{cp,ideal}}$, and $p_{\mathrm{cp},\omni}$ for ZBL, HBL and NOL scenarios to get Table \ref{Table::Blockage:SinglePrimary:PrimaryCov:Omni:SectorBeam:IdealBeam:ZBL:HBL:NOL} where $\mathcal{I} \! = \! \sum_{n = 0}^\infty [{(\! - \mu)^{n}}/({n!})] \mathcal{N} ({(n \! + \! 2)}/{\alpha_\los}, s_\los \rho)$ $\mathsf{n}_\mathrm{3} ({(n \! + \! 2)}/{\alpha_\los})$. Note the slight changes in the values of $s_T$ and $\mathsf{n}_\mathrm{3}$.
\begin{table*}[!t]
\centering
\caption{Primary coverage under sectorized, ideal and omnidirectional beam pattern for ZBL, HBL and NOL scenarios.}
\label{Table::Blockage:SinglePrimary:PrimaryCov:Omni:SectorBeam:IdealBeam:ZBL:HBL:NOL}
\resizebox{0.98\textwidth}{!}{
\begin{tabular}{|l|c|c|c|c|}
\hline
Scenarios 	& $p_{\mathrm{cp}}(\tau, \mu, \rho)$ & $p_{\mathrm{cp}, \mathrm{sec}}(\tau, \mu, \rho)$ &  $p_{\mathrm{cp}, \mathrm{ideal}}(\tau, \mu, \rho)$ &  $p_{\mathrm{cp}, \mathrm{sec}}(\tau, \mu, \rho)$ \\ \hline
ZBL: $p_{\mathrm{cp}}(\SThres, 0, \rho)$ & $e^{\!- s_\los \sigma^2 \! - \! \lambda_\seco \, \mathcal{N} ( {2}/{\alpha_\los}, s_\los \rho) \, \mathsf{n}_\mathrm{3} \left({2}/{\alpha_\los} \right)}$ & $s_\los = \frac{\tau r_\prim^{\alpha_\los}}{C_\los p_\prim \apt \apr}$, $\mathsf{n}_\mathrm{3} \left(\frac{2}{\alpha_\los}\right) = 2\pi \sum_{i = 1:4} \mathcal{Q}_i \mathcal{G}_i^{2/\alpha_\los}$ & $s_\los = \frac{\tau r_\prim^{\alpha_\los}}{C_\los p_\prim \apt \apr}$, $\mathsf{n}_\mathrm{3} \left(\frac{2}{\alpha_\los}\right) = 2\pi \mathcal{Q}_1 \mathcal{G}_1^{2/\alpha_\los}$ & $s_\los = \frac{\tau r_\prim^{\alpha_\los}}{C_\los p_\prim}$, $\mathsf{n}_\mathrm{3} \left(\cdot\right) = 2\pi$ \\ \hline
HBL: $p_{\mathrm{cp}}(\SThres, \infty, \rho)$ & $e^{\!- s_\nlos \sigma^2 \! - \! \lambda_\seco \, \mathcal{N} ( {2}/{\alpha_\nlos}, s_\nlos \rho) \, \mathsf{n}_\mathrm{3} \left({2}/{\alpha_\nlos} \right)}$ & $s_\nlos = \frac{\tau r_\prim^{\alpha_\nlos}}{C_\nlos p_\prim \apt \apr}$, $\mathsf{n}_\mathrm{3} \left(\frac{2}{\alpha_\nlos}\right) = 2\pi \sum_{i = 1:4} \mathcal{Q}_i \mathcal{G}_i^{2/\alpha_\nlos}$ & $s_\nlos = \frac{\tau r_\prim^{\alpha_\nlos}}{C_\nlos p_\prim \apt \apr}$, $\mathsf{n}_\mathrm{3} \left(\frac{2}{\alpha_\nlos}\right) = 2 \pi \mathcal{Q}_1 \mathcal{G}_1^{2/\alpha_\nlos}$ & $s_\nlos = \frac{\tau r_\prim^{\alpha_\nlos}}{C_\nlos p_\prim}$, $\mathsf{n}_\mathrm{3} \left(\cdot\right) = 2\pi$ \\ \hline
NOL: $p_{\mathrm{cp}}(\SThres, \mu, \rho) \lvert_{C_\nlos = 0}$ & $p_\los (r_\prim) e^{- s_\los \sigma^2 - \lambda_\seco e^{-p} \left(({2\pi}/{\mu^2}) - \mathcal{I}\right)}$ & $s_\los = \frac{\tau r_\prim^{\alpha_\los}}{C_\los p_\prim \apt \apr}$, $\mathsf{n}_\mathrm{3} \left({m}/{\alpha_\los}\right) = 2\pi \sum_{i = 1:4} \mathcal{Q}_i \mathcal{G}_i^{m/\alpha_\los}$ & $s_\los = \frac{\tau r_\prim^{\alpha_\los}}{C_\los p_\prim \apt \apr}$, $\mathsf{n}_\mathrm{3} \left(\frac{m}{\alpha_\los}\right) = 2\pi \mathcal{Q}_1 \mathcal{G}_1^{m/\alpha_\los}$ & $s_\los = \frac{\tau r_\prim^{\alpha_\los}}{C_\los p_\prim}$, $\mathsf{n}_\mathrm{3} \left(\cdot\right) = 2\pi$ \\ \hline
\end{tabular}}
\vspace{-.1in}
\end{table*}

\subsection{Coverage of secondary links} \label{Section:Blockage:SinglePrimary:Analysis:SecondaryCov}
To evaluate the performance of secondary network, a typical secondary link is picked randomly out of all secondary links.
\begin{figure}[b!]
\centering
{\includegraphics[trim = 5 45 5 5, clip, scale=0.15]{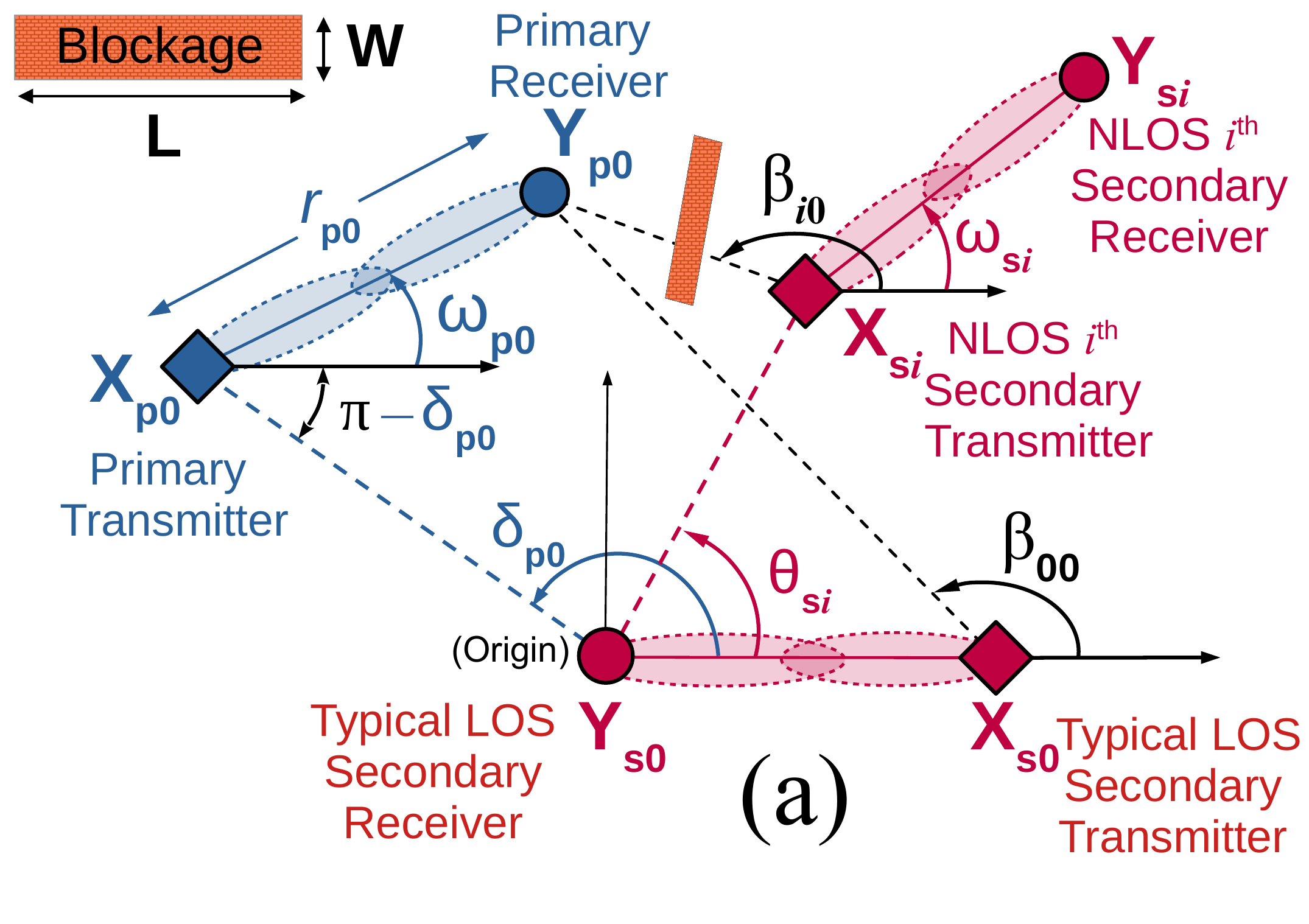} \label{fig:Blockage:SinglePrimary:SimpleCase-secondary} }
{\includegraphics[trim = 2 5 460 320, clip, scale=0.295]{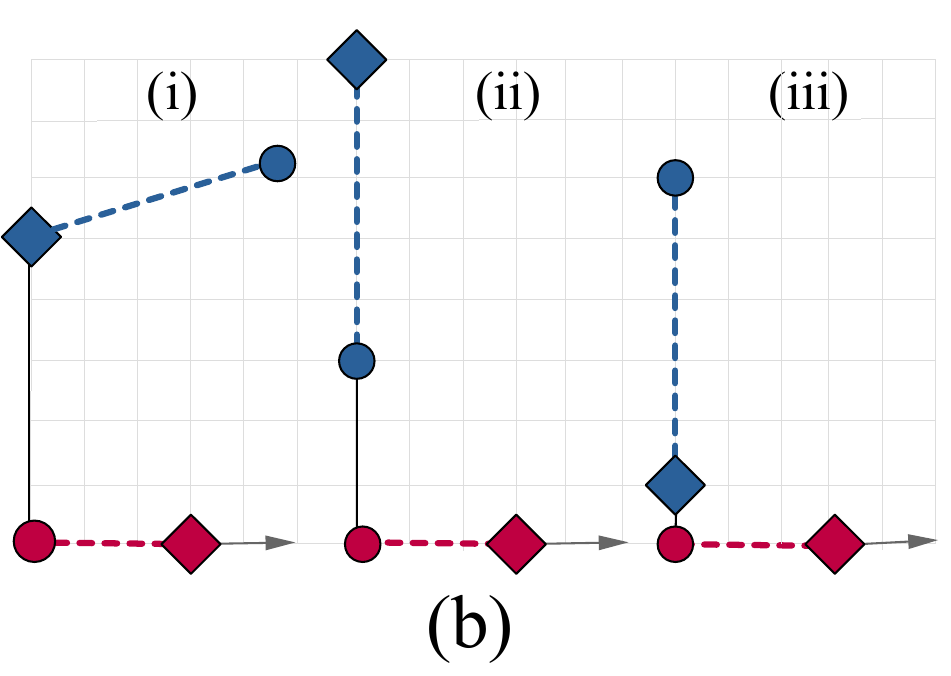} \label{fig:Blockage:SinglePrimary:SetUps}}
\caption{An illustration of (a) the new coordinate reference for the secondary links' coverage analysis and (b)the three set-ups for the primary link's location with parameters $\{\angle \delta_{\prim\mathrm{0}}, x_{\prim\mathrm{0}}, \angle{\omega_{\prim\mathrm{0}}}\}$ taken as (i) $\{\pi/2, 50 \,\text{m}, \pi/12\}$, (ii) $\{\pi/2, 80 \, \text{m}, -\pi/2\}$ and (iii) $\{\pi/2, 10 \,\text{m}, \pi/2\}$.}
\label{fig:Blockage:SinglePrimary:SystemModel-Secondary}
\vspace*{-0.36cm}
\end{figure}
\subsubsection{Secondary coverage at a given location}
We first consider a secondary link $\X_{\seco0}-\Y_{\seco0}$ at a given location/orientation relative to the primary link. For computational simplicity, we transform the coordinate axes such that this link $\X_{\mathrm{s0}}-\Y_{\mathrm{s0}}$ lies on the X-axis (see Fig. \ref{fig:Blockage:SinglePrimary:SystemModel-Secondary}(a)), with its receiver at the origin {\em i.e.} $\Y_{\mathrm{s0}} = \mathbf{o}$. Note that coordinate transformation maintains the relative locations between the primary and all secondary links; hence, it will not affect the performance. In the new coordinate reference, the expression for the received power at the primary receiver due to the $i^{\mathrm{th}}$ secondary transmitter is $P^{'}_{\mathrm{s}i \to \prim\mathrm{0}} = p_\seco G_{i\mathrm{0}} g_{\mathrm{pr}} (\omega_{\prim\mathrm{0}} - \beta_{i\mathrm{0}}) g_{\mathrm{st}}( \beta_{i\mathrm{0}} - \omega_{\mathrm{s}i})  \ell_{T_{\seco i \to \prim\mathrm{0}}} (z_{i\mathrm{0}})$, where $z_{i\mathrm{0}} = \|x_{\seco i}\angle \theta_{\seco i} - \mathbf{Y}_{\prim \mathrm{0}}\|$ and $\beta_{i\mathrm{0}}$ represents the orientation of primary receiver at $\Y_{\prim\mathrm{0}}$ with respect to the $i^{\mathrm{th}}$ secondary transmitter and is given as $\beta_{i\mathrm{0}} = \angle (x_{\seco i} \angle \theta_{\seco i} - \mathbf{Y}_{\prim \mathrm{0}}) =  \theta_{\seco i}$ $- \sin^{-1} ( ({y_{\prim \mathrm{0}}}/{z_{i\mathrm{0}}}) \sin (\theta_{\seco i} - \delta_{\prim \mathrm{0}} + \sin^{-1} ( ({r_\prim}/{y_{\prim \mathrm{0}}}) \sin (\delta_{\prim \mathrm{0}} - \omega_{\prim \mathrm{0}} ))))$. Note that the received power at the primary receiver due to the considered secondary transmitter is $P^{'}_{\mathrm{s0} \to \prim\mathrm{0}} = p_\seco G_\mathrm{00} g_{\mathrm{pr}}  (\omega_{\prim\mathrm{0}} - \beta_\mathrm{00} ) g_{\mathrm{st}} ( \beta_\mathrm{00} - \pi ) \ell_{T_{\seco\mathrm{0} \to \prim\mathrm{0}}} (z_\mathrm{00})$, where $z_\mathrm{00} = \| r_\seco$ $\angle 0 - \Y_{\prim\mathrm{0}} \|$. The indicator representing the transmission activity of the $ i^{\mathrm{th}} $ secondary transmitter, is  $U^{'}_{i} = \mathbbm{1} (p_\seco G_{i\mathrm{0}} g_{\mathrm{pr}} (\omega_{\prim\mathrm{0}} - \beta_{i\mathrm{0}}) g_{\mathrm{st}} ( \beta_{i\mathrm{0}} - \omega_{\mathrm{s}i} ) \ell_{T_{\seco i \to \prim\mathrm{0}}} (z_{i\mathrm{0}}) < \rho )$. Due to the new coordinate reference, the expression of MAP in \eqref{eq:Blockage:SinglePrimary:MAP} is slightly modified for the considered and $i^{\mathrm{th}}$ secondary transmitter (see the following lemmas). However, their proofs are similar to the Lemma \ref{Lemma:Blockage:SinglePrimary:MAP}.

\begin{lemma} 
The MAP of the $i^{\mathrm{th}}$ secondary transmitter at location $\X_{\mathrm{s}i} = x_{\mathrm{s}i} \angle \theta_{\seco i}$ in the presence of primary link is $p^{'}_{\mathrm{m}i} = 1 - \sum_{T = \{\los,\nlos\}} p_T (z_{i\mathrm{0}}) \exp (- (\rho z_{i\mathrm{0}}^{ \alpha_T})/(C_T p_\seco g_{\mathrm{st}} (\beta_{i\mathrm{0}} - \omega_{\mathrm{s}i}) g_{\mathrm{pr}} (\omega_{\prim\mathrm{0}} - \beta_{\mathrm{0}i})))$. 
\end{lemma}

\begin{lemma}
The MAP of the considered secondary transmitter $\X_{\seco0}$ due to primary link is $p^{'}_{\mathrm{m0}} = 1 - \sum_{T = \{\los,\nlos\}}p_T (z_\mathrm{00}) \exp$ $(- (\rho z_\mathrm{00}^{ \alpha_T})/(C_T p_\seco g_{\mathrm{st}} (\beta_\mathrm{00} - \pi ) g_{\mathrm{pr}} (\omega_{\prim\mathrm{0}} - \beta_\mathrm{00})))$.
\end{lemma}

\begin{remark}
Note that the indicator for a network under omni-SLS is $U^{'}_{i, \omni} = \mathbbm{1} (p_\seco G_{i\mathrm{0}} \ell_{T_{\seco i \to \prim\mathrm{0}}} (z_{i\mathrm{0}}) <  \rho )$ with  
MAP $p^{'}_{\mathrm{m}i, \omni} = 1 - \sum_{T = \{\los,\nlos\}} p_T(z_{i\mathrm{0}})\exp (-(\rho z_{i\mathrm{0}}^{ \alpha_T})/(C_T p_\seco))$. 
\end{remark} 

The instantaneous SINR at the considered secondary receiver at the origin $\mathbf{o}$ is given as
\begin{align}
\SINR_{\seco0} &= \frac{ F_{\seco0} U^{'}_0 p_\seco \gst{0}\gsr{0} \ell_{T_{\seco\mathrm{0}}} (r_\seco)}{ \sigma^2 + I_\prim +  I_\seco \left(\Phi_\seco \right)}. 
\label{eq:Blockage:SinglePrimary:SINR:SecondaryLink}
\end{align}
where $\noise$ is the  noise power and $\ell_{T_{\seco\mathrm{0}}} (r_\seco) = \ell_{T_{\seco\mathrm{0} \to \seco\mathrm{0}}}$ is the path loss of typical secondary link with link-length $r_\seco$. $I_\prim$ and $I_\seco (\Phi_\seco)$ are the interferences at the typical secondary receiver due to the primary link and the rest of the active secondary transmitters, respectively and are given as
\begin{align*}
I_\prim &= p_\prim \bar{\kappa}_\mathrm{00} G^{'}_\mathrm{00} \ell_{T_{\prim\mathrm{0} \to \seco\mathrm{0}}} (x_{\prim\mathrm{0}})\\
I_\seco \left(\Phi_\seco \right) &= \sum\nolimits_{ \X_{\mathrm{s}i} \in \Phi_{\seco}/ \{\X_{\mathrm{s}0}\} } p_\seco \kappa_{i\mathrm{0}} U^{'}_{i} F^{'}_{i\mathrm{0}} \ell_{T_{\seco i \to \seco\mathrm{0}}} (x_{\seco i}),
\end{align*}
where $\bar{\kappa}_\mathrm{00} = \gsr{\delta_{\prim\mathrm{0}}}  \gpt{\delta_{\prim\mathrm{0}} - \pi - \omega_{\prim\mathrm{0}}}$ and $\kappa_{i\mathrm{0}} = \gpr{\theta_{\mathrm{s}i}} \gst{\theta_{\mathrm{s}i} - \pi - \omega_{\mathrm{s}i}}$. The coverage probability of the typical secondary link ($p_\mathrm{cs}(\SThres, \mu, \rho) = \mathbb{P} [ \SINR_{\seco\mathrm{0}} > \tau \lvert \, \rho, \mu] $) is given in the following theorem (See Appendix \ref{thrm:proof:Blockage:SecondaryCov} for proof). 

\begin{theorem}\label{Theorem:Blockage:SinglePrimary:SecondaryCov}
The coverage probability of a secondary link in the presence of other secondary links, while ensuring interference at the primary link below $\rho$, is given as
\begin{align}
&p_\mathrm{cs} (\tau, \mu, \rho) = 
\!\!\!\!\!\!\!
\sum_{T = \{\los,\nlos\}} 
\!\!\!\!
p_T (r_\seco) \, e^{\!\! - \frac{\tau \sigma^2}{A^T_\mathrm{0}}} \,
\big[1 -
\!\!
\sum_{T_1} 
p_{T_1} \! (z_\mathrm{00}) \, e^{- {1}/{A^{T_1}_\mathrm{10}}} \big] 
\bigg(\sum_{T_2}
\frac{p_{T_2} (x_{\prim\mathrm{0}})}{1 + \tau {A^{T_2}_\mathrm{20}}/{A^T_\mathrm{0}}}\bigg) 
\times \nonumber \\
&\quad
\exp \! \left( 
\!\!
- \lambda_\seco 
\!\!
\int_0^{2\pi} 
\!\!\!
\int_0^\infty 
\!
\sum\limits_{T_3}
p_{T_3} (x_\seco) \times
\mathbb{E}_{\omega_\seco}
\! 
\bigg[ \frac{1 \! - \!
\sum_{T_4} 
p_{T_4} \! (z_\mathrm{0}) \, e^{\!\! - A^{T_4}_\mathrm{s0} z_\mathrm{0}^{ \alpha_{T_4}}} }{1 + C^{T_3}_\mathrm{s} x_\seco^{\alpha_{T_3}} } \bigg] x_\seco \dd x_\seco \dd \theta_\seco 
\!
\right),
\label{eq:Blockage:SinglePrimary:SecondaryCov}
\end{align}
where $A^T_\mathrm{0} \! = \! p_\seco g_\mathrm{st}(0) g_\mathrm{sr}(0) C_T r_\seco^{-\alpha_T}$ denotes the secondary link serving power, $A^{T_1}_\mathrm{10} \! = \! p_\seco g_\mathrm{st} (\beta_\mathrm{00} \! - \! \pi) g_\mathrm{pr} (\omega_{\prim\mathrm{0}} \! - \! \beta_\mathrm{00}) C_{T_1} z_\mathrm{00}^{-\alpha_{T_1}}\!/\rho $ denotes the considered secondary-to-primary-link interference relative to $\rho$, $A^{T_2}_\mathrm{20} \! = \! p_\prim g_\mathrm{pt} (\delta_{\prim\mathrm{0}} \! - \! \pi \! - \! \omega_{\prim\mathrm{0}}) g_\mathrm{sr} (\delta_{\prim\mathrm{0}}) C_{T_2} x_{\prim\mathrm{0}}^{-\alpha_{T_2}}$ denotes the primary-to-considered-secondary-link interference, $A^{T_4}_\mathrm{s0} (\omega_\seco, \beta_\mathrm{0} (x_\seco,\theta_\seco)) \! = \! {\rho}/(p_\seco C_{T_4} g_\mathrm{st} (\beta_\mathrm{0}(x_\seco,\theta_\seco) \! - \! \omega_\seco) g_\mathrm{pr} (\omega_{\prim\mathrm{0}} \! - \! \beta_\mathrm{0} (x_\seco,\theta_\seco)))$ denotes the ratio between $\rho$ and the directional gain from the interfering secondary transmitter to the primary receiver and $C^{T_3}_\mathrm{s} (\omega_\seco,\theta_\seco) = {A^{T}_\mathrm{0}}/\SThres p_\seco C_{T_3} g_\mathrm{st} (\theta_\seco - \omega_\seco - \pi) g_\mathrm{sr} (\theta_\seco)$ with $T_n \! = \! \{\los, \nlos\}$ for $n = \{1\!-\!4\}$. $x_{\prim\mathrm{0}}$ and $\delta_{\prim\mathrm{0}}$ are the distance and the angle of primary transmitter with respect to $\Y_{\seco0}$. $z_\mathrm{00} \! = \! \| r_\seco \angle 0 \! - \! \Y_{\prim\mathrm{0}} \|$, $\beta_\mathrm{00} \! = \! \angle (r_{\seco}\angle 0 \! - \! \Y_{\prim\mathrm{0}})$,  $z_\mathrm{0} \! = \! \| x_\seco \angle \theta_\seco \! - \! \Y_{\prim\mathrm{0}} \|$, and $\beta_\mathrm{0} (x_\seco,\theta_\seco) \! = \! \angle (x_\seco\angle \theta_\seco \! - \! \Y_{\prim\mathrm{0}})$.
\end{theorem}

Note that the primary-protection-zone is no longer symmetric around the secondary receiver, unlike the case of primary coverage as computed in Theorem \ref{Theorem:Blockage:SinglePrimary:PrimaryCov}, and the presence of five LOS probability functions in \eqref{eq:Blockage:SinglePrimary:SecondaryCov} further hinders the simplification of expressions via symmetry arguments. However, we can derive the insights from \eqref{eq:Blockage:SinglePrimary:SecondaryCov} by dividing it into four terms as listed below: 
\begin{enumerate}[leftmargin=9pt]
\item[a)] $\mathsf{Term-1}: \sum_{T} p_T (r_\seco) \exp ( -{\tau \sigma^2}/{A^T_\mathrm{0}} )$ arises due to the noise ($\sigma^2$) and depends on the level of $\tau$ in comparison to the signal-to-noise ratio ($A^T_\mathrm{0}/\sigma^2$) of $\X_{\seco0}-\Y_{\seco0}$. Here, $A^T_\mathrm{0}$ with $T = \{\los,\nlos\}$ represents the signal strength of the LOS/NLOS $\X_{\seco0}-\Y_{\seco0}$ link, respectively.
\item[b)] $\mathsf{Term-2}: 1-\sum_{T_1} p_{T_1} (z_\mathrm{00}) (1 - e^{- 1/{A^{T_1}_\mathrm{10} } })$ with $T_1 = \{\los,\nlos\}$ represents the MAP of $\X_{\seco0}$ and depends on the strength of the received signal from a LOS/NLOS $\X_{\seco0}$ to $\Y_{\prim0}$ in terms of $A^{T_1}_\mathrm{10}$. 
\item[c)] $\mathsf{Term-3}: \sum_{T_2}( p_T (x_{\prim0}))/{(1 + \tau ({A^{T_2}_\mathrm{20}}/{A^T_\mathrm{0}}))}$ depends on the primary interference's level relative to the serving power of the considered secondary link. Here, $A^{T_2}_\mathrm{20}$ with $T_2 = \{\los,\nlos\}$ represents strength of received signal from $\X_{\prim0}$ to the LOS/NLOS $\Y_{\seco0}$.
\item[d)] $\mathsf{Term-4}:$ The integral term $I = \int_0^{2\pi} \mathbb{E}_{\omega_\seco} [\int_0^\infty \sum_{T_3} p_{T_3} (x_\seco)$ $[(\, 1 - \sum_{T_4} p_{T_4} (z_\mathrm{0}) e^{- A^{T_4}_\mathrm{s0} z_\mathrm{0}^{\alpha_{T_4}}})/(\, 1 + C^{T_3}_\mathrm{s} x_\seco^{\alpha_{T_3}})] x_\seco \dd x_\seco \, ] \dd \theta_\seco$ in which numerator and denominator, respectively, represent the MAP of $\X_{\seco0}$ and the interference power of the LOS/NLOS secondary transmitter located at $\X_\seco$.
\end{enumerate}

A mismatch occurs in signals corresponding to the cross-links $\X_{\seco0}-\Y_{\prim0}$ ($\mathsf{Term-2}$) and $\X_{\prim0}-\Y_{\seco0}$ ($\mathsf{Term-3}$) which dictate MAP $p_{\mathrm{m}0}$ of $\X_{\seco0}$ and the primary interference $I_\prim$ at $\Y_{\seco0}$, respectively. Even in the absence of blockages, the two links may face different directional gains; hence, the criterion of having received power of $\X_{\seco0}-\Y_{\prim0}$ being less than $\rho$ 
does not guarantee a low $I_\prim$. Specially, in the presence of blockages, this mismatch becomes more prominent due to the cases (a) when the first cross-link is NLOS while second one is  LOS resulting in a high MAP $p_{\mathrm{m}0}$ even when $I_\prim$ at $\Y_{\seco0}$ is high and (b) when the first cross-link is LOS while second one is blocked resulting in a low $p_{\mathrm{m}0}$ even when $I_\prim$ is low. If there is no mismatch, both cross-links will see the same gain {\em i.e.} $\rho A^{T_1}_\mathrm{10} = A^{T_2}_\mathrm{20}$, resulting in an appropriate control in $I_\prim$ by $p_{\mathrm{m}0}$. A similar mismatch occurs in $\mathsf{Term-4}$ due to the differences between the variables $C^{T_3}_\mathrm{s}$ and $A^{T_4}_\mathrm{s0}$, and their respective multipliers $x_\seco$ and $z_\mathrm{0}$. Blockages further intensify the mismatch in $\mathsf{Term-4}$ when (a) $\X_{\seco}-\Y_{\seco0}$ is blocked while $\X_{\seco}-\Y_{\prim0}$ is LOS resulting in a low secondary interference at $\Y_{\seco0}$ and a low MAP at $\X_\seco$, and (b) $\X_{\seco}-\Y_{\seco0}$ is LOS while the $\X_{\seco}-\Y_{\prim0}$ is NLOS resulting in a high secondary interference at $Y_{\seco0}$ and a high MAP for $\X_\seco$. 

We further simplified \eqref{eq:Blockage:SinglePrimary:SecondaryCov} for the ZBL, HBL, and NOL scenarios, as given in the following remark.

\begin{remark} \label{Remark:Blockage:SinglePrimary:SecondaryCov-ZBL-HBL-NOL}
Under the ZBL/HBL/NOL scenario, the secondary coverage is respectively given as 
\begin{align*}
&p_\mathrm{cs, ZBL}  = p_\mathrm{cs} (\tau, 0, \rho) \! = \! 
e^{\!\! - \frac{\tau \sigma^2}{A^\los_\mathrm{0}}} \bigg(\frac{1 \! - \! e^{- {1}/{A^\los_\mathrm{10}}}}{1 + \tau {A^\los_\mathrm{20}}/{A^\los_\mathrm{0}}}\bigg) 
\exp \! \bigg( 
\!\!
- \lambda_\seco 
\!\!
\int_0^{2\pi} 
\!\!\!
\int_0^\infty 
\!\!
\mathbb{E}_{\omega_\seco}
\! 
\bigg[ 
\!
\frac{1 \! - \! e^{- A^\los_\mathrm{s0} z_\mathrm{0}^{ \alpha_\los}}}{1 + C^\los_\mathrm{s} x_\seco^{\alpha_\los} } \bigg]  
\!
x_\seco \dd x_\seco \dd \theta_\seco 
\!
\bigg).\\
&p_\mathrm{cs, HBL} = p_\mathrm{cs} (\tau, \infty, \rho) \! = \! 
e^{\!\! - \frac{\tau \sigma^2}{A^\nlos_\mathrm{0}}} \bigg(\frac{1 \! - \! e^{- {1}/{A^\nlos_\mathrm{10}}}}{1 + \tau {A^\nlos_\mathrm{20}}/{A^\nlos_\mathrm{0}}}\bigg) 
\exp \! \bigg( 
\!\!
- \lambda_\seco 
\!\!
\int_0^{2\pi} 
\!\!\!
\int_0^\infty 
\!\!
\mathbb{E}_{\omega_\seco}
\! 
\bigg[ 
\!
\frac{1 \! - \! e^{- A^\nlos_\mathrm{s0} z_\mathrm{0}^{ \alpha_\nlos}}}{1 + C^\nlos_\mathrm{s} x_\seco^{\alpha_\nlos} } \bigg]  
\!
x_\seco \dd x_\seco \dd \theta_\seco 
\!
\bigg),\\
&p_\mathrm{cs, NOL} = \big[ \, p_\mathrm{cs} (\tau, \mu, \rho) \lvert C_\nlos = 0 \big] = p_\los (r_\seco) \, e^{ \! \!- \frac{\tau \sigma^2}{A^\los_\mathrm{0}}}  
\left( \! 1  \! -  \! p_\los (z_\mathrm{00}) \, e^{ \! \!- {1}/{A^\los_\mathrm{10}}}  \! \right)  
 \! \bigg(  \! \frac{ 1  \! +  \! \tau  \! \left({A^\los_\mathrm{20}}/{A^\los_\mathrm{0}}\right) p_\nlos (x_{\prim\mathrm{0}})}{1 + \tau \left({A^\los_\mathrm{20}}/{A^\los_\mathrm{0}}\right)} 
\!
\bigg) \times \nonumber \\
&\qquad\qquad\qquad
\exp \! \bigg( \!\!
- \lambda_\seco 
\!\!
\int_0^{2\pi} 
\!\!\!\!
\int_0^\infty 
\!\!\!\!\!
p_\los (x_\seco) \, \mathbb{E}_{\omega_\seco}
\! 
\bigg[ \frac{1 \! - \! p_\los (z_\mathrm{0}) e^{\!- A^\los_\mathrm{s0} z_\mathrm{0}^{ \alpha_\los}}}{1 + C^\los_\mathrm{s} x_\seco^{\alpha_\los} } \bigg]  
\!
x_\seco \dd x_\seco \dd \theta_\seco 
\!
\bigg)\!.
\end{align*}
\end{remark}

Since $z_\mathrm{0}$ and $x_\seco$ are not the same (although closely related), the simplification of the integral in \eqref{eq:Blockage:SinglePrimary:SecondaryCov} becomes difficult in the presence of blockages. The following two examples simplify $\mathsf{Term-4}$ under the NOL scenario by approximating its numerator to bring symmetry between $C^{T_3}_\mathrm{s}$ and $A^{T_4}_\mathrm{s0}$.

\begin{example}	
Under the assumption that considered secondary transmitter is very close to the primary receiver compared to the mean contact distance of the secondary network ( $= 1/(2\sqrt{\lambda_\seco})$), $z_0 \approx x_\seco$ and $g_\mathrm{st} (\beta_0 (x_\seco,\theta_\seco) - \omega_\seco) \approx g_\mathrm{st} (\theta_\seco - \pi - \omega_\seco)$ as all other secondary transmitters will be far away from the primary link resulting similar gain experience at the considered secondary and primary receivers from all other secondary transmitters. Then, ${A^\los_\mathrm{s0}}/{C^\los_\seco}$ is not a function of $\omega_\seco$ and $\mathsf{Term-4}$ in $p_\mathrm{cs, NOL}$ can be simplified as
\begin{align*}
&I_\mathrm{NOL, close} \! = \! \frac{e^{-p}}{\alpha_\los} 
\!
\sum_{n = 0}^\infty 
\!
\frac{(- \mu)^n}{n!} 
\!
\int_0^{2\pi} 
\!\!
\left(
\!
\pi \csc \! \bigg( \! \frac{n \! + \! 2}{\alpha_\los} \! \bigg)
\! + \!
\frac{2^n}{e^{p}} \, \Gamma \! 
\bigg(\! \frac{n \! + \! 2}{\alpha_\los} \! \bigg)
\times
\, U \! 
\left(\! \frac{n + 2}{\alpha_\los}, \frac{n + 2}{\alpha_\los}, \frac{A^\los_\mathrm{s0}}{C^\los_\seco} \! \right)
\!
\right) \times
\nonumber \\
&\qquad\qquad\qquad\qquad
\mathbb{E}_{\omega_\seco}
\! 
\left[ \! \left(C^\los_\seco\right)^{\!\! -\frac{(n + 2)}{\alpha_\los}} \! \right]  
\!
\dd \theta_\seco,
\end{align*}
where $U(a,b,z)$ is the Tricomi confluent hypergeometric function \cite{weisstein}.
\end{example} 

\begin{example}	
Under the assumption that considered secondary transmitter is far away from the primary receiver compared to the mean contact distance of the secondary network ( $= 1/(2\sqrt{\lambda_\seco})$), $y_{\prim0} \gg x_\seco$ where $y_{\prim0} = \lvert\lvert \Y_{\prim0}\rvert\rvert$. Thus $z_\mathrm{0} \approx y_{\prim0}$ and $\mathsf{Term-4}$ in $p_\mathrm{cs, NOL}$ can be simplified as 
\begin{align*}
&I_\mathrm{NOL, far} \! = \! 
{\textstyle\frac{\pi e^{-p}}{\alpha_\los} }
\!
\sum\nolimits_{n = 0}^\infty 
\!
{\textstyle \frac{(- \mu)^n}{n!}} \csc \! \left( \! \frac{(n+2)}{\alpha_\los} \! \right) 
\int_0^{2\pi} 
\!\!
\mathbb{E}_{\omega_\seco}
\! 
\Big[ 
\!
\big(1 \! - \! p_\los (y_{\prim\mathrm{0}}) \, e^{- A^\los_\mathrm{s0} y_{\prim\mathrm{0}}^{ \alpha_\los}}\big) 
\!
\left(C^\los_\seco\right)^{\!\! - \frac{(n + 2)}{\alpha_\los}}
\!
\Big] \dd \theta_\seco.
\end{align*}
\end{example} 

We simplify \eqref{eq:Blockage:SinglePrimary:SecondaryCov} under sectorized beam approximation to get Cor. \ref{corr:term4sec} (for proof, see the supplementary \cite{TripGupTheoremFile2025}). 

\begin{corr}\label{corr:term4sec}
Under the sectorized beam approximation, the $\mathsf{Term-4}$ in \eqref{eq:Blockage:SinglePrimary:SecondaryCov} can be simplified as
\begin{align*}
I_\mathrm{sec} \! &= \!\!\! 
\int_0^\infty 
\!\!
\sum_{k = 1}^{4} 
\!
\int_{\!\mathcal{F}_k} 
\!
\sum_{i = 1}^{4} q_i (\delta_\seco(x_\seco,\theta_\seco)) \,
\mathbb{E}_{\omega_\seco} 
\!
\left[ \sum_{T_3} p_{T_3} (x_\seco) 
\sum_{T_4} p_{T_4} (z_\mathrm{0}) \times
\right.
\\
&\qquad\qquad\qquad\qquad\qquad
\left.
\bigg[ 
\!
\frac{1 \! - \! \exp\left( \! - (\rho z_\mathrm{0}^{ \alpha_{T_4}})/(p_\seco C_{T_4} \mathcal{A}_i \mathcal{C}_k)\right)}{1 \! + \! (1/\tau) ({C_T}/{C_{T_3}}) ({x_\seco^{\alpha_{T_3}} }/{r_{\seco}^{ \alpha_T}}) ({(\ast \asr)}/{(\mathcal{B}_i \mathcal{D}_k)})} 
\!
\bigg]  
\! 
\right] 
\!
x_\seco \dd x_\seco \dd  \theta_\seco, \!\!
\end{align*}
where $\delta_\seco (x_\seco,\theta_\seco) = \theta_\seco - \beta_0 (x_\seco,\theta_\seco)$. Here, $\mathcal{A}_i$ and $\mathcal{B}_i$ represent transmit gains from interfering secondary transmitter to the primary receiver and the interfering secondary transmitter to considered secondary receiver with probability $q_i (\delta_\seco (x_\seco,\theta_\seco))$. Similarly, $\mathcal{C}_k$ and $\mathcal{D}_k$ represent gains from primary receiver to interfering secondary transmitter and considered secondary receiver to interfering secondary transmitter corresponding to events $\mathcal{F}_k$. These terms are given as \\ \\
{\hspace*{2cm}\noindent
\resizebox{0.75\textwidth}{!}{
\begin{minipage}{.6\linewidth}
\centering
\begin{tabular}{|c|c|c|}
\hline
Probability		&  $\mathcal{A}_i$ &  $\mathcal{B}_i$ \\ \hline
$q_\mathrm{1} (\delta_\seco(x_\seco,\theta_\seco))$	&  $\mathcal{A}_\mathrm{1} = \ast$ &  $\mathcal{B}_\mathrm{1} = \ast$ \\ \hline
$q_\mathrm{2} (\delta_\seco(x_\seco,\theta_\seco))$	&  $\mathcal{A}_\mathrm{2} = \ast$ &  $\mathcal{B}_\mathrm{2} = \bst$ \\ \hline
$q_\mathrm{3} (\delta_\seco(x_\seco,\theta_\seco))$	&  $\mathcal{A}_\mathrm{3} = \bst$ &  $\mathcal{B}_\mathrm{3} = \ast$ \\ \hline
$q_\mathrm{4} (\delta_\seco(x_\seco,\theta_\seco))$ &  $\mathcal{A}_\mathrm{4} = \bst$ &  $\mathcal{B}_\mathrm{4} = \bst$ \\ \hline
\end{tabular}
\end{minipage} 
\hspace*{.5cm}
\begin{minipage}{.4\linewidth}
\centering
\begin{tabular}{|c|c|c|}
\hline
$\mathcal{C}_k$	&  $\mathcal{D}_k$	&  $\mathcal{F}_k$	 										\\ \hline
$\apr$			&  $\asr$ 			&  $E_\mathrm{1} \cap E_\mathrm{2}$							\\ \hline
$\apr$			&  $\bsr$			&  $E_\mathrm{1} \cap E_\mathrm{2}^\mathrm{c}$ 				\\ \hline
$\bpr$			&  $\asr$			&  $E_\mathrm{1}^\mathrm{c} \cap E_\mathrm{2}$				\\ \hline
$\bpr$			&  $\bsr$			&  $E_\mathrm{1}^\mathrm{c} \cap E_\mathrm{2}^\mathrm{c}$	\\ \hline
\end{tabular}
\end{minipage}}}\\ \\
\noindent with $E_\mathrm{1} = \{ \theta_\seco :  \theta_\seco - \omega_{\seco0} \in ( -{\phi_{\mathrm{sr}}}/{2}, {\phi_{\mathrm{sr}}}/{2} ) \}$ and $E_\mathrm{2} = \{ \theta_\seco :  \beta_0 (x_\seco,\theta_\seco) - \omega_{\prim\mathrm{0}} \in ( - {\phi_{\mathrm{pr}}}/{2}, {\phi_{\mathrm{pr}}}/{2} ) \}$.
\end{corr}

We can further simplify \eqref{eq:Blockage:SinglePrimary:SecondaryCov} under omni-SLS as all terms are no longer the functions of $\omega_\seco$ or $\theta_\seco$) to get the following result.

\begin{remark}
The secondary coverage $p_\mathrm{cs, omni}(\SThres, \rho)$ for a network under omni-SLS is given as 
\begin{align*}
&\!\! p_\mathrm{cs, omni}(\SThres, \mu, \rho) \! = \!\! \sum_{T} p_T (r_\seco) \, e^{\!\! - \frac{\tau \sigma^2 r_{\seco}^{ \alpha_T}}{ C_T p_\seco}} 
\!
\left[\!  1 \!\! -
\!\!
\sum_{T_1} p_{T_1} \! (z_\mathrm{00}) \, e^{\!\! - \frac{\rho z_\mathrm{00}^{\alpha_{T_1}}}{C_{T_1} p_\seco}} 
\!
\right]
\bigg(
\!
\sum_{T_2} 
\!
\frac{p_{T_2} (x_{\prim\mathrm{0}})}{1 \! + \! \tau ({C_{T_2}}/{C_T}) ({r_{\seco}^{ \alpha_T}}/{x_{\prim\mathrm{0}}^{ \alpha_{T_2}}}) ({p_\prim}/{p_\seco})}
\!
\bigg) \\ 
&\qquad\qquad\qquad\qquad
\times \ \exp \left( 
\!\!
- 2\pi \lambda_\seco 
\!\!
\int_0^\infty 
\! 
\sum\limits_{T_3} p_{T_3} (x_\seco) \,
\!
\bigg[
\! 
\frac{1 \! - \! \sum_{T_4} p_{T_4} (z_\mathrm{0}) e^{\!\! - (\rho/p_\seco)/C_{T_4} z_\mathrm{0}^{-\alpha_{T_4}}} }{1 + (1/\tau) ({C_T}/{C_{T_3}}) ({x_\seco^{\alpha_{T_3}} }/{r_{\seco}^{ \alpha_T}})} 
\bigg]  
\!
x_\seco \dd x_\seco 
\!
\right).
\end{align*}
\end{remark}

\subsubsection{Coverage of the typical secondary link} 
The typical secondary link refers to the secondary link picked randomly out of all secondary links in the region of interest $\mathcal{R} = \Phi_\seco \cap \Ball(\mathbf{o}, R)$. The reason behind considering a finite value of $R$ is to observe the effect of distance from primary on the performance of the secondary user. When we consider the average user in the complete unbounded 2D space, the effect of primary vanishes. Note that $p_\mathrm{cs}(\SThres, \mu, \rho)$ in \eqref{eq:Blockage:SinglePrimary:SecondaryCov} denotes the coverage of the secondary link on the primary link location in terms of its relative location/orientation ($x_{\prim 0}$, $\delta_{\prim 0}$, $\omega_{\prim 0}$). Hence, the typical secondary coverage is given as
\begin{align}
p_\mathrm{cs}^\mathrm{avg}&(\SThres) = 
\frac{1}{\pi R^2} \int_0^R
\!\!
\int_0^{2\pi} \mathbb{E}_{\omega_\prim} \big[p_\mathrm{cs}(\SThres, \mu, \rho)\big] x_\prim \dd x_\prim \dd \delta_\prim.
\label{eq:Blockage:SinglePrimary:SecondaryCov:Average}
\end{align}

\section{Numerical Results and Insights} \label{Section:Blockage:SinglePrimary:NumericalResults}
This section presents numerical investigations to validate derived results and derive insights. The default values for various parameters are given in Table-\ref{table:ParameterTabel}. We consider a simulation radius $R_\mathrm{sim} = 4000$ m. We consider the secondary's density of $\lambda_\seco = 8 \times 10^{-5}$ SBSs/$\text{m}^2$ corresponding to a mutual distance of 112 m between neighbouring points. We define a parameter $L_\mu = 1/\mu$ representing the average LOS distance where $L_\mu \to \infty$ and $L_\mu \to 0$ denote the ZBL (LOS-only) and HBL (NLOS-only) scenarios. Also note that we have considered $M = (M_\prim,\,M_\seco) = (4,\,4)$ with $M_\mathrm{pt} = M_\mathrm{pr} = M_\prim$ and $M_\mathrm{st} = M_\mathrm{sr} = M_\seco$ until mentioned otherwise.
\begin{table}[ht!]
\centering
\caption{Parameters for numerical evaluations}
\label{table:ParameterTabel}
\resizebox{0.45\textwidth}{!}{
\begin{tabular}{|p{0.5in}|p{0.85in}|p{0.5in}|p{1in}|}
\hline
\textbf{Parameters} & \textbf{Numerical value} &\textbf{Parameters} & \textbf{Numerical value} \\
\hline
$f$,\, $\mathrm{BW}$ & $60$ GHz, $200$ MHz  &$p_\prim $, $p_\seco$ & 27 dBm, 17 dBm \\
\hline
$r_\prim$,$r_\seco$ & $50$ m, $20$ m  &	$\alpha_\los$, $\alpha_\nlos$ & $2.4$, $4.2$ \\
\hline
$C_\los$ , $C_\nlos$  & $-60$ dB, $-70$ dB & $\sigma^2$ & $7.9621 \times 10^{-13}$ Watts \\
\hline
$\kappa$ & $121^{\circ}$ &	$\phi$ & ${\kappa}/{M}$ when $M > 1$ \\
\hline
\end{tabular}}
\end{table}
To demonstrate the impact of relative distance and orientation of the primary on the coverage of concerned secondary link, we consider three specific configurations of $\textbf{TN} = [\angle{\delta_\prim} \ x_\prim \ \angle{\omega_\prim}]$, termed \textbf{Type 1 (T1)}, \textbf{Type 2 (T2)}, and \textbf{Type 3 (T3)} (see Fig. \ref{fig:Blockage:SinglePrimary:SystemModel-Secondary}(b)) representing a minimal, a moderate and a high mutual effect of interferences between primary to secondary cross-links. The typical secondary link is denoted by \textbf{Type 4 (T4)}. Its distance $x_\prim$ from the primary receiver (taken at the origin) is generated as $x_\prim = \sqrt{u_\prim}$ where $u_\prim = \mathcal{U}(0, R^2)$. Further, $\{\delta_\prim, \omega_\prim\} = \mathcal{U} (0,2\pi)$. To avoid the simulation edge effect for secondary interference due to finite $R_\mathrm{sim}$, we consider $R = R_\mathrm{sim}/2$ as the radius of the region of interest. 

\subsection{MAP of secondary links} \label{Section:Blockage:SinglePrimary:NumericalResults:MAP}
Fig. \ref{numres:Blockage:SinglePrimary:MAP} shows the impact of the location on the MAP of a secondary transmitter. As shown in \eqref{eq:Blockage:SinglePrimary:MAP}, we can observe that the MAP depends on their locations, LOS states and directionality. Comparing  Fig. \ref{numres:Blockage:SinglePrimary:MAP}(a) and (b), we see that reducing $\rho$ decreases secondary transmission opportunities. Further, we can observe that the number of active secondary transmitters increases as $L_\mu$ changes from $\infty$ (LOS-only) to $0$ (NLOS-only). 
\begin{figure}[!ht]
\centering
{\includegraphics[trim = 57 3 48 5, clip, scale=0.2999]{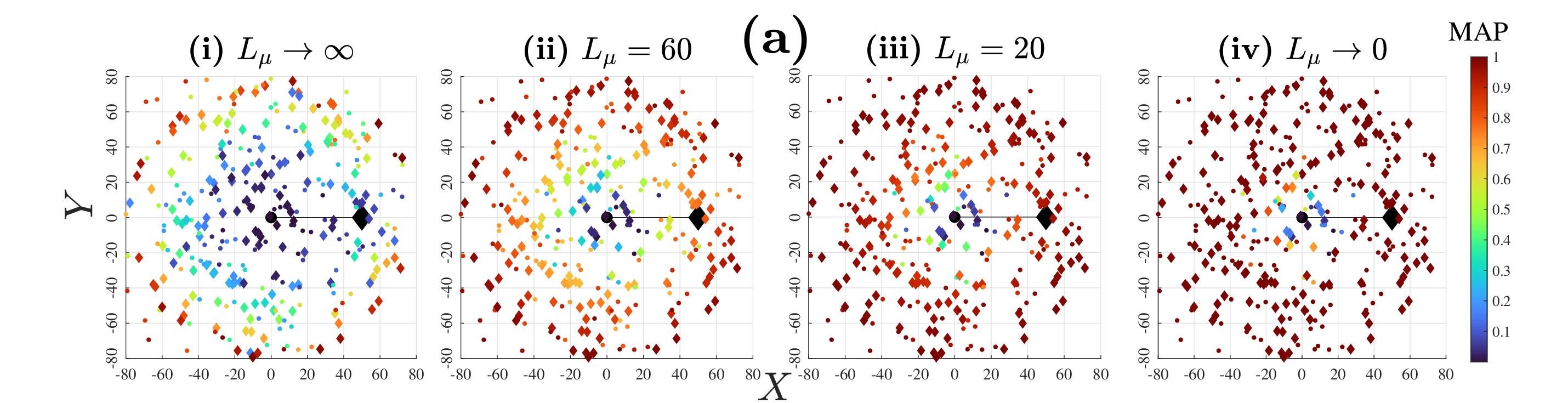}\label{numres:MAP_M_4_Rho_1eMinus14}}
{\includegraphics[trim = 57 3 48 5, clip, scale=0.2999]{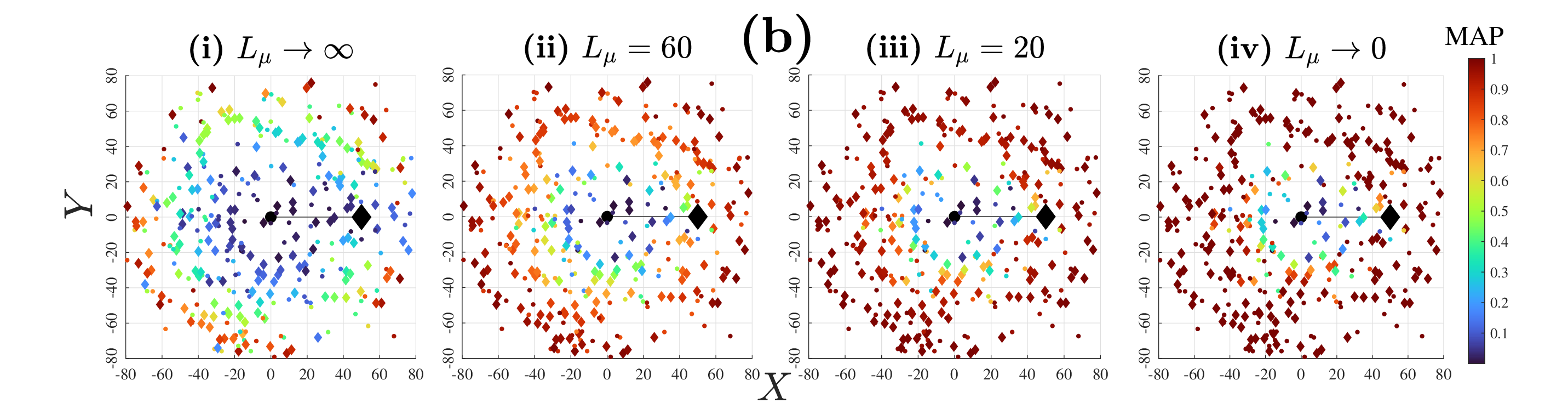}\label{numres:MAP_M_4_Rho_1eMinus15}}
\caption{The spatial variation of $p_{\mathrm{m}i}$ with $\lambda_\seco = 8 \times 10^{-3}$ /$\text{m}^2$ and $M = (4, \, 4)$ for (a) $\rho = 10$ femto-Watts and (b) $\rho = 1$ femto-Watts. Both ends of the primary and secondary links have antennas and each $\bullet\!\!\!-\!\!\!-\!\!\!\blacklozenge$ represents a primary receiver-transmitter pair.}
\label{numres:Blockage:SinglePrimary:MAP}
\vspace{-.15in}
\end{figure}

\subsection{AF of the secondary network} \label{Section:Blockage:SinglePrimary:NumericalResults:AF}
We now investigate the AF of the secondary network and the impact of various system parameters on it.
\subsubsection{Impact of spatially-aware SLS}  
\label{Section:Blockage:SinglePrimary:NumericalResults:AF:DirectionalityBlockage}
Fig. \ref{numres:Blockage:SinglePrimary:AFvs1byMu} shows the variation of secondary AF with average LOS distance $L_\mu$ for omni and directional cases. We can observe that AF improves with larger $M$ (narrower antenna beamwidth), higher $\rho$ (milder restriction on secondary transmissions) and smaller $L_\mu$ (higher blockage losses). Specifically, when $L_\mu \to 0$ (HBL scenario), we can observe a significant improvement in secondary AF with directionality for a given value of $\rho$. Fig. \ref{numres:Blockage:SinglePrimary:AFvsRatio}(a) shows the variation of secondary AF with $\rho/p_\seco$ for different severities of blockages. Here, we can observe that $\eta_\seco \to 1$ requires an approximate shift of $-60$ dB in $\rho/p_\seco$ when $L_\mu$ switches from $\infty$ (LOS only) to $0$ (NLOS only). Thus, the presence of blockages permits the use of stricter $\rho$ without affecting AF. 
\begin{figure}[ht!]
\vspace{-.20in}
\centering
{\includegraphics[trim = 20 12.5 25 7, clip, scale=0.3899]{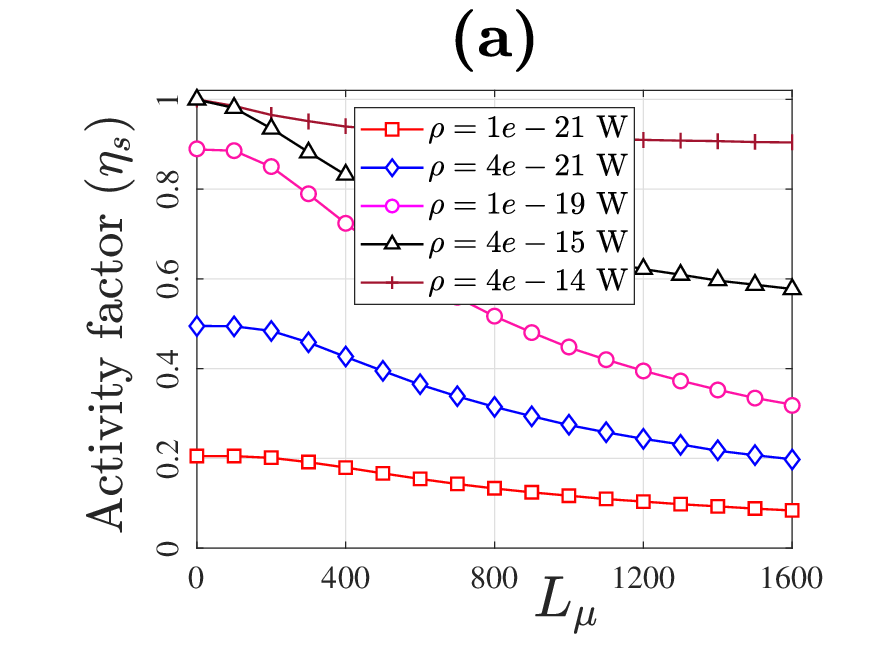}\label{numres:AFvs1byMu_Rho_Alpha_2_4_R400_M_1}}
{\includegraphics[trim = 74 12.5 27 7, clip, scale=0.3899]{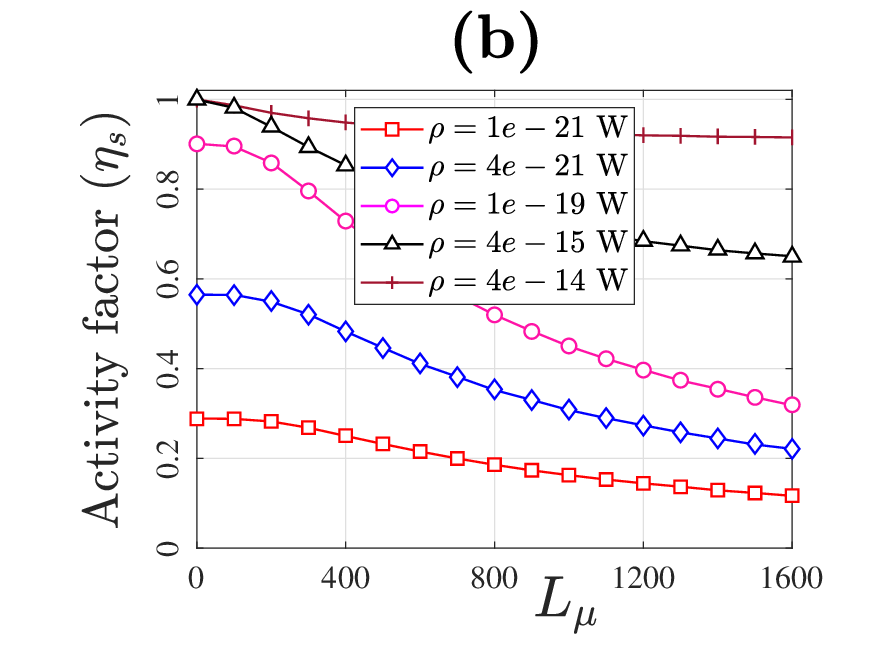}\label{numres:AFvs1byMu_Rho_Alpha_2_4_R400_M_4}}
\caption{The variation of secondary AF $\eta_\seco$ with average LOS distance $L_\mu = 1/\mu$ (in meters) for different values of interference-threshold $\rho$ with (a) $M = (1,1)$ and (b)  $M = (4,4)$. Here, radius of region of interest $R = 1000$ m. }
\label{numres:Blockage:SinglePrimary:AFvs1byMu}
\vspace{-.2in}
\end{figure}
\begin{figure}[ht!]
\vspace{-.2in}
\centering
{\includegraphics[trim = 20 12.5 27 9.5, clip, scale=0.3899]{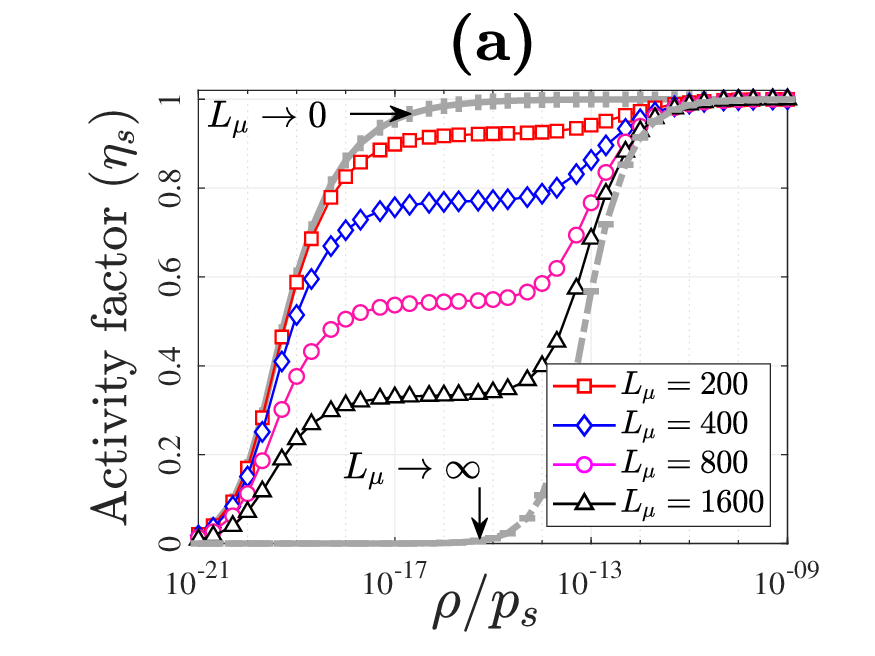} \label{numres:AF_1byMu_R_1000} }
{\includegraphics[trim = 76 12.5 27 9.5, clip, scale=0.3899]{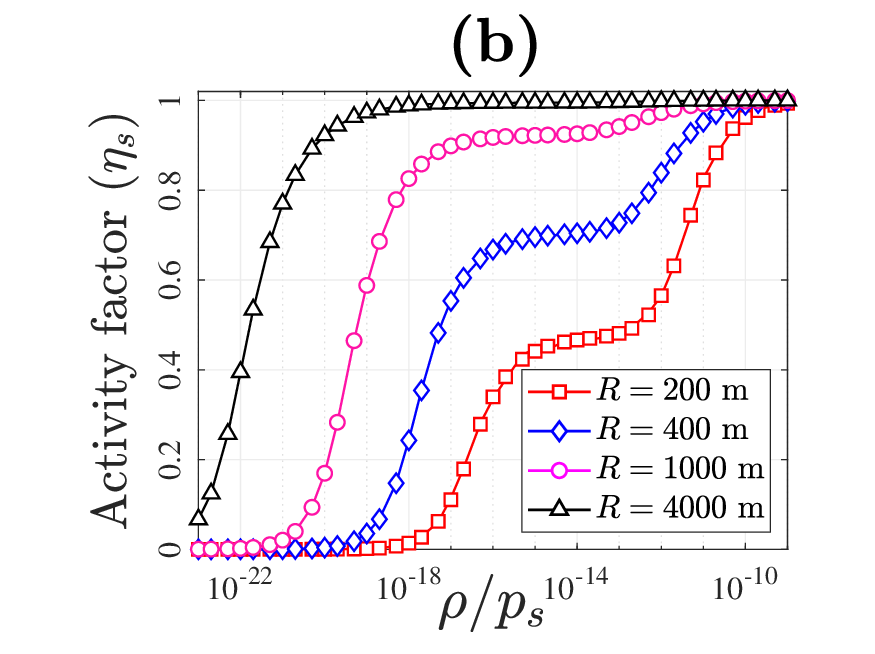} \label{numres:AF_R_1byMu_200} }
\caption{The variation of secondary AF $\eta_\seco$ with the interference-threshold $\rho$ for $M = (4,4)$ by varying (a) the value of $L_\mu$ (in m) 
with $R = 1000$ m and (b) the values of region-of-interest radius $R$ with $L_\mu = 200$ m.}
\label{numres:Blockage:SinglePrimary:AFvsRatio}
\vspace{-.15in}
\end{figure}
\subsubsection{Fundamental impact of blockages}
Fig. \ref{numres:Blockage:SinglePrimary:AFvsRatio}(a) shows that the presence of both LOS/NLOS links changes the fundamental shape of AF, compared to LOS-only or NLOS case, due to dual path-loss. For example, a plateau (flat) section surrounded by two falls is visible in the middle of the curves for cases with non-zero finite $L_\mu$ in Fig. \ref{numres:Blockage:SinglePrimary:AFvsRatio}(a). This section signifies a constant secondary activity irrespective of the changing $\rho$. The variation in the left of these curves with $\rho/p_\seco$ ({\em i.e.,} lower values of $\rho$) denotes the improvement in the activity of NLOS secondary links unless almost all NLOS links become active. Since LOS links have significantly higher power, they remain inactive even with an increase in $\rho$, thus resulting in a plateau region. Beyond a certain value of $\rho$, the LOS links also start becoming active with the increase in $\rho$, which represents the right side of the curves. Since the fraction of NLOS links decreases with $L_\mu$, the plateau height decreases with $L_\mu$. Note that this region appears only when there is a mixture of LOS and NLOS links. It highlights the importance of studying the effect of the dual path-loss function resulting from blockage compared to that involving the single path-loss function, representing either the LOS-only or the NLOS-only cases.
\subsubsection{Impact of secondary link proximity from the primary receiver} 
\label{Section:Blockage:SinglePrimary:NumericalResults:AF:Proximity} 
To demonstrate the impact of proximity to primary on the secondary link's performance, Fig. \ref{numres:Blockage:SinglePrimary:AFvsRatio}(b) shows the variation of secondary AF with $\rho/p_\seco$ for different region-of-interest radius $R$. We can observe that $\eta_\seco$ reaches $1$ at a smaller threshold $\rho$ as $R$ increases. It shows that the impact of primary on the secondary links diminishes quickly with its distance from these links, which is also consistent with \eqref{eq:Blockage:SinglePrimary:AF}. 

\subsection{Role of spatially-aware SLS}\label{Section:Blockage:SinglePrimary:NumericalResults:ThresholdImpact}
Figs. \ref{numres:NPvsRhoset4_alpha_2point4_4point2_R_4000_1byMus_Inf_0_tau_Minus10_5_dBs} and \ref{numres:NPvsRhoset4_alpha_2point4_4point2_R_4000_1byMus_200_50_tau_Minus10_5_dBs} show the effect of the various design/system parameters ($\rho$, $M$, $L_\mu$) of spatially-aware SLS on the variation of primary and secondary coverage. For a given $\tau$, we can observe three interesting trends.
\begin{figure}[ht!]
\vspace{-.18in}
\centering 
{\includegraphics[trim = 20 7 35 9.5, clip, scale=0.3899]{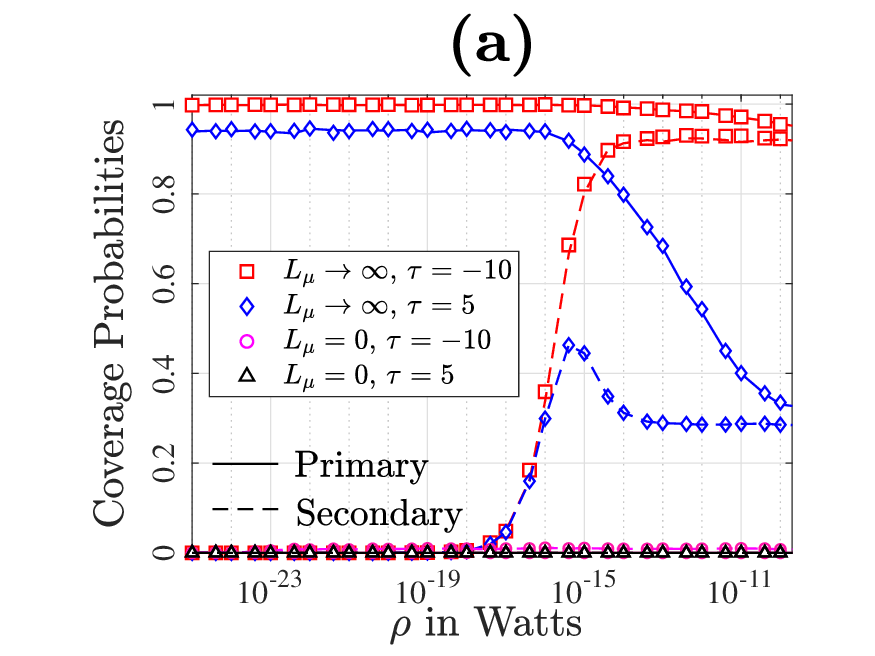}\label{numres:NPvsRhoset4_alpha_2point4_4point2_R_4000_1byMus_Inf_0_M_1_tau_Minus10_5_dBs}}
{\includegraphics[trim = 72 7 35 9.5, clip, scale=0.3899]{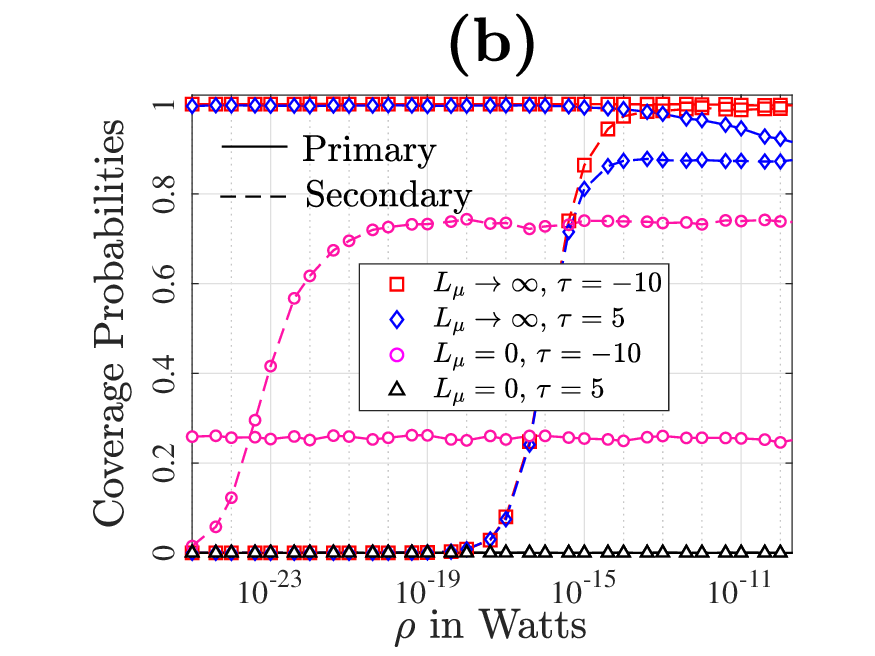} \label{numres:NPvsRhoset4_alpha_2point4_4point2_R_4000_1byMus_Inf_0_M_4_tau_Minus10_5_dBs} }
\caption{The variation of primary and secondary coverage probabilities with interference-threshold $\rho$ for different combinations of $L_\mu$ (in m) and $\tau$ (in dB) for (a) $M = (1,1)$ and (b) $M = (4,4)$ for T4 (the average secondary user).}
\label{numres:NPvsRhoset4_alpha_2point4_4point2_R_4000_1byMus_Inf_0_tau_Minus10_5_dBs}
\vspace{-.2in}
\end{figure}
\begin{figure}[ht!]
\vspace{-.2in}
\centering 
{\includegraphics[trim = 20 7 35 9.5, clip, scale=0.3899]{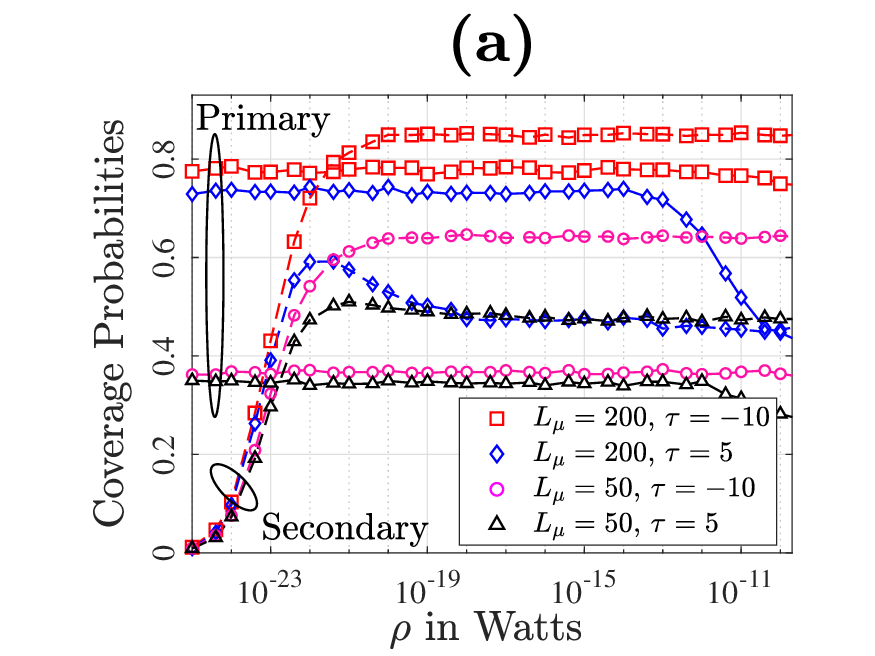}\label{numres:NPvsRhoset4_alpha_2point4_4point2_R_4000_1byMus_200_50_M_1_tau_Minus10_5_dBs}}
{\includegraphics[trim = 72 7 35 9.5, clip, scale=0.3899]{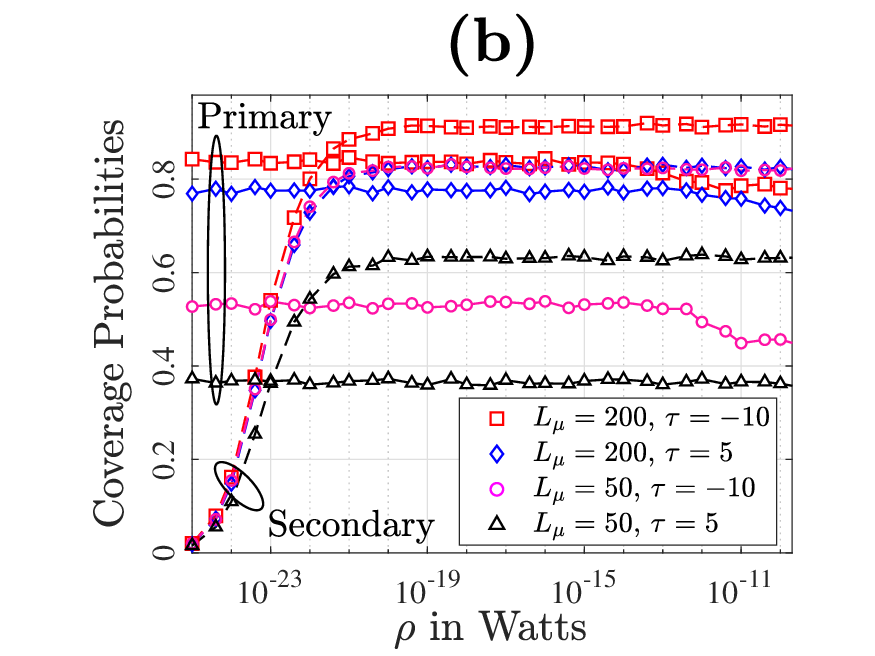} \label{numres:NPvsRhoset4_alpha_2point4_4point2_R_4000_1byMus_200_50_M_4_tau_Minus10_5_dBs}}
\caption{The variation of primary and secondary coverage probabilities with $\rho$ for different combinations of $L_\mu$ (in m) and $\tau$ (in dB) for (a) $M = (1,1)$ and (b) $M = (4,4)$ for T4 (the average secondary user).}
\label{numres:NPvsRhoset4_alpha_2point4_4point2_R_4000_1byMus_200_50_tau_Minus10_5_dBs}
\vspace{-.15in}
\end{figure}
\subsubsection{Role of directionality} 
\label{Section:Blockage:SinglePrimary:NumericalResults:ThresholdSelection:Directionality}
We observe that higher directionality plays a decisive role in improving the primary and secondary coverage performances. Along with the role of $\rho$ in limiting the secondary interference, we observe that if given enough directionality, a suitable value of $\rho$ can be obtained which not only preserve the primary link’s performance but also guarantee a reasonable secondary coverage.
\subsubsection{Role of blockage} 
\label{Section:Blockage:SinglePrimary:NumericalResults:ThresholdSelection:Blockage}
We observe that a larger average LOS distance $L_\mu$, in general, results in better primary as well as secondary coverage (within an appropriate range of interference-threshold $\rho$). Further, an increase in $L_\mu$ ({\em i.e.,} an increase in LOS probability) not only improves the serving signal, but can also increase interference, leading to a trade-off. However, as shown in Example \ref{Example:Blockage:SinglePrimary:MAP}, due to transmit restriction, the secondary interference is restricted from increasing when $L_\mu$ is increased, leading to overall improvement in primary coverage with $L_\mu$. Since such protection is not available for secondary links, it may be possible that an increase in $L_\mu$ can reduce the secondary coverage due to increased interference. This is especially evident at high values of $\rho$ where secondary activity is large. This behaviour can be seen in Figs. \ref{numres:NPvsRhoset4_alpha_2point4_4point2_R_4000_1byMus_Inf_0_tau_Minus10_5_dBs}(a) and \ref{numres:NPvsRhoset4_alpha_2point4_4point2_R_4000_1byMus_200_50_tau_Minus10_5_dBs}(a) while moving from $L_\mu = 200$ m to $L_\mu = \infty$ for the case with $\tau = 5$ dB and $M = 1$. However, directionality can help reduce the interference, resulting in a performance increase with $L_\mu$ for $M = 4$. Further, at the low value of $\rho$, where secondary activity is not large, the secondary interference is not a crucial concern. For example, in Fig. \ref{numres:NPvsRhoset4_alpha_2point4_4point2_R_4000_1byMus_Inf_0_tau_Minus10_5_dBs}, observe the LOS only scenario ($L_\mu \to \infty$) where secondary performance degraded due to low activity. Thus, blockages are usually beneficial for secondary links.
\subsubsection{Role of $\rho$} 
\label{Section:Blockage:SinglePrimary:NumericalResults:ThresholdSelection:Rho}
Usually, an increase in $\rho$ provides more transmission opportunities to secondary links, leading to an increase in secondary coverage. However, for some scenarios (see  Figs. \ref{numres:NPvsRhoset4_alpha_2point4_4point2_R_4000_1byMus_Inf_0_tau_Minus10_5_dBs}(a) and \ref{numres:NPvsRhoset4_alpha_2point4_4point2_R_4000_1byMus_200_50_tau_Minus10_5_dBs}(a)), it can also increase the secondary interference leading to an overall reduction in secondary coverage, especially at the high value of SINR threshold $\tau$. However, higher antenna directionality $M$ can easily mitigate this reduction in coverage (see Figs. \ref{numres:NPvsRhoset4_alpha_2point4_4point2_R_4000_1byMus_Inf_0_tau_Minus10_5_dBs}(b) and \ref{numres:NPvsRhoset4_alpha_2point4_4point2_R_4000_1byMus_200_50_tau_Minus10_5_dBs}(b)). A low $\rho$ ensures better primary coverage but restricts more secondary devices, whereas a high $\rho$ allows higher secondary MAPs at the expense of degrading the primary performance. Also, note that an increased secondary activity beyond a certain value can also reduce the secondary performance due to increased secondary interference.

\subsection{Design of spatially-aware SLS: Selection of $\rho$}\label{Section:Blockage:SinglePrimary:NumericalResults:ThresholdSelection}
As discussed above, an optimal value of $\rho$ (let us denote it by $\rho^\dagger$) is required to balance the trade-off that allows an appropriate number of active secondary devices without affecting the primary performance significantly. In absence of blockages, the way of finding $\rho^\dagger$ is to fix $p_{\mathrm{cp}}$ and $p_{\mathrm{cs}}$ and select the minimum value of $\rho$ as  $\rho^\dagger$ that allows this. Due to monotonic behaviour of $p_{\mathrm{cp}}$ and $p_{\mathrm{cs}}$ with $\rho$, it occurs when $p_{\mathrm{cs}}$) saturates. However, this method is no longer efficient in the presence of blockages because of the presence of a peak in secondary's performance before reaching saturation. For example, Figs. \ref{numres:NPvsRhoset4_alpha_2point4_4point2_R_4000_1byMus_Inf_0_tau_Minus10_5_dBs}(a) and \ref{numres:NPvsRhoset4_alpha_2point4_4point2_R_4000_1byMus_200_50_tau_Minus10_5_dBs}(a)) show a sharp degradation in $p_{\mathrm{cs}}$ {\em e.g.} when $L_\mu \to \infty$ and $L_\mu = 200$ m, with further increase in $\rho$ even when the value of $p_{\mathrm{cp}}$  is fixed. 
\begin{table}[!b]
\vspace*{0.068in}
\centering
\caption{The appropriate value of $\rho^\dagger$ satisfying $\{p_{\mathrm{cp}} > 70 \%, p_{\mathrm{cs}} > 50 \%\}$ criteria  for the different secondary links.
Prefixes p, f, a and y stand for pico ($10^{-12}$), femto ($10^{-15}$), atto ($10^{-18}$) and yocto ($10^{-24}$), respectively.}
\label{Table:Blockage:SinglePrimary:SuitableRho}
\resizebox{0.55\textwidth}{!}{
\vspace*{-1.0in}
\begin{tabular}{|c|c|c|cccc|}
\hline
\multirow{2}{*}{Type}  & \multirow{2}{*}{$L_\mu$} & \multirow{2}{*}{$\tau^\star$ (in dB)} & \multicolumn{4}{c|}{$\rho^\dagger$} \\ \cline{4-7} 
&&& \multicolumn{1}{c|}{$M_{uv} = 1$} & \multicolumn{1}{c|}{$M_{uv} = 2$} & \multicolumn{1}{c|}{$M_{uv} = 4$} & $M_{uv} = 8$ \\ \hline
\multirow{4}{*}{$T1$}  & $\infty$ & $\tau \leq -7$ & \multicolumn{1}{c|}{$0.12$ pW} & \multicolumn{1}{c|}{$45.60$ fW} & \multicolumn{1}{c|}{$34.90$ fW}   & $31$ fW  \\ \cline{2-7}   
& $200$ m & $\tau \leq -3$ & \multicolumn{1}{c|}{$0.19$ pW} & \multicolumn{1}{c|}{$37.10$ fW}   & \multicolumn{1}{c|}{$25.10$ fW}   & $21.70$ fW  \\ \cline{2-7} 
&$50$ m & $\tau \leq -15$ & \multicolumn{1}{c|}{-} & \multicolumn{1}{c|}{-}   & \multicolumn{1}{c|}{$3.88$ aW}   & $3.02$ aW  \\ \cline{2-7} 
& $0$ & $\tau \leq -10$ & \multicolumn{1}{c|}{-} & \multicolumn{1}{c|}{-} & \multicolumn{1}{c|}{-}  & $1.87$ aW \\ \hline
\multirow{4}{*}{$T2$} & $\infty$ & $\tau \leq -3$ & \multicolumn{1}{c|}{$0.72$ pW}   & \multicolumn{1}{c|}{$0.25$ pW}   & \multicolumn{1}{c|}{$0.18$ pW}   & $0.15$ pW  \\ \cline{2-7} 
& $200$ m & $\tau \leq -3$ & \multicolumn{1}{c|}{$0.47$ pW}   & \multicolumn{1}{c|}{$0.21$ pW}   & \multicolumn{1}{c|}{$0.15$ pW}   & $0.12$ pW \\ \cline{2-7} 
&$50$ m & $\tau \leq -15$ & \multicolumn{1}{c|}{-}   & \multicolumn{1}{c|}{-}   & \multicolumn{1}{c|}{$27.40$ fW}   & $9.40$ fW  \\ \cline{2-7} 
& $0$ & $\tau \leq -10$ & \multicolumn{1}{c|}{-}   & \multicolumn{1}{c|}{-}   & \multicolumn{1}{c|}{-}   & \multicolumn{1}{c|}{$36.39$ aW}    \\ \hline
\multirow{4}{*}{$T3$} & $\infty$ & $\tau \leq -10$ & \multicolumn{1}{c|}{$86.20$ fW} & \multicolumn{1}{c|}{$0.67$ pW} & \multicolumn{1}{c|}{$55.60$ fW} & $41.10$ fW \\ \cline{2-7} 
& $200$ m & $\tau \leq -10$ & \multicolumn{1}{c|}{$63.10$ fW}   & \multicolumn{1}{c|}{$2.07$ pW}   & \multicolumn{1}{c|}{$42.70$ fW}   & $28.90$ fW  \\ \cline{2-7} 
& $50$ m & $\tau \leq -15$ & \multicolumn{1}{c|}{-}   & \multicolumn{1}{c|}{-}   & \multicolumn{1}{c|}{$4.14$ fW}   & $16$ aW  \\ \cline{2-7} 
& $0$ & $\tau \leq -10$ & \multicolumn{1}{c|}{-}   & \multicolumn{1}{c|}{-}   & \multicolumn{1}{c|}{-}   & $3.09$ aW \\ \hline
\multirow{4}{*}{$T4$} & $\infty$ & $\tau \leq 0$ & \multicolumn{1}{c|}{$12.60$ aW} & \multicolumn{1}{c|}{$9.20$ aW} & \multicolumn{1}{c|}{$6.80$ aW} & $5.20$ aW \\ \cline{2-7} 
&$200$ m & $\tau \leq 7$ & \multicolumn{1}{c|}{$1.8 - 45.7$ yW}   & \multicolumn{1}{c|}{$0.86$ yW}   & \multicolumn{1}{c|}{$0.59$ yW}   & $0.43$ yW  \\ \cline{2-7} 
& $50$ m & $\tau \leq -15$ & \multicolumn{1}{c|}{-}   & \multicolumn{1}{c|}{-}   & \multicolumn{1}{c|}{$0.41$ yW}   & $0.33$ yW  \\ \cline{2-7} 
& $0$ & $\tau \leq -10$ & \multicolumn{1}{c|}{-}   & \multicolumn{1}{c|}{-}   & \multicolumn{1}{c|}{-}  & $0.39$ yW \\ \hline
\end{tabular}}
\vspace*{-.17in}
\end{table}
Thus, the presence of blockages changed a fundamental behaviour that a milder transmission restriction (higher $\rho$) on secondary network does not always ensure better secondary coverage. As discussed, $L_\mu$ and $M$ are important factors in determining  $\rho^\dagger$. Therefore, guidelines for the selection of suitable $\rho^\dagger$ can be modified as: {\em for some $\prim^\star, \, \seco^\star$ with $0 \leq \{\prim^\star, \seco^\star\} \leq 1$, $\mu^\star$ and $\tau^\star$, find}
\begin{align}
\!\!\!\!
\rho^\dagger =  \min \rho  &\ \ \mathrm{s.t.}~  p_{\mathrm{cp}} (\tau^\star\!, \mu^\star\!, \rho) \geq \prim^\star \, ; \, p_{\mathrm{cs}} (\tau^\star\!, \mu^\star\!, \rho) \geq  \seco^\star. 
\!\!
\label{eq:OptCondition} 
\end{align}

Note that if $\tau^\star$ exists, the condition in \eqref{eq:OptCondition} is also true for all $\tau \leq \tau^\star$. For the exposition of the system's behaviour, Table \ref{Table:Blockage:SinglePrimary:SuitableRho} lists the solution for \eqref{eq:OptCondition} for some sets of parameter values.

\subsection{Impact of antenna directionality $M_{uv}$} \label{Section:Blockage:SinglePrimary:NumericalResults:Directionality}
From Figs. \ref{numres:NPvsRhoset4_alpha_2point4_4point2_R_4000_1byMus_Inf_0_tau_Minus10_5_dBs} and \ref{numres:NPvsRhoset4_alpha_2point4_4point2_R_4000_1byMus_200_50_tau_Minus10_5_dBs}, we can observe that directionality improves the feasibility of implementing SLS even in the presence of high blockage losses. For example, for $L_\mu = 50$ m, $\tau \leq -10$ dB and $\prim^\star = \seco^\star = 0.5$, no feasible value of $\rho^{\dagger}$ exists under omni SLS (see Fig. \ref{numres:NPvsRhoset4_alpha_2point4_4point2_R_4000_1byMus_200_50_tau_Minus10_5_dBs}). However, we can easily satisfy the constraint with $M_{uv} = 4$. From Table \ref{Table:Blockage:SinglePrimary:SuitableRho}, we can observe that on average, higher directionality ($M_{uv}$) and smaller average-LOS distance ($L_\mu$) allow a higher restriction (smaller $\rho^\dagger$) to be put on secondary links while providing the same performance guarantee for primary and secondary links. We fix $\rho = 1$ pW unless stated otherwise.
\subsubsection{On the primary coverage}
\label{Section:Blockage:SinglePrimary:NumericalResults:Directionality:PrimaryCov}
Fig. \ref{numres:Blockage:SinglePrimary:PPwithM} shows the variation of primary coverage with  $\tau$ under omni and spatially-aware SLS with $\rho = 1$ pW for different values of $L_\mu$. 
\begin{figure}[b!]
\vspace{-.15in}
\centering
{\includegraphics[trim = 20 10 35 40, clip, scale=0.3899]{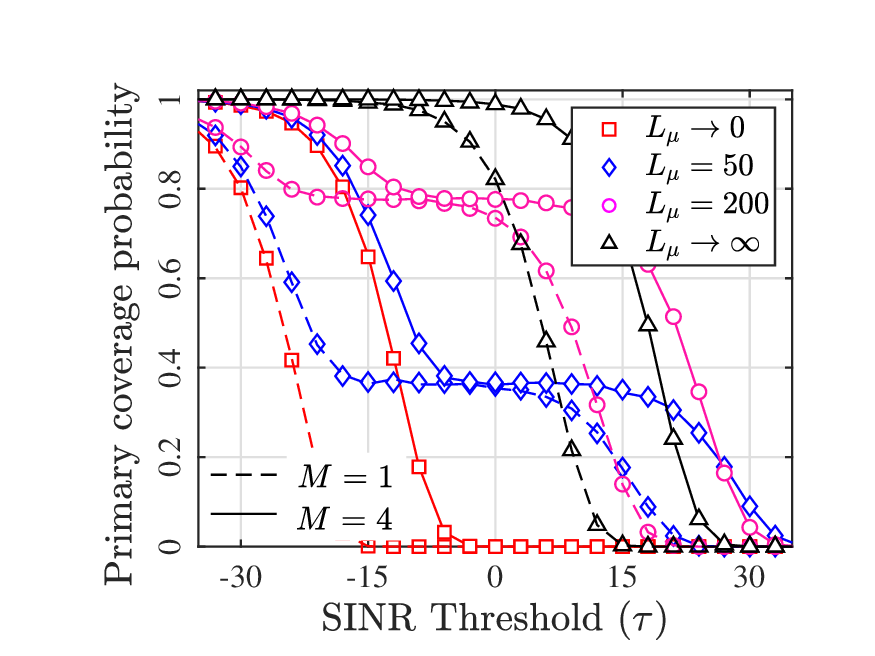} \label{numres:PP_1byMu_alpha_2point4_4point2_Rho_1eMinus12_R_4000} }
\caption{The variation of  primary link performance with SINR threshold $\tau$ (in dB) at different values of $L_\mu$ for $\rho = 1$ pW with $\lambda_\seco = 8 \times 10^{-5}$ /m$^2$.}
\label{numres:Blockage:SinglePrimary:PPwithM}
\end{figure}
First, we note that the blockages can change the fundamental shape of the coverage curve. For the case when both LOS/NLOS links are present ({\em i.e.,} $L_\mu = 50$ m, $200$ m), there exists a plateau region surrounded by two falls on each side of the graph. These left and right falls correspond to coverage of primary link being in NLOS and LOS state, respectively, while the plateau region (of approximate $21$ dB stretch) represents the negligible number of NLOS links that can provide coverage in this range of $\tau$. This plateau is absent in ZBL ($L_\mu \to \infty$) and HBL ($L_\mu \to 0$) scenarios, and there is only one fall. Second, we notice that blockages can improve the primary coverage by reducing the secondary interference; however, only up to a certain severity. Increasing blockage severity beyond this may render the primary link NLOS, degrading its performance. Hence, moderate blockage is favourable. Further, we observe that the antenna directionality improves coverage as seen by a consistent positive shift of $12$ dB, as we go from $M \!=\! 1$ (omni) to $4$ ($\phi = 30.25^\circ$). We note that the shift may not be visible in the plateau region present when $L_\mu = 50$ or $200$ m. While the directionality improves the performance of both LOS and NLOS links, it cannot change the link state. Hence, we see no gain between $M \!= \!1$ and $M \!=\! 4$ curves for an approximate $3$ dB range, which can be justified as follows. At $\tau = -15$ and $0$ dB, the primary coverage with $M = 1$ is mainly contributed by the LOS link, as the NLOS primary link has negligible probability to give SINR above these thresholds. Even with $M$ increased to $4$, an NLOS primary link cannot still be improved beyond $0$ dB. On the other hand, $M = 4$ can help it go above $-15$ dB showing, an increase in $p_\mathrm{cp}$ for $\tau=-15$ dB. 

\subsubsection{On the secondary coverage performance} 
\label{Section:Blockage:SinglePrimary:NumericalResults:Directionality:SecondaryCov}
Figs. \ref{numres:Blockage:SinglePrimary:SPwithM}(a)-(d) shows the variation of secondary coverage with $\tau$ under the omni and directional cases for all four types of secondary users (see Fig. \ref{fig:Blockage:SinglePrimary:SystemModel-Secondary}(b)), where we observe the following interesting trends. 
\begin{figure}[b!]
\vspace*{-0.24in}
\centering
{\includegraphics[trim = 20 8 35 0, clip, scale=0.32889]{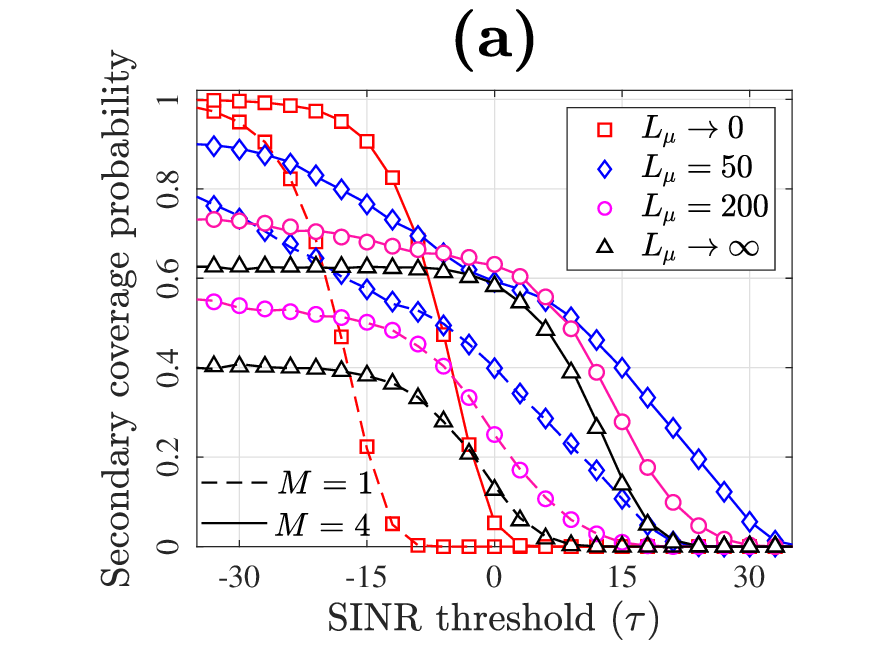} \label{numres:SPset1_1byMu_alpha_2point4_4point2_Rho_1eMinus12_R_4000}}
{\includegraphics[trim = 92 8 35 0, clip, scale=0.32889]{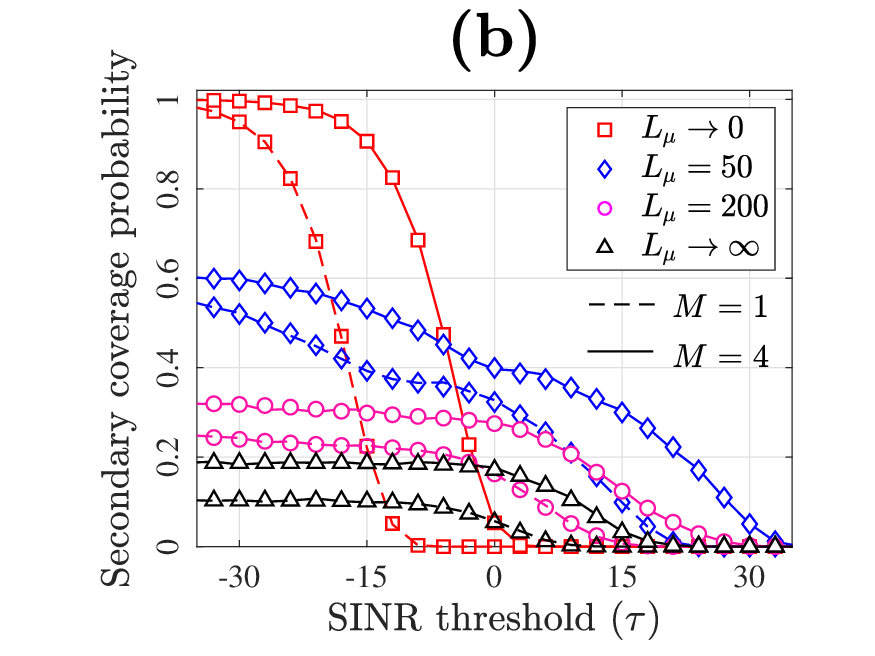} \label{numres:SPset2_1byMu_alpha_2point4_4point2_Rho_1eMinus12_R_4000}}
{\includegraphics[trim = 92 10 35 0, clip, scale=0.32889]{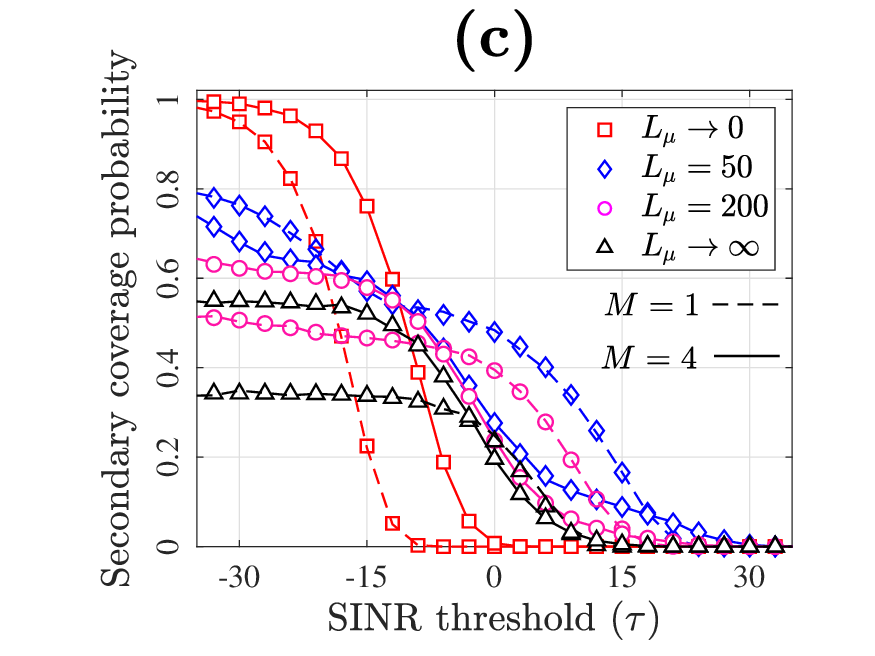} \label{numres:SPset3_1byMu_alpha_2point4_4point2_Rho_1eMinus12_R_4000}}
{\includegraphics[trim = 92 10 35 0, clip, scale=0.32889]{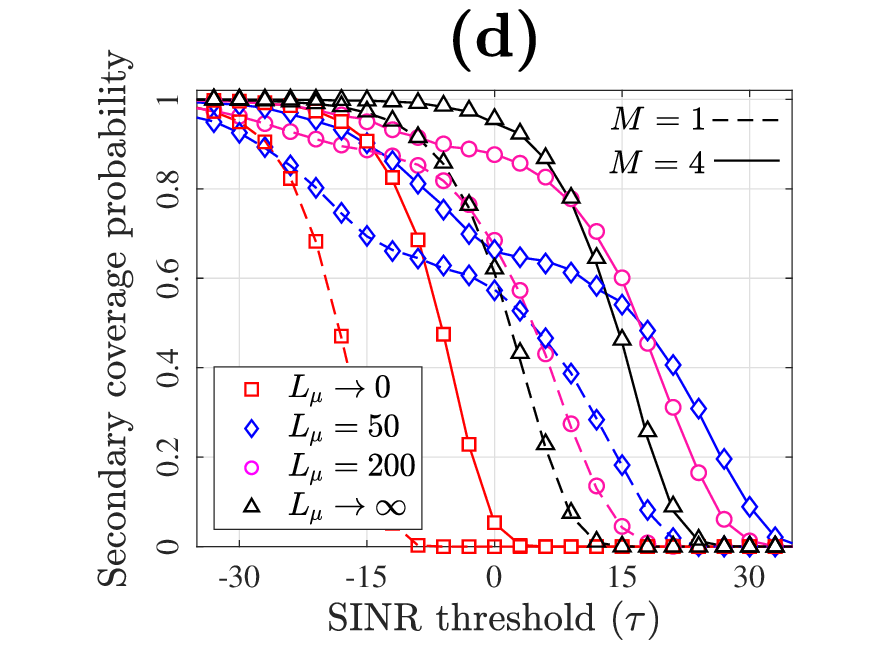} 
\label{numres:SPset4_1byMu_alpha_2point4_4point2_Rho_1eMinus12_R_4000}}
\caption{The variation of the secondary link performance with SINR threshold $\tau$ (in dB) at different values of  $L_\mu$ (in m) for (a) T1, (b) T2, (c) T3, (d) T4.}
\label{numres:Blockage:SinglePrimary:SPwithM}
\end{figure}
First, the impact of directionality and blockages differs significantly for these users, highlighting the role of its location/orientation in its performance. For example, as $M$ is changed from $1$ to $4$, a positive shift in $p_\mathrm{cs}$ is observed for T1 and T2 users  (see Figs. \ref{numres:Blockage:SinglePrimary:SPwithM}(a) and (b)) while, for the T3 user, it does not always improve $p_\mathrm{cs}$. In particular, except for the case with $L_\mu \to 0$, both negative and positive shifts are observed in $p_\mathrm{cs}$ depending on $\tau$. For example, with $L_\mu = 200$m and $L_\mu \to \infty$, a positive shift is observed at low $\tau$ while a negative shift is observed at high $\tau$ for the T3 user (see Fig. \ref{numres:Blockage:SinglePrimary:SPwithM}(c)). Intuitively, this occurs mainly because the directionality increases the activity of other secondary transmitters, leading to an increase in the secondary interference at the unprotected secondary user. Typically, if the secondary user is close to the primary receiver with a similar orientation, it enjoys the protection offered to the primary receiver, as seen in T1 and T2. The specific configuration of T3 causes a mismatch, resulting in a loss of protection region and hence, an elevated level of secondary interference at itself. Fig. \ref{numres:Blockage:SinglePrimary:SPwithM}(d) shows the variation of the typical secondary user's coverage. Here, we observe a positive shift in $p_\mathrm{cs}$ with higher antenna directionality, irrespective of the values of $L_\mu$, which represents an average effect. Thus, we can conclude that higher antenna directionality improves the coverage of both the primary link and the typical secondary link, even in the presence of blockage. It is noteworthy that, similar to primary coverage,  blockages can also improve the secondary coverage; however, only up to a certain severity. Increasing blockage severity beyond this may degrade the coverage. Hence, moderate blockage is favourable for both primary and secondary links. 

\subsection{Impact of secondary density}\label{Section:Blockage:SinglePrimary:NumericalResults:Density}
Fig. \ref{numres:Blockage:SinglePrimary:PPwithMnew} shows the variation of primary and secondary coverage with $\tau$ for a ten-fold higher secondary density of $8 \times 10^{-4}$ /m$^2$. A significant negative shift of $-9$ dB in $p_\mathrm{cp}$ is observed (see Fig. \ref{numres:Blockage:SinglePrimary:PPwithMnew}(a)) in comparison to Fig. \ref{numres:Blockage:SinglePrimary:PPwithM}, which is due to the threshold $\rho = 1$ pW not being able to limit the secondary interference on the primary for this value of $\lambda_\seco$. However, if we decrease $\rho$ to $1$ fW to make the SLS restriction harsher, we can reduce the secondary interference to negate this degradation as shown in Fig. \ref{numres:Blockage:SinglePrimary:PPwithMnew}(b). Similarly, Fig. \ref{numres:Blockage:SinglePrimary:SPwithM_T4new}(c) shows a significant degradation in $p_\mathrm{cs}$ of the typical secondary user in comparison to Fig. \ref{numres:Blockage:SinglePrimary:SPwithM}(d). However, decreasing $\rho$  from $1$ pW to $1$ fW is not able to negate this degradation, as shown in Fig. \ref{numres:Blockage:SinglePrimary:SPwithM_T4new}(d), as a lower $\rho$ also restricts secondary opportunities and hence, its coverage.
\begin{figure}[ht!]
\vspace*{-0.178in}
\centering
{\includegraphics[trim = 20 8 35 0, clip,scale=0.32799]{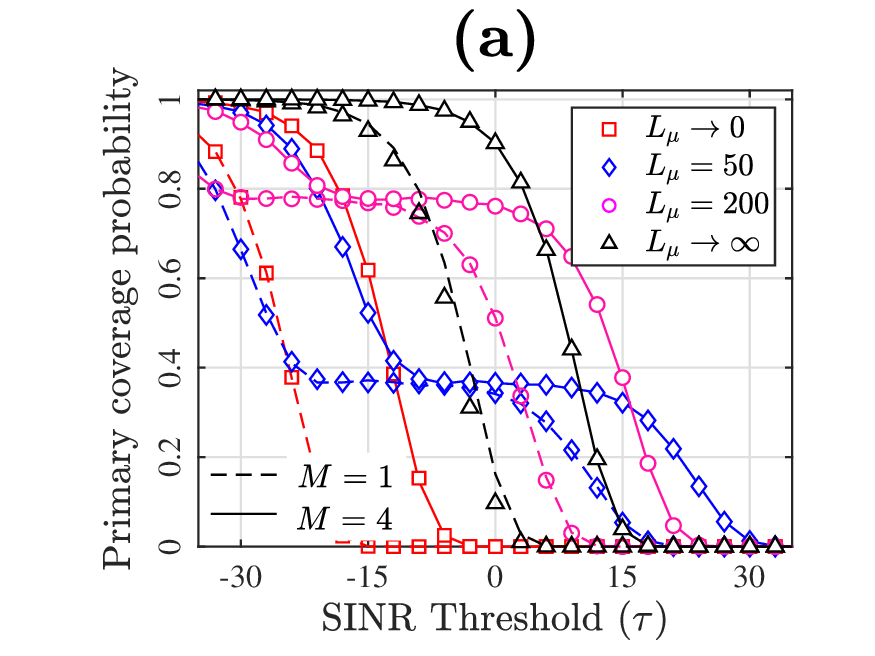} \label{numres:PP_Lambda_Multiplied_By_10_1byMu_alpha_2point4_4point2_Rho_1eMinus12_R_4000}}
{\includegraphics[trim = 72 8 38 0, clip,scale=0.32799]{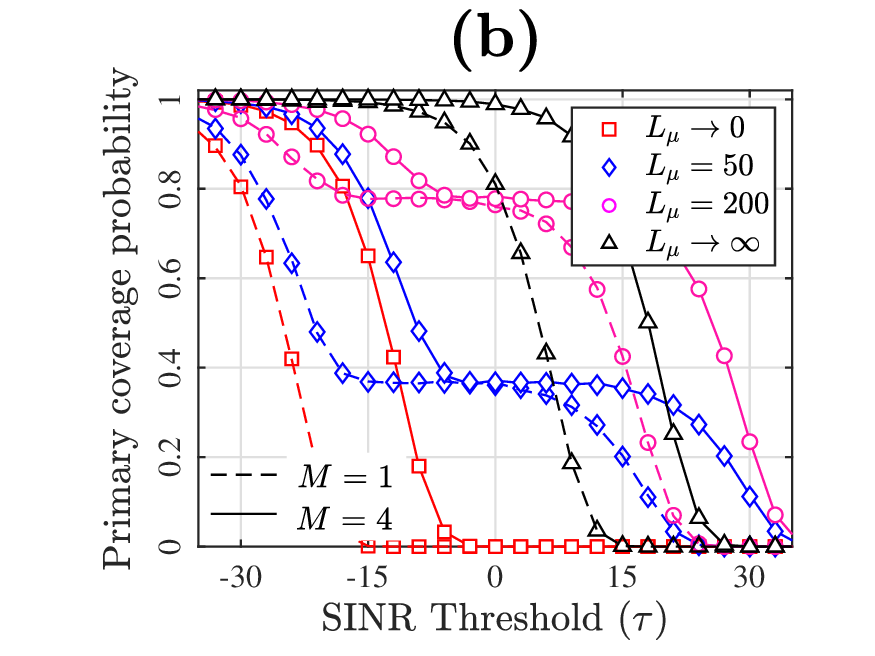} \label{numres:PP_Lambda_By_4_1byMu_alpha_2point4_4point2_Rho_1eMinus12_R_4000}} 
{\includegraphics[trim = 30 10 35 0, clip,scale=0.32799]{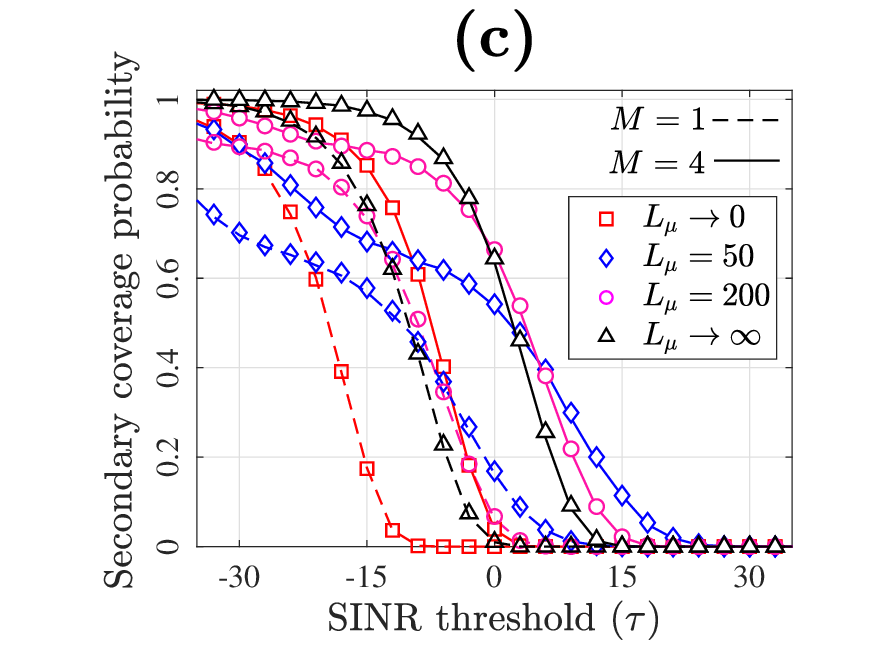} \label{numres:SPset4_Lambda_Multiplied_By_10_1byMu_alpha_2point4_4point2_Rho_1eMinus12_R_4000}}
{\includegraphics[trim = 72 10 35 0, clip,scale=0.32799]{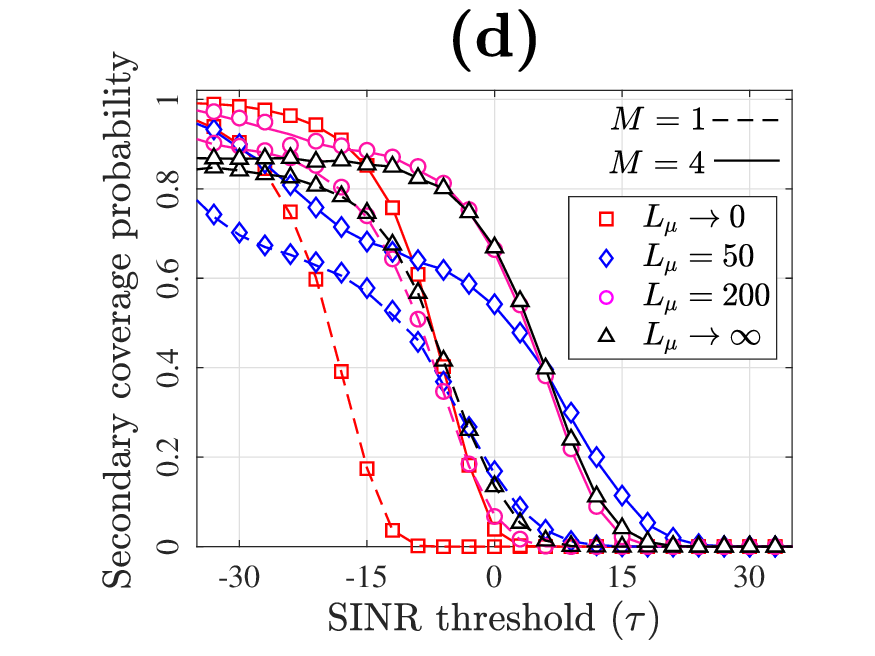} \label{numres:SP_Lambda_By_4_1byMu_alpha_2point4_4point2_Rho_1eMinus12_R_4000}} 
\caption{The variation of primary and secondary link performances with SINR threshold $\tau$ (in dB) at different values of average LOS distance $L_\mu$ (in meters) for $\lambda_\seco = 8 \times 10^{-4}$ /m$^2$ with (a) $\rho = 1$ pW and (b) $\rho = 1$ fW.}
\label{numres:Blockage:SinglePrimary:PPwithMnew}
\label{numres:Blockage:SinglePrimary:SPwithM_T4new}
\vspace{-.35in}
\end{figure}

\subsection{Impact of blockages $L_\mu$} \label{Section:Blockage:SinglePrimary:NumericalResults:BlockageParameter}
Figs. \ref{numres:Blockage:SinglePrimary:PPwithM} - \ref{numres:Blockage:SinglePrimary:SPwithM_T4new} show the impact of $L_\mu$ on the shapes of coverage curves. Specifically, for $L_\mu = 50$ m and $200$ m, primary coverage has two falls separated by a plateau region, which is absent for LOS and NLOS only cases, as explained in Section \ref{Section:Blockage:SinglePrimary:NumericalResults:Directionality}.1. Further, for $L_\mu = 50$ m, the coverage curves of the secondary user T3  for $M = 1$ and $M = 4$ intersect each other at multiple points, as shown in Fig. \ref{numres:Blockage:SinglePrimary:SPwithM}(c). Therefore, we now carefully investigate the role of $L_\mu$.
\subsubsection{On the primary coverage} 
\label{Section:Blockage:SinglePrimary:NumericalResults:BlockageParameter:PrimaryCov}
Fig. \ref{numres:Blockage:SinglePrimary:PPwithMu} shows the variation of $p_\mathrm{cp}$ with $L_\mu$ at different values $\tau$. Note that increasing $L_\mu$ (i) improves $p_\mathrm{cp}$ due to the higher LOS probability of the primary link and (ii)  increases secondary interference due to a larger number of LOS secondary transmitters. A smaller $\rho$ does not allow secondary interference to increase; hence, the first effect dominates, resulting in an improvement in $p_\mathrm{cp}$ with $L_\mu$. When $\rho$ is large, both effects are active, and we see a trade-off. 
\begin{figure}[ht!]
\vspace{-0.15in}
\centering
{\includegraphics[trim = 20 6 29 7, clip,scale=0.3799]{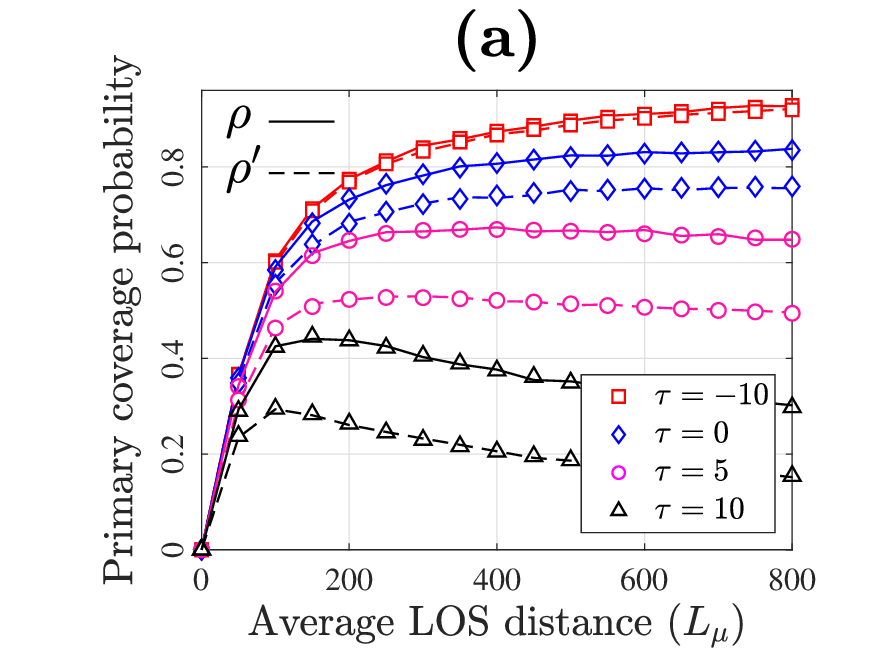} \label{numres:PPvs1byMu_alpha_2point4_4point2_Rho_1eMinus11_1eMinus12_M_1} } 
{\includegraphics[trim = 72 6 29 7, clip,scale=0.3799]{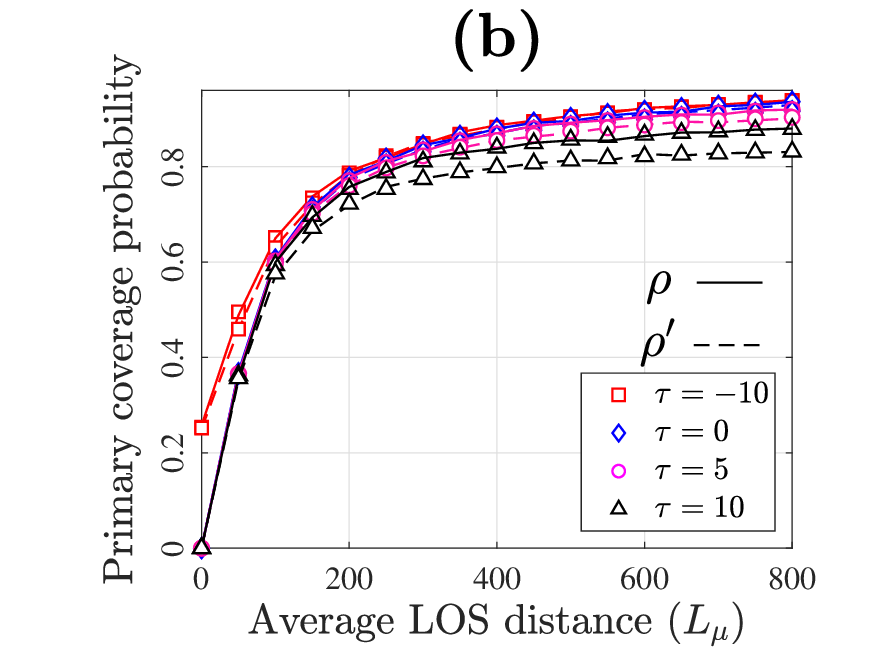} \label{numres:PPvs1byMu_alpha_2point4_4point2_Rho_1eMinus11_1eMinus12_M_4} }
\caption{The variation of primary link performance with average LOS distance $L_\mu = 1/\mu$ (in meters) at different values of SINR threshold $\tau$ (in dB) by varying the primary-transmit-protection threshold $\rho$ with (a) $M = (1,\, 1)$ and (b) $M = (4,\, 4)$. Here, $\rho = 1$ pico-watts and $\rho{'} = 10$ pico-watts.}
\label{numres:Blockage:SinglePrimary:PPwithMu}
\vspace{-.15in}
\end{figure}
Initially, the $p_\mathrm{cp}$ improves due to the first effect. However, beyond a certain value of $L_\mu$, the primary LOS probability does not increase much while the secondary interference becomes significantly high, reducing $p_\mathrm{cp}$ overall. Further, coverage at high $\tau$ is mostly due to the LOS primary link with NLOS secondary interference. Hence, the trade-off is more prominent in cases with high $\tau$ and $\rho$. Also, since directionality improves the primary link's serving power and reduces the effect of secondary interference, a consistent improvement in $p_\mathrm{cp}$ is observed with $L_\mu$ (see Fig. \ref{numres:Blockage:SinglePrimary:PPwithMu}(b)).
\subsubsection{On the secondary coverage performance}
\label{Section:Blockage:SinglePrimary:NumericalResults:BlockageParameter:SecondaryCov} 
Fig. \ref{numres:Blockage:SinglePrimary:SPwithMu} shows the variation of $p_\mathrm{cp}$ with $L_\mu$ at different values of $\tau$ for all four types of secondary users. We observe similar trends as seen in $p_\mathrm{cp}$ in Fig \ref{numres:Blockage:SinglePrimary:PPwithMu}. We also observe a sharp rise in $p_\mathrm{cs}$ until $L_\mu \approx 50$ m, followed by gradual degradation with increasing $L_\mu$. 
\begin{figure}[ht!]
\vspace*{-0.20in}
\centering
{\includegraphics[trim = 20 5 15 0, clip,scale=0.30285]{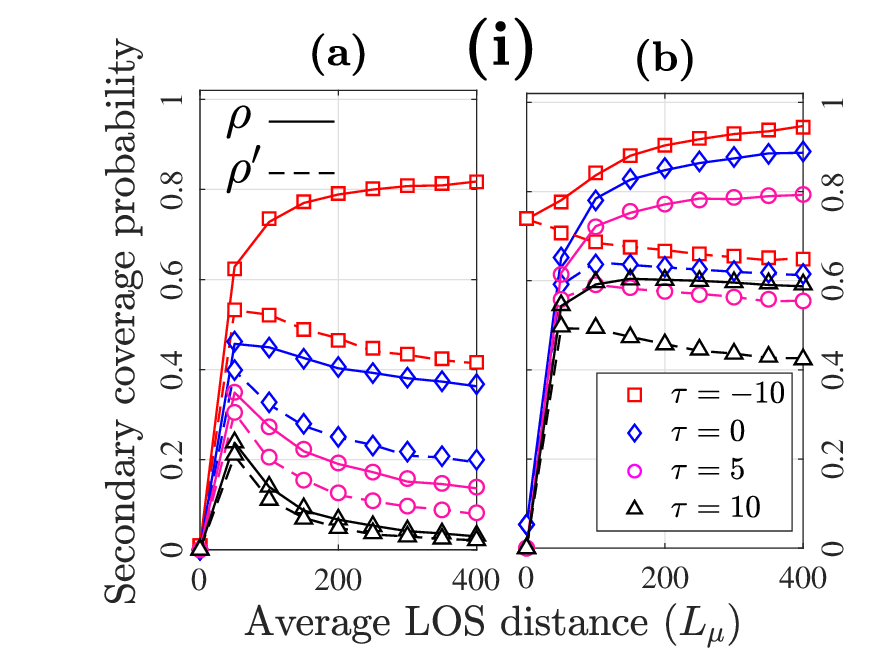} \label{numres:SPset1_1byMu_alpha_2point4_4point2_Rho_1eMinus11_1eMinus12_R_4000_M_1_4} } 
{\includegraphics[trim = 72 5 20 0, clip,scale=0.30285]{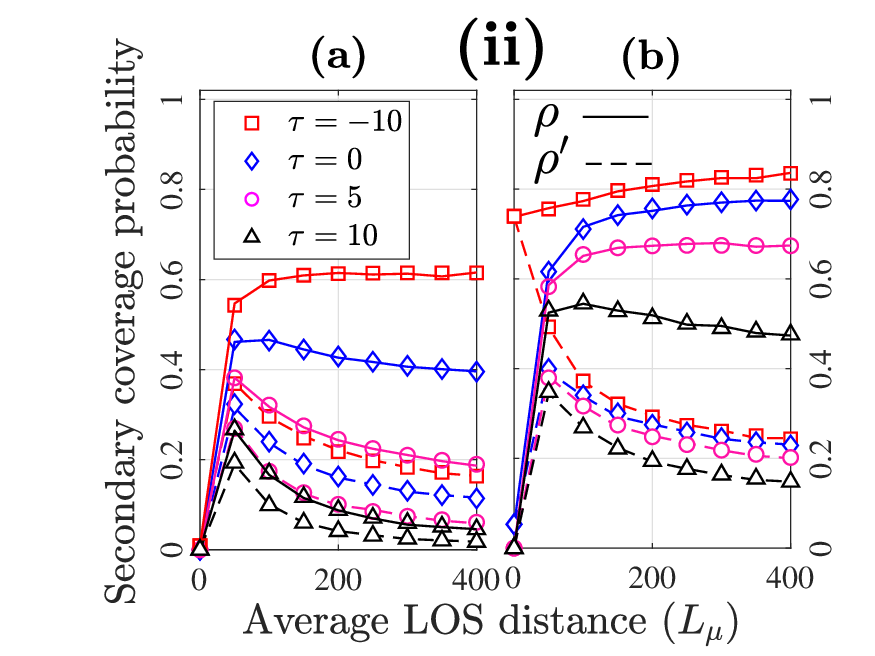} \label{numres:SPset2_1byMu_alpha_2point4_4point2_Rho_1eMinus11_1eMinus12_R_4000_M_1_4} }
{\includegraphics[trim = 30 5 15 0, clip,scale=0.30285]{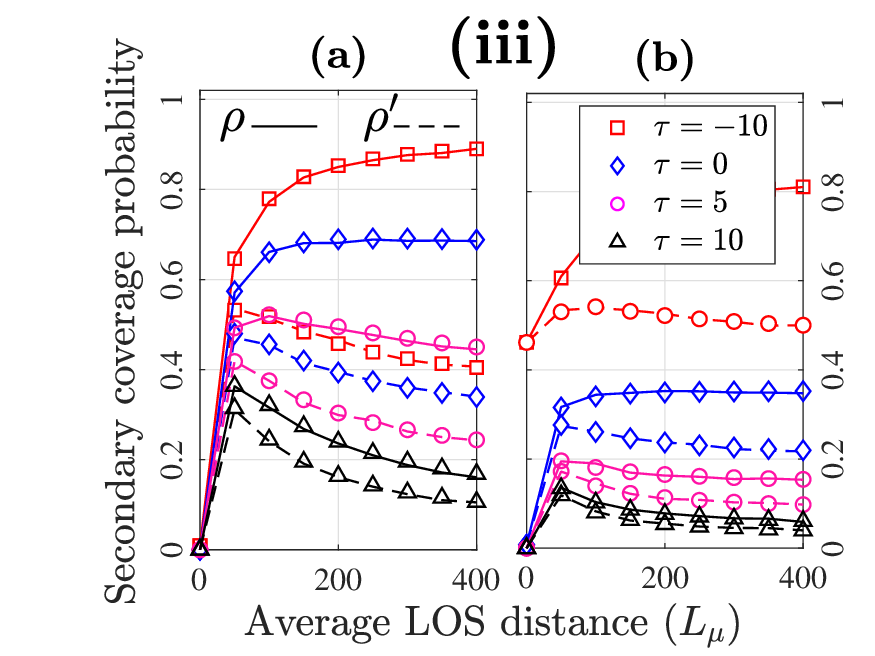} \label{numres:SPset3_1byMu_alpha_2point4_4point2_Rho_1eMinus11_1eMinus12_R_4000_M_1_4} } 
{\includegraphics[trim = 72 5 15 0, clip,scale=0.30285]{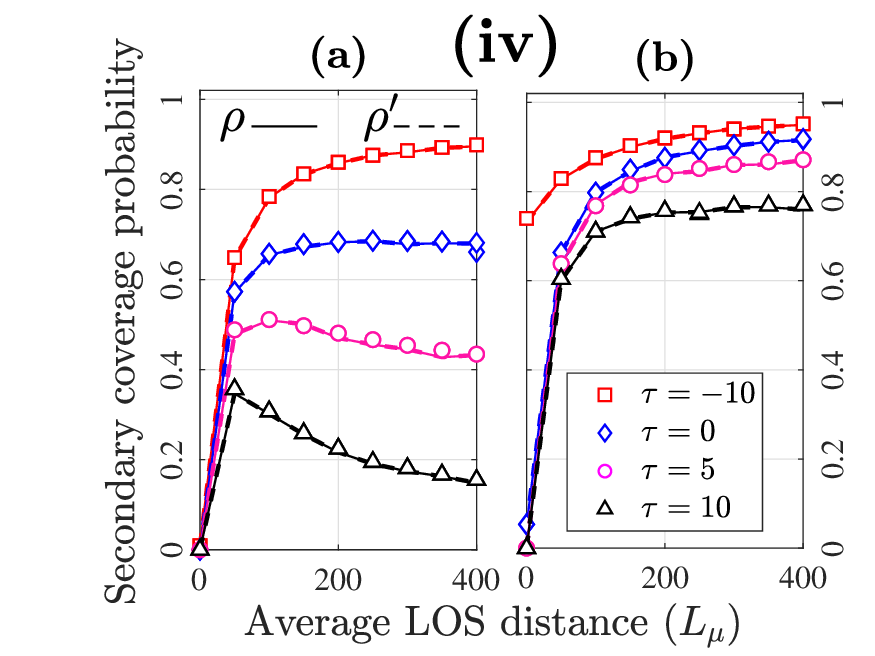} \label{numres:SPset4_1byMu_alpha_2point4_4point2_Rho_1eMinus11_1eMinus12_R_4000_M_1_4} }
\caption{Secondary performance vs  $L_\mu$ (in m) at different values of $\tau$ (in dB) by varying $\rho$ with (a) $M = (1,\, 1)$ and (b) $M = (4,\, 4)$ for (i) T1 (ii) T2 (iii) T3 and (iv) T4  secondary users. Here, $\rho = 10$ pW and $\rho{'} = 1$ pW.}
\label{numres:Blockage:SinglePrimary:SPwithMu}
\vspace{-.15in}
\end{figure}
First note that for the given value of $\lambda_\seco$, the average distance of the closest secondary transmitter from the primary receiver is $R_\seco = 1/2\sqrt{\lambda_\seco} = 55$ m. With $L_\mu$, the secondary link LOS probability improves. As $L_\mu < R_\seco$, most of the secondary transmitters are NLOS, causing a low level of interference. Hence $p_\mathrm{cs}$ improves. However, as $L_\mu$ increases beyond $R_\seco$, LOS secondary transmitters can degrade $p_\mathrm{cs}$.

\section{Conclusions}\label{Section:Blockage:SinglePrimary:Conclusion}
In this work, we proposed an analytical framework to study and design a spatially-aware SLS in a mmWave network consisting of a  primary and multiple secondary links. The work resulted in many valuable insights. We show that the presence of blockages can fundamentally change the shape of AF and coverage curves due to the presence of dual-slop path-loss. We find that blockages improve transmission opportunities significantly, which may increase the secondary interference. On the other hand, blockages can also render links NLOS, reducing the individual interference. Hence, increasing blockages results in a trade-off between the secondary transmitters' activity and their interference. We show that directionality can reduce the secondary interference, offering a better trade-off. Further, we showed that due to transmit restriction offering protection to the primary, the secondary interference does not grow large at the primary user, while this cannot be said about secondary users. We saw that while relaxed restriction degrades primary coverage, it can improve as well as degrade secondary coverage due to the lack of any protection offered to it against mutual secondary interference. We also show that directionality and blockages can improve the feasibility of SLS and establish the benefits of spatially-aware SLS. Overall, moderately severe blockages can provide improvement to both primary and secondary links. 

\appendices

\section{} \label{thrm:proof:Blockage:AF}
Applying Campbell’s theorem on \eqref{eq:Blockage:SinglePrimary:AF:Definition} and using \eqref{eq:Blockage:SinglePrimary:MAP}, we get
\begin{align*}
\eta_{\mathrm{s}} &= \frac{e^{-p}}{\pi R^2} \int_0^{2\pi} 
\!\!
\mathbb{E}_{\omega_\seco} \Big[  \int_0^R e^{ - \mu x_\seco - \frac{\kappa_T  x_\seco^{ \alpha_T}}{\kappa_\mathrm{0}}}  x_\seco \dd x_\seco \Big] \dd \theta_\seco.
\end{align*}

Substituting $\exp(-\mu x_\seco) \!\!\! = \!\!\! \sum_{n = 0}^\infty {(- \mu x_\seco)^n}/{(n!)}$ with $t = \rho x_\seco^{\alpha_T}/(p_\seco C_T \gpr{\theta_\seco} \gst{\theta_\seco - \pi - \omega_\seco})$ and using the definition of $\psi_T(m, u)$ in \eqref{eq:Blockage:SinglePrimary:AF:PsiFunction} will give the desired result.

\section{} \label{thrm:proof:Blockage:AF-sector}
From \eqref{eq:Blockage:SinglePrimary:AF}, let us define $f \! = \!\! \int_0^{2\pi} \! \mathbb{E}_{\omega_\seco} [ \psi_T \left(n + 2, \, {\kappa_T}/{\kappa_\mathrm{0}}\right)] \dd \theta_\seco$. Substituting  $\kappa_\mathrm{0} = g_{\mathrm{pr}} \left( \theta_\seco \right) g_{\mathrm{st}} \left( \theta_\seco - \pi - \omega_\seco \right)$ will give
\begin{align}
&f 
\! = \! 2 \, \mathbb{E}_{\omega_\seco} 
\!\! 
\left[ \int_0^{{\phi_{\mathrm{pr}}}/{2}} 
\!\!\!\!\!\!
\psi_T \Big(n + 2, \, \frac{\kappa_T}{\apr g_{\mathrm{st}} (\theta_\seco - \pi - \omega_\seco)}\Big) 
\dd \theta_\seco 
\! + \!\!
\int_{{\phi_{\mathrm{pr}}}/{2}}^{\pi}  
\!\!\!\!\!\!
\psi_T \Big(n + 2, \, \frac{\kappa_T}{\bpr g_{\mathrm{st}} (\theta_\seco - \pi - \omega_\seco)}\Big) 
\right]
\!\! 
\dd \theta_\seco
\label{eq:Blockage:SinglePrimary:AF:Sector-1}
\end{align}

\noindent Since $ \omega_{\mathrm{s}} \in [-\pi,\ \pi] $,  the secondary transmission gain 
$ g_{\mathrm{st}} ( \theta_\seco - \pi - \omega_\seco) $ takes values $\ast$ and $\bst$ with probabilities $\qst$ and $ (1 - \qst)$, respectively. 
Using values of $\mathcal{Q}_i$ and $\mathcal{G}_i$ in \eqref{eq:Blockage:SinglePrimary:AF:Sector-1} gives $f = 2\pi \sum_{i = 1:4} \mathcal{Q}_i$ $\psi_T (n + 2, {\kappa_T}/{\mathcal{G}_i})$. Substituting $f$ in \eqref{eq:Blockage:SinglePrimary:AF} will give desired result.

\section{} \label{thrm:proof:Blockage:PrimaryCov}
From \eqref{eq:Blockage:SinglePrimary:SINR:PrimaryLink}, coverage $p_\mathrm{cp}(\SThres, \mu, \rho)$ of the primary link is
\begin{align}
p_\mathrm{cp}(\SThres, \mu,\rho) &\overset{(a)}{=} 
\sum\nolimits_{T} 
p_T (r_\prim) 
\, 
\exp \left(- s_T \sigma^2\right) \mathcal{L}_{ I_\seco } \left( s_T \right). 
\label{eq:Blockage:SinglePrimary:PrimaryCov-1}
\end{align}
where $s_T = \tau r_\prim^{\alpha_T} / (C_T p_\prim g_\mathrm{pt}(0) g_\mathrm{pr}(0))$ with $T = \{\los,\nlos\}$ and $\mathcal{L}_{ I_\seco}(s)= \mathbb{E} \left[ e^{ - s I_{ \mathrm{s} }  \left( \Phi_{\seco} \right)} \right]$ is the Laplace transform (LT) of $I_\seco$. Here, $(a)$ is due to $H_{\prim\mathrm{0}} \sim \exp(1)$. Now 
\begin{align}
\mathcal{L}_{ I_\seco } \left( s_T \right) 
&= \mathbb{E} \left[ e^{- s_T 
\sum\nolimits_{\X_{\mathrm{s}i} \in \Phi_{\seco}}
p_\seco \kappa_{i\mathrm{0}} U_{i} G_{i\mathrm{0}} \ell_{T_{\seco i \to \prim\mathrm{0}}} (x_{\seco i})} \right] \nonumber \\
&\ 
\overset{(b)}{=} e^{ - \lambda_\seco \int_0^{2\pi} \int_0^\infty \left( 1 -  \mathbb{E} \left[  e^{- s_T p_\seco \kappa_\mathrm{0} U G_\mathrm{0} \ell_{T_{\seco \to \prim\mathrm{0}}} (x_\seco)} \right] \right) x_\seco \dd x_\seco \dd \theta_\seco},
\label{eq:Blockage:SinglePrimary:PrimaryCov-6}
\end{align}
where $\kappa_\mathrm{0} \! = \! g_{\mathrm{pr}} \left( \theta_\seco \right) g_{\mathrm{st}} \left( \theta_\seco - \pi - \omega_\seco \right)$. Here, $(b)$ is due to  PGFL of PPP $\Phi_{\seco}$ \cite{andrews2023introduction}. Note that the expectation $\mathbb{E} \left[\cdot\right]$ is with respect to four RVs $\ell_{T_{\seco \to \prim\mathrm{0}}}$, $U$, $G_\mathrm{0}$ and $\omega_\seco$. Since ${T_{\seco \to \prim\mathrm{0}}}$ and $U$ are Bernoulli RVs, we can write 
\begin{align}
\mathcal{L}_{ I_\seco } \left( s_T \right) &= e^{\! - \lambda_\seco 
\! 
\int\limits_0^{2\pi} 
\!
\int\limits_0^\infty 
\! 
\big(\! 1 -  \mathbb{E}_{G_\mathrm{0}, \omega_\seco} \big[
\!
\sum\limits_{T_1} p_T (x_\seco) \, e^{\!\!- s_T \chi_{T_1} G_\mathrm{0} \mathbb{I} ( G_\mathrm{0} < {\rho}/{\chi_{T_1}})} \big] \big) x_\seco \dd x_\seco \dd \theta_\seco} \nonumber \\
&\overset{(c)}{=} e^{\!\! - \lambda_\seco 
\!\!
\int\limits_0^{2\pi} 
\!
\int\limits_0^\infty 
\!
\left(\! 1 \! - \!
\sum\limits_{T_1} p_{T_1} \! (x_\seco) \, 
\mathbb{E}_{\omega_\seco} 
\! 
\left[ e^{ \!\!\! - \frac{\rho}{\chi_{T_1}}} + \frac{ 1 - e^{\!\! -( \! s_T + {1}/{\chi_{T_1}} \!) \rho} }{1 + s_T \chi_{T_1}} 
\!
\right] 
\!
\right)
\!
x_\seco \dd x_\seco \dd \theta_\seco}, 
\!\!\!\!\!\!\!
\label{eq:Blockage:SinglePrimary:PrimaryCov-8}
\end{align}
where $\chi_{T_1} = p_\seco \kappa_\mathrm{0} C_{T_1} x_\seco^{-\alpha_{T_1}}$ for ${T_1} = \{\los, \nlos \}$ and $(c)$ is due to $G_\mathrm{0} \sim \exp(1)$. Substituting \eqref{eq:Blockage:SinglePrimary:PrimaryCov-8} in \eqref{eq:Blockage:SinglePrimary:PrimaryCov-1} along with the values of $s_T$, $\chi_{T_1}$ and $\kappa_\mathrm{0}$ gives 
\eqref{eq:Blockage:SinglePrimary:PrimaryCov}. Now, from \eqref{eq:Blockage:SinglePrimary:PrimaryCov}, we can write
\begin{align}
&p_\mathrm{cp}(\tau,\mu,\rho) \! = \!
\sum_{T} p_T (r_\prim) \, e^{\!\!- s_T \sigma^2 - \lambda_\seco \int_0^{2\pi} \mathbb{E}_{\omega_\seco}\left[ \mathrm{I_1} -  \mathrm{I_2} + \mathrm{I_3} \right] \dd \theta_\seco},
\label{eq:Blockage:SinglePrimary:PrimaryCov:Simplified-1}
\end{align}	
\vspace*{-0.25in}
\begin{align*}
\text{with } \qquad\qquad\qquad\qquad\qquad &\mathrm{I_1} = \!\!\int_0^\infty\!\! \bigg( 1 - \frac{ 1 - e^{- \big(s_T + \frac{x_\seco^{\alpha_\nlos}}{p_\seco\kappa_\mathrm{0} C_\nlos}\big) \rho} }{1 + s_T \frac{p_\seco\kappa_\mathrm{0} C_\nlos}{x_\seco^{\alpha_\nlos}}} - e^{- \frac{\rho x_\seco^{\alpha_\nlos}}{p_\seco\kappa_\mathrm{0} C_\nlos}} \bigg) x_\seco \mathrm{d}x_\seco, \qquad\qquad\qquad\qquad\nonumber\\
&\mathrm{I_2} 
= \int_0^\infty \!\!\! e^{-(\mu x_\seco + p)} \! \bigg[ \frac{ 1 - e^{- \big(s_T + \frac{x_\seco^{\alpha_\los}}{p_\seco\kappa_\mathrm{0} C_\los}\big)\rho} }{1 + s_T \frac{p_\seco\kappa_\mathrm{0} C_\los}{x_\seco^{\alpha_\los}}} + e^{- \frac{ \rho x_\seco^{\alpha_\los}}{p_\seco\kappa_\mathrm{0} C_\los}} \bigg] x_\seco \mathrm{d}x_\seco, \nonumber\\
&\mathrm{I_3} 
= \int_0^\infty \!\!\! e^{-(\mu x_\seco + p)} \! \bigg[ \frac{ 1 - e^{- \big(s_T + \frac{x_\seco^{\alpha_\nlos}}{p_\seco\kappa_\mathrm{0} C_\nlos}\big) \rho} }{1 + s_T \frac{p_\seco\kappa_\mathrm{0} C_\nlos}{x_\seco^{\alpha_\nlos}}} + e^{- \frac{\rho x_\seco^{\alpha_\nlos}}{p_\seco\kappa_\mathrm{0} C_\nlos}} \bigg] x_\seco \mathrm{d}x_\seco, \nonumber
\end{align*}
where first integral is simplified by using change of variables $A = s_T \rho$ and $C = \kappa_\nlos/\kappa_\mathrm{o}$ with $A + C x_\mathrm{s}^{\alpha_{T_1}} = A y$ to give $\mathrm{I_1} = \mathcal{N} ( {2}/{\alpha_\nlos}, s_T \rho) \mathsf{n}_\mathrm{3} ({2}/{\alpha_\nlos})$. Similarly, the second and the third integrals are simplified by using Taylor's series expansion for $\exp(-\mu x_\seco)$ to give $\mathrm{I_2} = - e^{-p} \sum_{n = 0}^\infty [(- \mu)^{n}/n!] \mathcal{N} ((n + 2)/\alpha_\los, s_T \rho ) \mathsf{n}_\mathrm{3} ((n + 2)/\alpha_\los)$ and $\mathrm{I_3} = - e^{-p} \sum_{n = 0}^\infty [(- \mu)^{n}/n!] \mathcal{N} ((n + 2)/\alpha_\nlos, s_T \rho ) \mathsf{n}_\mathrm{3} ((n + 2)/\alpha_\nlos)$. Substituting these values of $\mathrm{I_1}$, $\mathrm{I_2}$ and $\mathrm{I_3}$ in \eqref{eq:Blockage:SinglePrimary:PrimaryCov:Simplified-1} gives \eqref{eq:Blockage:SinglePrimary:PrimaryCov:Simplified} (see supplementary \cite{TripGupTheoremFile2025} for detailed proof).

\section{} \label{thrm:proof:Blockage:PrimaryCov-sector}
Using \eqref{eq:GainApproximation}, we get
\begin{align}
\mathsf{n}_\mathrm{3} ({m}/{\alpha_{T_1}}) \! 
&\overset{(a)}{=} 2 \apr^{m/\alpha_{T_1}} 
\!\!
\textstyle\int_0^{\phi_{\mathrm{pr}}/2} 
\!
\big[\qst \ast^{m/\alpha_{T_1}} \! + (1 - \qst) \bst^{m/\alpha_{T_1}} \big] 
\mathrm{d} \theta_\seco  
\nonumber \\
&\qquad
\textstyle+ 2 \bpr^{m/\alpha_{T_1}} 
\!\!
\int_{\phi_{\mathrm{pr}}/2}^{\pi}  
\!
\big[\qst \ast^{m/\alpha_{T_1}} + (1 - \qst) \bst^{m/\alpha_{T_1}} \big]  
\mathrm{d}\theta_\seco  
\nonumber \\
&= 2 \pi \sum\nolimits_{i = 1:4} \mathcal{Q}_i \mathcal{G}_i^{m/\alpha_{T_1}},
\label{eq:Blockage:SinglePrimary:PrimaryCov-sector-1}
\end{align} 
where $(a)$ is due to $ g_{\mathrm{pr}} \left( \theta_\seco \right)$ being $\apr$ and $ \bpr $ when $ \lvert \theta_\seco \rvert \leq \phi_{\mathrm{pr}}/2 $ and $ \lvert \theta_\seco \rvert > \phi_{\mathrm{pr}}/2 $, respectively with 
$ \omega_{\mathrm{s}} \in [0,\ 2\pi] $ and $ \gst{ \theta_\seco - \pi - \omega_\seco} $ being $\ast$ and $\bst$ with probabilities $\qst$ and $ \left( 1 - \qst \right)$, respectively.

\section{} \label{thrm:proof:Blockage:SecondaryCov}
From \eqref{eq:Blockage:SinglePrimary:SINR:SecondaryLink}, SINR coverage of the typical secondary link is 
\begin{align}
p_\mathrm{cs}(\SThres, \mu, \rho) &\overset{(a)}{=} \mathbb{P} \left[ F_{\seco0} >  \frac{ \tau \left(\sigma^2 + I_\prim +  I_\seco \left(\Phi_\seco \right)\right)}{ p_\seco \gst{0}\gsr{0} \ell_{T_{\seco\mathrm{0}}} (r_\seco) } 
\right]  p^{'}_{\mathrm{m0}} \nonumber \\
&\ 
\overset{(b)}{=} p^{'}_{\mathrm{m0}} \, \mathbb{E} \left[\exp \left( - \frac{\tau \left(\sigma^2 + I_\prim +  I_\seco \left(\Phi_\seco \right)\right)}{ p_\seco \gst{0}\gsr{0} \ell_{T_{\seco\mathrm{0}}} (r_\seco) } \right) \right] \nonumber \\ 
&\ 
\overset{(c)}{=} p^{'}_{\mathrm{m0}} 
\!
\sum\nolimits_{T = \{\los,\nlos\}} p_T (r_\seco) \,
e^{- s^{'}_T \sigma^2 } \mathcal{L}_{I_\prim} \left(s_T\right) 
\mathcal{L}_{I_\seco} \left( s_T \right), 
\label{eq:Blockage:SinglePrimary:SecondaryCov-1}
\end{align}
where $s^{'}_T = \tau r_\seco^{ \alpha_T}/C_T p_\seco \gst{0} \gsr{0}$. Here $\mathcal{L}_{ I_\prim}$ and $\mathcal{L}_{ I_\seco}$ are the LTs of $I_\prim$ and $I_\seco$, respectively. Here, $(a)$ is due to $U^{'}_0 \sim \mathrm{Bernoulli}(p^{'}_{\mathrm{m}0})$, $(b)$ is due to $F_{\seco 0} \sim \exp(1)$ and $(c)$ is due to $\ell_{T_{\seco\mathrm{0}}} \sim \mathrm{Bernoulli}(p_\los (r_\seco))$. For $T_1 = \{\los,\nlos\}$ and $\bar{\kappa}_\mathrm{00} = \gsr{\delta_{\prim\mathrm{0}}}  \gpt{\delta_{\prim\mathrm{0}} - \pi - \omega_{\prim\mathrm{0}}}$, LT of ${ I_\prim}$ is
\begin{align}
&\mathcal{L}_{ I_\prim} (s^{'}_T) \overset{(d)}{=} \mathbb{E} \left[\fracS{ 1 }{ (1 + s^{'}_T p_\prim \bar{\kappa}_\mathrm{00} \ell_{T_{\prim\mathrm{0} \to \seco\mathrm{0}}} (x_{\prim\mathrm{0}}) )}\right] \nonumber \\
&\overset{(e)}{=} \sum_{T_1} \frac{p_{T_1} (x_{\prim\mathrm{0}})}{1 + \tau \frac{C_{T_1}}{C_T} \frac{r_{\seco}^{ \alpha_T}}{x_{\prim\mathrm{0}}^{ \alpha_{T_1}}} \frac{p_\prim}{p_\seco} \frac{\gsr{\delta_{\prim\mathrm{0}}}  \gpt{\delta_{\prim\mathrm{0}} - \pi - \omega_{\prim\mathrm{0}}}}{\gst{0} \gsr{0}}}.
\label{eq:Blockage:SinglePrimary:SecondaryCov-3}
\end{align} 
where $(d)$ is due to $G^{'}_{\mathrm{00}} \sim \exp(1)$ and $(e)$ is due to $\ell_{T_{\prim\mathrm{0} \to \seco\mathrm{0}}} \sim \mathrm{Bernoulli}(p_\los (x_{\prim\mathrm{0}}))$. Now LT of ${ I_\seco}$ is
\begin{align}
\mathcal{L}_{ I_\seco} (s^{'}_T) &\overset{(f)}{=} \mathbb{E}_{\Phi_\seco} \left[ \prod\nolimits_{ \X_{\seco i} \in \Phi_{\seco} } 
\!\!\!
\expU {- s^{'}_T p_\seco \kappa_{i\mathrm{0}} U^{'}_{i} F^{'}_{i\mathrm{0}} \ell_{T_{\seco i \to \seco\mathrm{0}}} (x_{\seco i}) } \right] \nonumber \\
&\overset{(g)}{=} e^{- \lambda_\seco \int_0^{2\pi} \!\! \int_0^\infty \left(1 - \mathbb{E} \left[ e^{ - s^{'}_T p_\seco \kappa_\mathrm{0} U^{'} \!\! F^{'}_\mathrm{0} \ell_{T_{\seco  \to \seco\mathrm{0}}} (x_\seco) } \right] \right)  x_\seco \dd x_\seco \dd \theta_\seco},
\end{align}
where $\kappa_\mathrm{0} = \gsr{ \theta_\seco} \gst{ \theta_\seco - \pi - \omega_\seco}$. Here, $(f)$ is due to Silvnyak theorem \cite{andrews2023introduction} and $(g)$ is due  PGFL of $\Phi_\seco$. If $\chi_{T_2} = p_\seco \kappa_\mathrm{0} C_{T_2} x_\seco^{-\alpha_{T_2}}$ for $T_2 = \{\los,\nlos\}$, we get
\begin{align}
\mathcal{L}_{ I_\seco} (s^{'}_T) &\overset{(h)}{=} e^{\!\! - \lambda_\seco 
\!
\int_0^{2\pi} 
\!\!
\int_0^\infty 
\!
\sum\limits_{T_2} p_{T_2} (x_\seco) \left( 1 - \mathbb{E} \left[ e^{\!\! - s^{'}_T \chi_{T_2} U^{'} \!\! F^{'}_\mathrm{0}} \right] \right) x_\seco \dd x_\seco \dd \theta_\seco}
\nonumber \\
&\overset{(i)}{=} e^{\!\! - \lambda_\seco 
\! 
\int_0^{2\pi} 
\!\!
\int_0^\infty 
\!
\sum_{T_2} p_{T_2} (x_\seco) 
\,
\mathbb{E}_{\omega_\seco} \! \left[\fracS{p^{'}_{\mathrm{m}} }{(1 + {1}/{s^{'}_T \chi_{T_2}})}\right] x_\seco \dd x_\seco \dd \theta_\seco},
\label{eq:Blockage:SinglePrimary:SecondaryCov-7}
\end{align}
where $(h)$ is due to $\ell_{T_{\seco \to \seco\mathrm{0}}} \sim \mathrm{Bernoulli}(p_\los (x_\seco))$ and $(i)$ is due to $U^{'} \sim \mathrm{Bernoulli}(p^{'}_\mathrm{m})$ and $F^{'}_{\mathrm{0}} \sim \exp(1)$. Substituting \eqref{eq:Blockage:SinglePrimary:SecondaryCov-3} and \eqref{eq:Blockage:SinglePrimary:SecondaryCov-7} in \eqref{eq:Blockage:SinglePrimary:SecondaryCov-1} along with values of $s^{'}_T$, $\chi_{T_2}$, $p^{'}_{\mathrm{m}}$ and $p^{'}_{\mathrm{m0}}$  and using change of variables $A^T_\mathrm{0}$, $A^{T_1}_\mathrm{10}$, $A^{T_2}_\mathrm{20}$, $A^{T_4}_\mathrm{s0}$ and $C^{T_3}_\mathrm{s0}$ will give the desired result (see supplementary \cite{TripGupTheoremFile2025} for detailed proof).

\bibliographystyle{IEEEtran}
\bibliography{IEEEabrv,PapersNameBlockage}
\end{document}